\documentclass{lmcs} 
\pdfoutput=1
\usepackage[latin1]{inputenc}

\usepackage{lastpage}
\lmcsdoi{18}{3}{5}
\lmcsheading{}{\pageref{LastPage}}{}{}%
{Apr.~28,~2021}{Jul.~28,~2022}{}


\usepackage[english]{babel}
\usepackage{etex}

\usepackage{epsfig}
\usepackage{latexsym}
\usepackage{graphicx}
\usepackage{amssymb}
\usepackage{amsfonts}
\usepackage{amsthm}
\usepackage{enumitem}

\usepackage{amsmath}
\usepackage{verbatim}
\usepackage{url}
\usepackage{colortbl}
\usepackage{stmaryrd}
\usepackage{epic,eepic}
\usepackage{xspace}
\usepackage{tikz}
\usepackage[noend]{algpseudocode}
\usepackage{tcolorbox}
\usepackage{multirow}

\usetikzlibrary{automata}
\usetikzlibrary{arrows}
\usetikzlibrary{backgrounds}
\usetikzlibrary{positioning}
\usetikzlibrary{fit}
\usetikzlibrary{calc}
\usetikzlibrary{matrix}
\usetikzlibrary{patterns}

\usepackage{stackengine}
\usepackage{xspace}

\usepackage{xcolor}
\usepackage[linesnumbered,ruled,vlined]{algorithm2e}

\SetCommentSty{mycommfont}
\SetKwInput{KwData}{Input}

\definecolor{ao(english)}{rgb}{0.0, 0.5, 0.0}

\usepackage{microtype} 
\usepackage{ellipsis}         
\usepackage[abbreviations,british]{foreign} 

\usetikzlibrary{shadows}
\tikzset{
diagonal fill/.style 2 args={fill=#2, path picture={
\fill[#1, sharp corners] (path picture bounding box.south west) -|
                         (path picture bounding box.north east) -- cycle;}},
reversed diagonal fill/.style 2 args={fill=#2, path picture={
\fill[#1, sharp corners] (path picture bounding box.north west) |- 
                         (path picture bounding box.south east) -- cycle;}}
}

 \makeatletter
 \def\desclabel#1#2{\begingroup
    \def\@currentlabel{#1}%
    #1\label{#2}\endgroup
 }
 \makeatother

\theoremstyle{plain}

\newcommand{\set}[1]{\{ #1 \}}

\newcommand{\pair}[2]{(#1,#2)}
\newcommand{\triple}[3]{\tup{#1,#2,#3}}

\newcommand{\Nat}{\ensuremath{\mathbb{N}}}

\newcommand{\bottom}{\perp}

\newcommand{\card}[1]{{\sf card}(#1)}

\newcommand{\aset}{X}
\newcommand{\asetbis}{Y}

\newcommand{\avarprop}{p}
\newcommand{\avarpropbis}{q}
\newcommand{\avarpropter}{r}
\newcommand{\ahvarprop}{\mathsf{h}}
\newcommand{\avvarprop}{\mathsf{v}}

\newcommand{\aformula}{\phi} 
\newcommand{\aformulabis}{\psi} 

\newcommand {\length}[1] {\ensuremath{|#1|}}

\newcommand{\egdef}{\deff}
\newcommand{\equivdef}{\stackrel{\mbox{\begin{tiny}def\end{tiny}}}{\equivaut}} 
\newcommand{\equivaut}{\;\Leftrightarrow\;}

\newcommand{\amap}{\mathfrak{f}}

\newcommand {\pspace} {\textsc{PSpace}\xspace}

\newcommand {\nexpspace} {\textsc{NExpSpace}\xspace}
\newcommand {\expspace} {\textsc{ExpSpace}\xspace}

\newcommand {\nexptime} {\NExpTime}

\newcommand {\tower} {\textsc{Tower}\xspace}
\newcommand {\aexppol} {\textsc{AExp}$_{\textsc{pol}}$\xspace}

\newenvironment {lemma*} {\noindent {\bf Lemma} \em} {\rm}

\newcommand{\atranslation}{t}

\newcommand{\defstyle}[1]{{\emph{#1}}}

\newcommand{\cut}[1]{}
\newcommand{\interval}[2]{[#1,#2]}

\mathchardef\mhyphen="2D 
\def\first(#1){\mathtt{first}(#1)}
\def\last(#1){\mathtt{last}(#1)}
\makeatletter
\newcommand{\@transb}[2]{t_{#1}\!\left(#2\right)}
\newcommand{\@transnob}[1]{t\!\left(#1\right)}
\newcommand{\trans}{\@ifstar{\@transb}{\@transnob}}
\makeatother

 \newcounter{theoreme}[section]

\newcommand{\alogic}{\mathfrak{L}}

\newcommand{\tup}[1]{\langle #1 \rangle}

\newcommand{\PVAR}{\AP}

\algrenewcommand\algorithmicrequire{\textbf{In:}}
\algrenewcommand\algorithmicensure{\textbf{Out:}}

\algblockdefx{Cases}{EndCases}%
   [1]{\textbf{case} #1 \textbf{of}}%
   {\textbf{end case}}
\algcblockx[Cases]{Cases}{Case}{EndCases}%
   [1]{#1:\enspace}%
   {\textbf{end case}}

\definecolor{Gray}{gray}{0.8}

\newcommand{\CTL}{\ensuremath{\mathsf{CTL}}\xspace}
\newcommand{\LTL}{\ensuremath{\mathsf{LTL}}\xspace}
\newcommand{\ATL}{\ensuremath{\mathsf{ATL}}\xspace}
\newcommand{\CTLStar}{\ensuremath{\mathsf{CTL}^{*}}\xspace}
\newcommand{\anominal}{x}
\newcommand{\anominalbis}{y}

\newcommand{\anode}{v}
\newcommand{\anodebis}{u}
\newcommand{\atreemodel}{\str{T}}

\newcommand{\Iff}{\Leftrightarrow} 

\newcommand{\synnumber}[1]{\mathtt{nb}(#1)}
\newcommand{\semnumber}[2]{\mathfrak{nb}_{#1}(#2)}

\newcommand{\phitype}[1]{{\rm type}(#1)}

\newcommand{\phifirst}[1]{{\rm first}(#1)}
\newcommand{\philast}[1]{{\rm last}(#1)}
\newcommand{\phiunique}[1]{{\rm uniq}(#1)}
\newcommand{\phipopulate}[1]{{\rm compl}(#1)}

\newcommand{\eqk}[3]{\synnumber{#2}  =_{#1} \synnumber{#3}}
\newcommand{\succk}[3]{\synnumber{#3}  =_{#1} \synnumber{#2}+1}
\newcommand{\gk}[3]{\synnumber{#2}  <_{#1} \synnumber{#3}}

\newcommand{\anothereqk}[1]{\mathtt{nb} \ =_{#1} \ \tow(#1,n)}

\newcommand{\lsr}[2]{{\rm LSR}_{#2}(#1)}
\newcommand{\lsrone}[2]{{\rm LSR}^1_{#2}(#1)}
\newcommand{\lsrtwo}[2]{{\rm LSR}^2_{#2}(#1)}
\newcommand{\lsrthree}[2]{{\rm LSR}^3_{#2}(#1)}
\newcommand{\aroot}{\varepsilon}
\newcommand{\apath}{\pi}

\newcommand{\EX}{\mathbf{EX}}
\newcommand{\EXOne}{\mathbf{EX}_{=1}}
\newcommand{\AX}{\mathbf{AX}}
\newcommand{\E}{\mathbf{E}}
\newcommand{\A}{\mathbf{A}}
\newcommand{\U}{\mathbf{U}}
\newcommand{\F}{\mathbf{F}}
\newcommand{\X}{\mathbf{X}}
\newcommand{\EF}{\mathbf{EF}}

\newcommand{\AG}{\mathbf{AG}}
\newcommand{\AF}{\mathbf{AF}}

\newcommand{\at}{@}

\newcommand{\att}[2]{\at_{#1}^{#2}}
\newcommand{\bindd}[2]{\mathtt{nom}(#1,#2)}
\newcommand{\distinctbindd}[2]{\mathtt{diff{\mbox{-}}nom}(#1,#2)}

\newcommand{\tow}{\mathfrak{t}}

\newcommand{\val}[1]{\mathit{val}}
\newcommand{\layer}[1]{\mathit{layer}_{#1}}

\newcommand{\kripkeK}{\mathcal{K}}

\newcommand{\K}{\ensuremath{\mathsf{K}}\xspace}

\newcommand{\KD}{\ensuremath{\mathsf{KD}}\xspace}
\newcommand{\Sfour}{\ensuremath{\mathsf{S4}}\xspace}

\newcommand{\Kfour}{\ensuremath{\mathsf{K4}}\xspace}
\newcommand{\Sfive}{\ensuremath{\mathsf{S5}}\xspace}
\newcommand{\GL}{\ensuremath{\mathsf{GL}}\xspace}

\newcommand{\QKt}{\ensuremath{\mathsf{QK}^t}\xspace}

\newcommand{\QKDt}{\ensuremath{\mathsf{QKD}^t}\xspace}
\newcommand{\QSfourt}{\ensuremath{\mathsf{QS4}^t}\xspace}

\newcommand{\QKfourt}{\ensuremath{\mathsf{QK4}^t}\xspace}

\newcommand{\QGLt}{\ensuremath{\mathsf{QGL}^t}\xspace}

\newcommand{\QCTL}{\ensuremath{\mathsf{QCTL}}\xspace}
\newcommand{\QCTLEX}[1]{\ensuremath{\mathsf{QCTL}_{\X}}\xspace}
\newcommand{\QCTLEF}[1]{\ensuremath{\mathsf{QCTL}_{\F}}\xspace}
\newcommand{\QCTLEFplus}[1]{\ensuremath{\mathsf{QCTL}_{\X\F}}\xspace}
\newcommand{\sQCTL}[1]{\ensuremath{\mathsf{QCTL}^s_{#1}}\xspace}
\newcommand{\QLTL}{\ensuremath{\mathsf{QLTL}}\xspace}

\newcommand{\tQCTL}[1]{\ensuremath{\mathsf{QCTL}^t_{#1}}\xspace}
\newcommand{\gtQCTL}[1]{\ensuremath{\mathsf{QCTL}^{gt}_{#1}}\xspace}
\newcommand{\tQCTLEX}[1]{\ensuremath{\QCTL^t_{\X #1}}}

\newcommand{\tQCTLEF}[1]{\ensuremath{\QCTL^t_{\F #1}}\xspace}
\newcommand{\tQCTLEFplus}[1]{\ensuremath{\QCTL^t_{\X\F #1}}\xspace}
\newcommand{\gtQCTLEFplus}[1]{\ensuremath{\QCTL^{gt}_{\X \F}}\xspace}

\newcommand{\ftQCTL}[1]{\ensuremath{\QCTL^{ft}_{#1}}}
\newcommand{\ftQCTLEX}[1]{\ensuremath{\QCTL^{ft}_{\X}}}
\newcommand{\ftQCTLEFplus}[1]{\ensuremath{\QCTL^{ft}_{\X \F}}}

\newcommand{\ExpSpace}{\textsc{ExpSpace}}

\newcommand{\NExpTime}{\textsc{NExpTime}}

\newcommand{\kNExpTime}{k\textsc{-NExpTime}}

\newcommand{\TOWER}{\textsc{Tower}}%

\newcommand{\BigOh}[1]{\mathcal{O}(#1)}
\newcommand{\str}[1]{{\mathfrak{#1}}}

\newcommand{\AP}{\mathit{AP}} 

\newcommand{\md}[1]{\mathit{md}(#1)}

\newcommand{\deff}{\stackrel{\text{def}}{=}}

\let\cH\undefined

\let\cT\undefined
\let\cV\undefined
\let\cP\undefined

\newcommand{\cH}{\mathcal{H}}

\newcommand{\cT}{\mathcal{T}}
\newcommand{\cTzero}{\cT_0}
\newcommand{\cTacc}{\cT_{\textit{acc}}}
\newcommand{\cTmulti}{\cT_{\textit{multi}}}
\newcommand{\cPmulti}{\mathtt{AMTP}}

\newcommand{\cV}{\mathcal{V}}
\newcommand{\cP}{\mathcal{P}}

\newcommand{\restr}{\!\!\restriction\!\!}
\newcommand{\arestrictedtreemodel}[1]{\atreemodel \restr_{#1}}
\newcommand{\atiling}{\tau}
\newcommand{\LongVersionOnly}[1]{}
\newcommand{\satproblem}[1]{{\rm SAT(}#1{\rm)}}

\renewcommand{\pair}[2]{(#1,#2)}

\colorlet{jaune}{yellow!80!green}

\colorlet{vert}{green!45!black}

\colorlet{bleu}{blue!70!black}

\colorlet{rouge}{red!80!black}

\tikzstyle{minirond}=[draw,circle,minimum height=2mm,inner sep=0pt]
\tikzstyle{ptrond}=[draw,circle,minimum height=2mm]
\tikzstyle{medrond}=[draw,circle,minimum height=5mm]
\tikzstyle{rond}=[draw,circle,minimum height=7mm]
\tikzstyle{carre}=[draw,minimum width=6mm,minimum height=6mm]
\tikzstyle{medcarre}=[draw,minimum width=4mm,minimum height=4mm]
\tikzstyle{ptcarre}=[draw,minimum width=2.8mm,minimum height=2.8mm]
\tikzstyle{minicarre}=[draw,minimum width=1.5mm,minimum height=1.5mm,inner sep=0pt]
\tikzstyle{rouge}=[draw=red,fill=red!20!white]
\tikzstyle{vert}=[draw=green!80!black,fill=green!80!black!20!white]
\tikzstyle{jaune}=[draw=yellow!60!red,fill=yellow!60!red!30!white]
\tikzstyle{bleu}=[draw=blue,fill=blue!40!white]
\tikzstyle{gris}=[draw=black!80!white,fill=black!40!white]
\tikzstyle{rougef}=[draw=red,fill=red!60!white]
\tikzstyle{vertf}=[draw=green!80!black,fill=green!80!black!60!white]
\tikzstyle{jaunef}=[draw=yellow!80!black,fill=yellow!80!black!60!white]
\tikzstyle{bleuf}=[draw=blue,fill=blue!70!white]
\tikzstyle{rjaune}=[style=rond,style=jaune]
\tikzstyle{rbleu}=[style=rond,style=bleu]
\tikzstyle{rvert}=[style=rond,style=vert]
\tikzstyle{rrouge}=[style=rond,style=rouge]
\tikzstyle{rgris}=[style=rond,style=gris]
\tikzstyle{cjaune}=[style=carre,style=jaune]
\tikzstyle{cbleu}=[style=carre,style=bleu]
\tikzstyle{cvert}=[style=carre,style=vert]
\tikzstyle{crouge}=[style=carre,style=rouge]
\tikzstyle{cgris}=[style=carre,style=gris]
\tikzstyle{rjaunef}=[style=rond,style=jaunef]
\tikzstyle{rbleuf}=[style=rond,style=bleuf]
\tikzstyle{rvertf}=[style=rond,style=vertf]
\tikzstyle{rrougef}=[style=rond,style=rougef]

\def\myref#1{\footnotetext{\leavevmode\hskip-12mm\fontsize{7pt}{8pt}\selectfont
  \color{red!80!blue!50!black}\hbox to\linewidth{#1\hfill}}}

\begin{document}
\title[Why Does Prop.~Quantification Make Logics on Trees Robustly Hard?]{Why Does
Propositional Quantification Make Modal and Temporal Logics on Trees Robustly Hard?}

\author[B.~Bednarczyk]{bartosz Bednarczyk\lmcsorcid{0000-0002-8267-7554}}[a]  
\address{Computational Logic Group, TU Dresden \& Institute of Computer Science, University of Wroc\l{}aw} 
\email{bartosz.bednarczyk@cs.uni.wroc.pl}  

\author[S.~Demri]{St\'ephane Demri\lmcsorcid{0000-0002-3493-2610}}[b]    
\address{Universit{\'e} Paris-Saclay, ENS Paris-Saclay, CNRS, LMF, 91190, Gif-sur-Yvette, France} 
\email{demri@lsv.fr}  

\begin{abstract}
Adding propositional quantification to the modal logics \K, $\mathsf{T}$ or \Sfour is known to lead to undecidability  
but~\CTL with propositional quantification under the tree semantics (\tQCTL{}) admits a non-elementary 
\tower-complete satisfiability problem. 
We investigate the complexity of strict fragments of \tQCTL{} as well as of the modal logic \K with propositional 
quantification under the tree semantics. 
More specifically, we show that \tQCTL{} restricted to the temporal operator~$\EX$ is already \tower-hard, 
which is unexpected as $\EX$ can only enforce local properties. 
When \tQCTL{} restricted to $\EX$ is interpreted on $N$-bounded trees for 
some~$N \geq 2$, we prove that the satisfiability problem is \aexppol-complete; 
\aexppol-hardness is established by reduction from a recently introduced
tiling problem, instrumental for studying the model-checking problem for interval temporal logics. 
As consequences of our proof method, 
we prove \tower-hardness of \tQCTL{} restricted to  $\EF$ or to $\EX\EF$ and  
of the well-known modal logics such as \K, \KD, \GL, 
\Kfour and \Sfour with propositional quantification under a semantics based on classes of trees. 
\end{abstract}

\maketitle

\section{Introduction}
\label{section-introduction}


\paragraph{Propositional quantification in modal and temporal logics}

 A natural way to design logics that dynamically
update their models, consists in adding propositional quantification as done for example
to define QBF from SAT. 
Propositional quantification is a very powerful feature to update models but this may have consequences
in terms of computability.
In the realm of modal logics~\cite{blackburn_rijke_venema_2001}, the paper~\cite{Bull69} remains a quite early work adding 
propositional quantification. 
The undecidability of the propositional modal logic 
\K (resp. $\mathsf{T}$, \Kfour and \Sfour) augmented with propositional
quantification is first established in~\cite{Fine70}, and this is done thanks to a reduction
from  the second-order arithmetic. By contrast, the decidability of second-order versions of the
modal logic \Sfive is first proved in~\cite[Chapter 3]{Fine69} (see also~\cite{Fine70,Kaplan70}) 
but two \Sfive modalities and propositional quantification already lead to 
undecidability~\cite{Antonelli&Thomason02,Kaminski&Tiomkin96}.

Many subsequent works have dealt with second-order modal logics, 
see \eg~\cite{Kaminski&Tiomkin96,Kremer97,French01,tenCate06}, but in the realm of 
temporal logics, \LTL with propositional quantification (written \QLTL) is 
introduced in Sistla's PhD thesis~\cite{Sistla83} 
(see also~\cite{Sistla&Vardi&Wolper87}) and non-elementarity of the satisfiability
problem is a consequence of~\cite{Meyer73}. 
So, the \QLTL satisfiability problem is decidable but with high complexity. 
By contrast, \CTL with propositional 
quantification (written \QCTL) already admits an undecidable satisfiability problem
by~\cite{Fine70} (as \CTL captures the modal logic \K) 
but its variant under the tree semantics (written \tQCTL{})  admits a non-elementary \tower-complete satisfiability 
problem~\cite{LaroussinieM14,DavidLM16} (the complexity class \tower is introduced in~\cite{Schmitz16}). 
Having a tree semantics  means that the formulae of 
\tQCTL{} are interpreted on computation trees obtained from the unfolding of 
finite (total) Kripke structures, which allows us to regain decidability (see a 
similar approach in~\cite{Zach04} with a quantified version of the modal logic \Sfour). 
This is a major observation from~\cite{LaroussinieM14}, 
partly motivated by the design of decision procedures for \ATL with strategy 
contexts~\cite{Laroussinie&Markey15}. The tree semantics, as far as 
satisfiability is concerned, amounts to considering Kripke structures
that are finite-branching trees in which all the maximal branches are infinite.
This is an elegant way to regain decidability. 
More generally, decidability in the presence of propositional quantification can be regained when
tree-like models are involved, see \eg~\cite{Artemov&Beklemishev93,Baaz&Ciabattoni&Zach00,Zach04,LaroussinieM14}, 
essentially by taking advantage of Rabin's Theorem~\cite{Rabin69}. 

The modal logic \K with propositional quantification from~\cite{Fine70} is 
interpreted under the structure semantics, as classified in~\cite{LaroussinieM14}, 
but many variants of propositional quantification exist in the literature 
(see \eg~\cite{Patthaketal02,French03,Bozzelli&vanDitmarsch&Pinchinat15} 
and~\cite{Belardinelli&vanderHoek16,Belardinelli&vanderHoek&Kuijer18}
in the context of epistemic reasoning). Sometimes, propositional quantification
is syntactically restricted  in the temporal language, but its inclusion is 
motivated by a gain of expressive power while preserving the decidability of the
reasoning tasks. By way of example, in~\cite{Pinchinat&Riedweg03}, an extension
of the modal $\mu$-calculus with partial propositional quantification is 
introduced to perform control synthesis, whereas an extension for model-checking 
computer systems is also presented in~\cite{Kupferman99}. 

Interestingly enough, propositional quantification can sometimes have a more hidden
presence. For instance, hybrid  logics with the down-arrow operator $\downarrow_{\anominal}$, 
see e.g~\cite{Areces&Blackburn&Marx00,Areces&Blackburn&Marx01,Marx02,Weber07,Kernberger19}, 
can be understood as a form of propositional quantification
since $\downarrow_{\anominal} \aformula$ enforces that the propositional variable 
$\anominal$ holds true only at the current world (before evaluating the formula $\aformula$). 
In such logics, the companion formula $@_{\anominal} \aformulabis$ 
expresses that 
 the unique world satisfying $\anominal$ also satisfied $\aformulabis$; $@_{\anominal}$ is a powerful 
operator to navigate in the structure~\cite{Marx02} (obviously, it is related 
to the universal modality, see \eg~\cite{Goranko&Passy92}).


\paragraph{Our motivations} 

As recalled above, the modal logic \K augmented with propositional quantification 
is undecidable~\cite{Fine70} and a fortiori, undecidability holds for fragments 
of \CTL with propositional quantification. Actually, these results hold under 
the structure semantics but \tQCTL{} (tree semantics) admits \tower-complete 
satisfiability and model-checking problems~\cite{LaroussinieM14}. As \tQCTL{} 
can express that every tree node has exactly one child, \tower-hardness for the 
satisfiability problem for \tQCTL{}, is a corollary of the \tower-hardness of 
the satisfiability problem for \QLTL, 
see \eg~\cite{Meyer73,Sistla&Vardi&Wolper87,Demri&Goranko&Lange16}.

Given the central position of the modal logic \K, surprisingly, the complexity 
of the satisfiability problem for \K with propositional quantification under the 
tree semantics has never been investigated 
 (closely related to \tQCTLEX{} -- \ie 
 \tQCTL{} restricted to $\EX$--  as $\EX$ corresponds to $\Diamond$ in \K but with total models).
\tQCTLEX{} is a 
natural and modest fragment of~\tQCTL{} and we aim at 
characterising its 
complexity. Furthermore, the model-checking problem for 
\tQCTL{} is \tower-hard even with input Kripke structures having at most two 
successor worlds per world and \tower-hardness of \QLTL holds with linear 
structures of length~$\omega$, see \eg~\cite{Sistla&Vardi&Wolper87}. 
Thus, it is worth understanding what happens with the satisfiability problem 
for \tQCTLEX{} when the tree models are 
$N$-bounded 
(\ie each node has at most $N$ children) for some fixed $N \geq 2$.


\paragraph{Our contributions} 

Given $\tQCTL{}$, the extension of \CTL with propositional quantification under 
the tree semantics  (\ie the models are  finite-branching trees where all the 
maximal branches are infinite), let $\tQCTL{\leq N}$ be its variant in which the 
models are $N$-bounded, for some~$N \geq 2$. We~write~$\tQCTLEX{}$ and 
$\tQCTLEX{,\leq N}$, to denote respectively the restriction to the operator $\EX$, 
and $\tQCTLEF{}$ to denote the restriction of $\tQCTL{}$ to $\EF$.

\begin{itemize}
\item 
We first present to the reader the toolkit of local nominals and explain the basic ideas behind the hardness results of this paper by proving, for all $N \geq 2$, that the satisfiability problem for $\tQCTLEX{,\leq N}$ is \aexppol-complete (Section~\ref{section-bounded-degree}).
\aexppol is the class of problems solvable by 
exponential-time alternating Turing machines with a polynomial number of alternations.
By using a small model property and
the complexity of model-checking (with upper bound \aexppol), \aexppol-easiness is established. 
As far as \aexppol-hardness is concerned, 
the alternating multi-tiling problem introduced in~\cite{Bozzellietal17}
is instrumental to establish that the model-checking
problem for the interval temporal logic
$B \bar{E}$ with regular expressions is \aexppol-hard.
As a corollary, we get that the modal logic \K with propositional quantification
interpreted on finite trees of branching degree bounded by some fixed $N \geq 2$ is also
\aexppol-complete.

\item More generally and despite the modest scope of~$\EX$, the satisfiability problem 
for $\tQCTLEX{}$ is shown to be \tower-hard (Theorem~\ref{theorem:tQCTLEX-tower}), by 
a uniform reduction from $\kNExpTime$-complete tiling problems (uniformity is 
with respect to $k$). The corresponding upper bound is known 
from~\cite{LaroussinieM14} and it is worth noting that all the \tower upper bound 
results presented in this paper are based on translations into the satisfiability 
problem for \tQCTL, sometimes via intermediate decision problems, which eventually 
uses Rabin's Theorem~\cite{Rabin69} in some essential way.
The hardness proof is one of the main results of the paper and amounts first to 
showing that one can enforce that a node has a number of children equal to some 
tower of exponentials of height $k$ with a formula of size exponential in $k$. 
By contrast, checking the satisfiability status of $\CTLStar$ formulae,
requires only to consider tree models with branching degree bounded by the size 
of the formula, see \eg~\cite{Emerson&Sistla84,Demri&Goranko&Lange16}.
Once this complex construction enforcing a  very high number of children is 
performed, the reduction from the tiling problems can be done with the help of 
other properties on the number of children. Hence, even though \QCTLEX{} under 
the structure semantics is undecidable~\cite{Fine70} and the variant of \tQCTLEX{} 
under the tree semantics ($\tQCTL{}$) is decidable by~\cite{LaroussinieM14},
the problem admits a high complexity despite the local range of $\EX$.
The \tower lower bound for \tQCTL{} crucially depends on the availability of very wide trees,
which contrasts with the \aexppol upper bound for $\tQCTLEX{,\leq N}$ for which trees are $N$-bounded.

\item By adapting our proof method, we show that \tQCTLEF{} and \tQCTLEFplus{} 
(a variant of \tQCTLEF{} with the unique operator~$\EX\EF$) are \tower-hard too
(consult Section~\ref{section-qctlef}).\\

\item As $\EX$, $\EF$ and $\EX\EF$ correspond to the modality~$\Diamond$ in several 
modal logics (\eg $\EX\EF$ corresponds to~$\Diamond$ in transitive frames), we 
are able to establish \tower-completeness for standard modal logics with 
propositional quantification when interpreted on tree-like structures (see 
Section~\ref{section-modal-logics}). For instance, as the provability logic~\GL 
(after G\"odel and L\"ob) is complete for the class of finite transitive trees, 
\ie the accessibility relation is the transitive closure of the child-relation 
in the tree, see \eg~\cite{blackburn_rijke_venema_2001}, we also investigate
the satisfiability problem under the finite tree semantics. 
We show that \ftQCTLEX{} and \ftQCTLEFplus{} (`$ft$' stands for 'finite tree 
semantics') are \tower-complete too. The satisfiability problem for~\K 
(resp. \GL) with propositional quantification under the finite (resp. transitive) 
tree semantics is shown to be \tower-complete. 
Similar results are shown for \KD, \Kfour and \Sfour with propositional quantification
but interpreted on appropriate tree-like Kripke structures. 
See Section~\ref{section-modal-logics} for the details. 
\end{itemize}

\noindent
This work is a revised and complete version of our conference paper~\cite{Bednarczyk&Demri19}. 
To keep the main body of the paper short we delegate all the proofs to the technical appendix.


\section{Preliminaries}
\label{section-preliminaries}
\subsection{Kripke structures and computation trees}
Below, we recall standard definitions about Kripke structures.
Let $\AP = \set{\avarprop, \avarpropbis, \anominal, \anominalbis, \ldots}$ 
be a countably infinite set of propositional variables.
A \defstyle{Kripke structure} $\kripkeK$ is a triple $\triple{W}{R}{l}$, 
where $W$ is a set of \defstyle{worlds}, $R \subseteq W \times W$ is a  transition relation 
and $l : W \rightarrow 2^{\AP}$ is a labelling function.
A Kripke structure $\kripkeK$ is \defstyle{total} whenever 
for all $w \in W$, there is $w' \in W$ such that~$\pair{w}{w'} \in R$. 
Totality is a standard property for defining classes of models for temporal logics
such as \CTL. In the sequel, by a `Kripke structure' we mean a structure according 
to the above definition, otherwise when arbitrary (or total, finite, etc.)
Kripke structures need to be considered, we  explicitly specify which classes of
structures we have in mind. For instance, a tree model (resp. finite tree model) $\atreemodel =
\triple{V}{E}{l}$  is an arbitrary Kripke structure where $\pair{V}{E}$ is a (resp. finite) rooted tree.
Standard definitions about trees are omitted herein. Such structures play an important role in the paper,
as the tree semantics involves specific tree models (those obtained as computation trees of 
Kripke structures). 

Given a Kripke structure $\kripkeK = \triple{W}{R}{l}$ and a world $w \in W$, a \defstyle{finite path}
$\apath$ from $w$ is a finite sequence $w_0, \ldots, w_n$ such that 
$w_0 = w$ and  for all
$i \in \interval{0}{n-1}$, we have $\pair{w_i}{w_{i+1}} \in R$. 
An~\defstyle{infinite path} from $w$ is an infinite sequence $w_0, \ldots, w_n, \ldots$
such that $w_0 = w$ and  for all $i \geq 0$, we have  $\pair{w_i}{w_{i+1}} \in R$. 
With~$\Pi_{\kripkeK, w}$ we denote the set of all finite paths starting from a world 
$w \in W$ on a Kripke structure $\kripkeK$. 

For a Kripke structure $\kripkeK = \triple{W}{R}{l}$ and~$w \in W$, the 
\defstyle{computation tree unfolding $\kripkeK$ from $w$}
 is the tree model $\atreemodel_{\kripkeK,w} = \triple{V}{E}{l'}$ such that the following conditions are satisfied:
\begin{enumerate}[label=(\alph*)]
\item $V \egdef \Pi_{\kripkeK, w}$,
\item $\apath E \apath'$ $\equivdef$  $\apath = w_0, \ldots, w_n$, $\apath' = \apath, w_{n+1}$ for some $w_{n+1} \in W$ and $\pair{w_n}{w_{n+1}} \in R$,
\item for all $\apath = w_0, \ldots, w_n \in V$, we have $l'(\apath) \egdef l(w_n)$. 
\end{enumerate}
Thus, when $\kripkeK$ is finite and total, $\atreemodel_{\kripkeK,w}$ 
is a finite-branching tree model in which all the maximal branches
are infinite and it is also a total Kripke structure. Below,
unless otherwise stated, the tree semantics involves such finite-branching trees with infinite
maximal branches.
In the sequel, tree-like Kripke structures are denoted by $\atreemodel = \triple{V}{E}{l}$
in order to emphasize the tree-like nature of such models. 
 
\subsection{The logics  \texorpdfstring{\sQCTL{}}{sQCTL}, \texorpdfstring{\tQCTL{}}{tQCTL},
\texorpdfstring{\ftQCTL{}}{ftQCTL} and \texorpdfstring{\gtQCTL{}}{gtQCTL}}

In this section, we define the logics \texorpdfstring{\sQCTL{}}{sQCTL}, \texorpdfstring{\tQCTL{}}{tQCTL},
\texorpdfstring{\ftQCTL{}}{ftQCTL} and \texorpdfstring{\gtQCTL{}}{gtQCTL} 
whose formulae are interpreted over different classes of Kripke structures. 
The formulae of \QCTL-like logics are defined from the grammar below 
by extending the set of formulae from the computation tree logic 
\CTL~\cite{ClarkeE81} with propositional quantification:
\[
\aformula ::= 
\avarprop \mid  \neg \aformula \mid \aformula \wedge \aformula \mid
\EX \aformula \mid \E (\aformula \U \aformula) \mid \A (\aformula \U \aformula)
\mid \exists{\avarprop} \  \aformula,
\]
where $\avarprop \in \AP$. We use the standard abbreviations 
$\vee, \rightarrow, \leftrightarrow, \bottom, \top$, as well as other  
operators like $\AX$, $\EF$, $\AG$ and $\AF$:
\begin{itemize}
\item  $\AX(\aformula) \egdef \neg \EX(\neg \aformula)$,
\item $\EF(\aformula) \egdef \E (\top \U \aformula)$,
$\AG(\aformula) \egdef \neg \EF(\neg \aformula)$
and~$\AF(\aformula) \egdef \A (\top \U \aformula)$.
\item The universal quantifier~$\forall{\avarprop} \  \aformula$ is used as 
$\neg \exists{\avarprop} \  \neg \aformula$.
\end{itemize}
%
We denote by $\length{\aformula}$ the \defstyle{length} of the formula~$\aformula$ 
measured in a standard way, \ie as the number of symbols used to write 
$\aformula$. The \defstyle{modal/temporal depth} of a formula $\aformula$, 
written $\md{\aformula}$, is the maximal number of
nested temporal operators in $\aformula$. 
We stress that  $\md{\aformula}$ is linear in $\length{\aformula}$.

By restricting the set of allowed temporal operators in \QCTL  to the only temporal operator~$\EX$ (resp. $\EF$) we obtain the logic \QCTLEX{} (resp.~\QCTLEF{}). 
Note that $\AF$ does not occur in \QCTLEF{} but~$\AX$ (resp. $\AG$) is allowed in \QCTLEX{} (resp. in \QCTLEF{}) as it is the dual operator of $\EX$ (resp.~$\EF$). 
The main object of study in the paper is the logic 
\QCTLEX{} under the tree semantics (below, written~\tQCTLEX{}). Moreover, we write \QCTLEFplus{}
to denote the restriction of \QCTL to the (combined) temporal operator~$\EX \EF$, which provides a strict version of the future-time temporal operator $\EF$.
In the rest of the paper, we refer to \QCTL (or to some of its fragments) to denote
a set of formulae and the notation~$\QCTL^{\bullet}$ with a superscript `$\bullet$' 
refers to a logic based on \QCTL (or on some of its fragments) under a specific semantics and for
which the symbol `$\bullet$' is just a reminder of the semantics.

To define formally the semantics for propositional quantifiers,
we introduce the notion of $\aset$-equivalence between two Kripke structures, with $\aset$ being a set of propositional variables. 
Intuitively, two Kripke structures are $\aset$-equivalent, whenever the only differences (if any) between them are  restricted
to the interpretation of the propositional variables not in $\aset$. 
Formally, we say that two Kripke structures $\kripkeK = \triple{W}{R}{l}$ and 
$\kripkeK' = \triple{W'}{R'}{l'}$ are \defstyle{$\aset$-equivalent} 
(written $\kripkeK \approx_{\aset} \kripkeK'$), 
iff $W = W'$, $R = R'$ and $l(w) \cap \aset = l'(w) \cap \aset$ for all $w \in W$.  

Given a Kripke structure $\kripkeK = \triple{W}{R}{l}$, 
a world~$w \in W$ and a formula
$\aformula$ in \QCTL, the satisfaction relation
$\kripkeK, w \models \aformula$  is defined as follows:
\[
\begin{array}{lcl}
\kripkeK, w \models \avarprop &\text{iff}&  \avarprop \in l(w)\\
\kripkeK, w \models \neg \aformula  &\text{iff}& \kripkeK, w \not\models \aformula\\
\kripkeK, w  \models \aformula \wedge \aformulabis &\text{iff}&  \kripkeK, w \models \aformula \; \text{and} \; 
\kripkeK, w \models_s \aformulabis \\
\kripkeK, w \models \EX \aformula &\text{iff}&  {\rm there \ is} \  w' \;
\text{such that} \;\pair{w}{w'} \in R \; \text{and} \; \kripkeK, w' \models \aformula \\
\kripkeK, w \models \E (\aformula \U \aformulabis) &\text{iff}& {\rm there \ is \ a \ path}  \  w_0, \ldots, w_n \;
\text{such that} \; w_0 = w, \\
&& \ \kripkeK, w_n \models \aformulabis \; \text{and} \; {\rm for \ all} \ i \in \interval{0}{n-1},  \ {\rm we \ have} \   \kripkeK, w_i \models \aformula \\
\kripkeK, w \models \A (\aformula \U \aformulabis) &\text{iff}& 
{\rm for \ all \ infinite \ paths} \ w_0, \ldots, w_n, \ldots \; \text{such that} \; w_0 = w, {\rm there \ is} \\
&&  j \geq 0 \; \text{such that} \; \kripkeK, w_j \models \aformulabis \;  \text{and} \ {\rm for \ all} \ i \in \interval{0}{j-1},  \kripkeK, w_i \models \aformula \\
\kripkeK, w \models \exists{\avarprop} \; \aformula &\text{iff}& {\rm there \ is} \  \kripkeK'
\; {\rm such \ that}  \; \kripkeK \approx_{\AP \setminus \{\avarprop\}} \kripkeK' \; 
{\rm and} \; \kripkeK', w \models \aformula 
\end{array}
\]

Note that $\kripkeK, w \models \EF \aformula$ iff there is $w' \in R^*(w)$  (where $R^*$ is the reflexive transitive closure of $R$) 
such that $\kripkeK, w' \models \aformula$. Stating that there is a unique successor of $w$ satisfying 
the propositional variable $\anominal$ can be expressed by the formula 
$\EX \ \anominal \wedge \neg (\exists \ \avarprop \
\EX(\anominal \wedge \avarprop) \wedge \EX(\anominal \wedge \neg \avarprop))$, where $\avarprop$ is 
distinct from $\anominal$.

The satisfiability problem for the  logic \sQCTL{} (under the \defstyle{structure semantics}) takes as input a 
formula $\aformula$ in \QCTL and asks whether
there is a finite and total Kripke structure $\kripkeK$ and a world~$w$ such that $\kripkeK, w \models \aformula$. 

The \defstyle{tree semantics} introduced in~\cite{LaroussinieM14} is obtained by considering as only admissible models the computation trees of 
finite and total Kripke structures. 
As noted in~\cite[Remark 5.7]{LaroussinieM14}, an equivalent formulation can be provided: 
the satisfiability problem for the logic \tQCTL{} (under the \defstyle{tree semantics}) 
takes as input a formula in \QCTL and asks whether
there is a finite-branching tree model $\atreemodel$ in which all the maximal branches are infinite 
such that $\atreemodel, \aroot \models \aformula$
and $\aroot$ is the root of $\atreemodel$. 
This is the definition we adopt along the paper.

We write \satproblem{$\alogic$} to denote the satisfiability problem for the logic $\alogic$. 
The distinction between the tree semantics and the structure semantics is crucial 
and affects the computational properties of the satisfiability problems. 

\begin{prop}
(I) 
\satproblem{$\sQCTL{}$} (under the structure semantics) 
is undecidable~\cite[Theorem 5.1]{LaroussinieM14}.
(II) 
\satproblem{$\tQCTL{}$} (under the tree semantics)
is decidable and \tower-complete~\cite[Theorem 5.3]{LaroussinieM14}.
(III) The satisfiability problem for \tQCTL{} is \tower-hard even if restricted to $\omega$-sequences~\cite{Meyer73,Sistla&Vardi&Wolper87}.  
\end{prop}

All our forthcoming \tower upper bound results are based on translations into the satisfiability problem for
\tQCTL{}, sometimes via intermediate decision problems, which eventually invokes Rabin's Theorem~\cite{Rabin69} in some essential way.
This is not  surprising, as considering tree-like models and propositional quantification naturally leads to
invoking the decidability of S$\omega$S~\cite{Rabin69} or its linear version,  the second-order logic of one successor S1S~\cite{Buchi60}.

Let us recapitulate the different versions of quantified \CTL we have seen so far.
\begin{itemize}
\item \sQCTL{} is interpreted over finite and total Kripke structures and 
      \satproblem{\sQCTL{}} is undecidable, see \eg~\cite{LaroussinieM14}.
\item \tQCTL{} is interpreted over finite-branching trees in which all the maximal branches
are infinite and \satproblem{\tQCTL{}} is \tower-complete, see \eg~\cite{LaroussinieM14}.
\end{itemize}
Let us introduce two additional versions that are closely related to modal logics with propositional
quantification on tree-like models introduced in the forthcoming Section~\ref{section-collection}. 
\begin{itemize}
\item \ftQCTL{} is interpreted over finite trees, \satproblem{\ftQCTL{}} can be shown to be in \tower by a
logspace reduction into \satproblem{\tQCTL{}} and its restriction to $\EX$ will be shown to admit
a \tower-hard satisfiability problem (see Section~\ref{section-collection}). 
\item \gtQCTL{} is interpreted over finite-branching trees (maximal branches may be finite),
\satproblem{\gtQCTL{}}  can be shown to be in \tower by logspace reduction into \satproblem{\tQCTL{}}
and its restriction to $\EX \EF$ will be shown to admit
a \tower-hard satisfiability problem (see Section~\ref{section-collection}). 
\end{itemize}

As a side remark, the equivalence between the two formulations of \tQCTL{}, 
is mainly due to the fact that \tQCTL{} can be translated into Monadic Second-Order Logic (MSO) 
over tree models with arbitrary finite branching (getting decidability by Rabin's Theorem~\cite{Rabin69}).
Indeed, as MSO over tree models with arbitrary finite branching is 
decidable by  Rabin's Theorem~\cite{Rabin69}, the satisfiability problem for \tQCTL{} is decidable
too, by a standard translation internalising the tree semantics for \tQCTL{}. 
As explained in~\cite[Section 6.3]{Thomas97}, the existence of a tree model implies also the existence of
a {\em regular tree model}, that can be precisely originated by a finite Kripke structure.
Hence, \tQCTL{} has the finite model property. So, not only can every finite and total Kripke structure be unfolded
as a finite-branching tree in which all the maximal branches are infinite (a direct consequence of the
definition for computation trees) but the existence of a tree model with the above-mentioned properties
and  satisfying  $\aformula$ 
implies the existence of the computation tree of a finite and total Kripke structure satisfying  $\aformula$. Hence, satisfiability in
the computation tree of  a finite Kripke structure  is equivalent to satisfiability in  a finite-branching tree 
in which all the maximal branches are infinite, and in the sequel, we shall operate with the latter definition. 

Apart from \CTL-like logics (see \eg~\cite{LaroussinieM14} for a wealth of bibliographical references),
other logics with propositional quantification have been shown to be decidable 
by translation into S$\omega$S, see e.g.~\cite{Artemov&Beklemishev93,Baaz&Ciabattoni&Zach00,Zach04}. 
Besides, in~\cite{Artemov&Beklemishev93}  a fragment of \GL with propositional quantification
 is shown to be decidable by translation into 
the weak monadic second-order logic of one  successor WS1S~\cite{Buchi60}, and 
a version of G\"odel logic $\mathsf{LC}$ with propositional quantification is shown to be decidable by
translation into S1S~\cite{Buchi60,Baaz&Ciabattoni&Zach00}, see also~\cite{Zach04} solving partially an open problem from~\cite[\S 9]{Kremer97bis}.

\subsection{Complexity classes and tiling problems}
\label{sec:complexity}
In this section, we introduce tiling problems that are mainly used in
Sections~\ref{section-aexppol-hardness} and~\ref{section-uniform-reduction}.

Let $\tow : \Nat \times \Nat \to \Nat$ be a tetration function defined for integers $n,k \geq 0$, 
inductively as $\tow(0,n) {=} n$ and~$\tow(k{+}1,n) {=} 2^{\tow(k,n)}$. 
Intuitively the function $\tow$ defines the tower of exponentials of  
height~$k$, \eg we have $\tow(1, n) = 2^n$, $\tow(2, n) = 2^{2^n}$, 
and so on.
By $\kNExpTime$ we denote the class of all problems decidable with a
nondeterministic Turing machines of working time in $\BigOh{\tow(k,p(n))}$ for some polynomial $p(\cdot)$,
on each input of length $n$. We define $\tower$ as the class of all problems 
of time complexity bounded by a tower of exponentials, whose height is an 
elementary function~\cite{Schmitz16}. Thus, to show $\tower$-hardness (using elementary reductions~\cite{Schmitz16}), 
it is sufficient to prove $\kNExpTime$-hardness for all~$k$ using uniform 
reductions~\cite[Section 3.1]{Schmitz16}. It is worth recalling that \tower-hardness
is defined with the class of elementary reductions (i.e. those with time-complexity 
bounded by a tower of exponentials of fixed height)~\cite{Schmitz16}. Building a reduction from 
instances of $\kNExpTime$-hard problems with time-complexity $\amap(n,k)$ where the size $n$ of the input
leads to an elementary reduction when  $\amap(n,k)$ itself is elementary in $n$ and $k$.
This is  what we mean by a uniform reduction in $k$ to establish \tower-hardness. 

For proving hardness results, we make extensive use of tiling problems, see \eg~\cite{vanEmdeBoas97}.
%
\begin{defi} \label{def:tilings}
The \defstyle{tiling problem $\mathtt{Tiling}_k$} takes as inputs a triple $\triple{\cT}{\cH}{\cV}$
and $c \in \cT^{n}$ for some~$n \geq 1$ such that $\cT$ is a finite set of \defstyle{tile types},
$\cH \subseteq \cT \times \cT$ (resp. $\cV \subseteq \cT \times \cT$)
      represents the  \defstyle{horizontal (resp. vertical)  matching relation},
and  $c = t_0,t_1, \ldots, t_{n-1} \in \cT^{n}$ 
is the \defstyle{initial condition}. 
A \defstyle{solution} for the instance $\triple{\cT}{\cH}{\cV}, c$ is a mapping 
$\atiling : \interval{0}{\tow(k,n)-1} \times \interval{0}{\tow(k,n)-1} \to \cT$ 
such~that:
\begin{description}
\itemsep 0 cm 
\item[\desclabel{(init)}{tiling:init}] For all $i \in \interval{0}{n-1}$, $\atiling(0,i) = t_i$.
\item[\desclabel{(hori)}{tiling:hori}]  For all $i, j \in \interval{0}{\tow(k,n)-1}$,
              $i < \tow(k,n)-1$ implies $(\atiling(i,j), \atiling(i {+} 1, j)) \in \cH$.
\item[\desclabel{(verti)}{tiling:verti}] For all $i, j \in \interval{0}{\tow(k,n)-1}$,
              $j < \tow(k,n)-1$ implies $(\atiling(i,j),  \atiling(i, j {+} 1)) \in \cV$. 
\end{description}
A mapping $\atiling$ satisfying~\ref{tiling:hori} and~\ref{tiling:verti} is called a \defstyle{tiling}.
\cut{
A $k$-\emph{tiling problem} $\cP$ is a tuple $(\cT, \cH, \cV)$, 
where $\cT$ is a finite set of tile types and 
$\cH, \cV \subseteq \cT \times \cT$ represent the horizontal and
vertical matching conditions. Let $\cP$ be a tiling problem and let
$c = t_0,t_1, \ldots, t_{n-1} \in \cT^{n}$ be an initial condition.
A mapping $\tau : \{ 0,1, \ldots, \tow(k,n)-1 \} \times 
\{ 0,1, \ldots, \tow(k,n)-1 \} \rightarrow \cT$ is a \emph{solution} for 
$\cP$ and $c$ if and only if, for all $i,j$ s.t., $i+1,j+1 \leq \tow(k,n)-1$, 
the following holds $(\tau(i,j), \tau(i {+} 1, j)) \in \cH, (\tau(i,j), 
\tau(i, j {+} 1)) \in \cV$ and~$\tau(0,i) = t_i$ for all $i < n$.}
\end{defi}
\noindent The problem of checking if an instance of $\mathtt{Tiling}_k$
has a solution (note that $k$ does not appear in the instance and it governs
the size of the grid) is  $\kNExpTime$-complete~\cite{vanEmdeBoas97}.

Given $N \geq 2$,  let us consider the satisfiability problem for \tQCTLEX{, \leq N} in which the structures
are tree models where all the maximal branches are infinite but each node has at most $N$ children (and at least one child). 
To characterise the complexity of \satproblem{\tQCTLEX{, \leq N}}, we consider the
 complexity class \aexppol that
consists of all problems decidable by an alternating Turing machine (ATM)~\cite{ChandraKS81} working in 
exponential-time and using only polynomially many alternations~\cite{Bozzellietal17,Molinari19}. We stress here that 
allowing an unbounded number of alternations would give us the class $\ExpSpace$, and classes similar to \aexppol
have been considered in~\cite{Berman80}, typically  STA($\mathfrak{f}(n)$,$\mathfrak{g}(n)$,$\mathfrak{h}(n)$), 
where $\mathfrak{f}(n)$ refers to the restriction on the Space, 
$\mathfrak{g}(n)$ refers to the restriction on the Time,
and $\mathfrak{h}(n)$ refers to the restriction on the number of Alternations. 
Consequently,  \aexppol is the union of classes STA($\cdot$,$2^{\mathfrak{g}(n)}$,$\mathfrak{h}(n)$)
with polynomials $\mathfrak{g}(n)$, $\mathfrak{h}(n)$.
The complexity of several logical problems has been captured by the class \aexppol, see \eg~\cite{Ferrante&Rackoff75,Bozzelli&vanDitmarsch&Pinchinat15,Bozzellietal17}.  

For proving \aexppol-hardness, we use an elegant modification of 
$\mathtt{Tiling}_1$, introduced in~\cite{Bozzellietal17,Molinari19}.
The extension amounts to considering a stack of $n$ tilings,
with a matching relation between two consecutive tile types on
the same position of the grid, and quantifications over the tile types on 
the first row (initial conditions). Details follow below. 

\begin{defi}
The  \defstyle{alternating multi-tiling problem} $\cPmulti$ takes as 
inputs an even number $n$ (in unary), $\triple{\cT}{\cH}{\cV}$ (as for defining $\mathtt{Tiling}_1$), 
$\cTzero \subseteq \cT$, 
$\cTacc \subseteq \cT$ and  $\cTmulti \subseteq \cT \times \cT$.
Given an initial condition $c= (w_1, \ldots, w_n) \in (\cTzero^{2^n})^{n}$, 
a \defstyle{solution} for $c$ is a  multi-tiling $(\atiling_1, \ldots, \atiling_n)$ on the grid~$\interval{0}{2^n-1} \times \interval{0}{2^n-1}$ such that: 
\begin{description}
\itemsep 0 cm 
\item[\desclabel{(m-init)}{amtp:m-init}] For all $\alpha \in \interval{1}{n}$, for all $j \in \interval{0}{2^n-1}$, $\atiling_{\alpha}(0,j) = w_{\alpha}(j)$
(\ie the first row of $\atiling_{\alpha}$ is $w_{\alpha}$).
\item[\desclabel{(m-tiling)}{amtp:m-tiling}] For  $\alpha \in \interval{1}{n}$, $\atiling_{\alpha}$ satisfies  \ref{tiling:hori} and \ref{tiling:verti}.
\item[\desclabel{(m-multi)}{amtp:m-multi}]  For  $\alpha \in \interval{1}{n-1}$,  
                     for all $i,j \in  \interval{0}{2^n-1}$, $\pair{\atiling_{\alpha}(i,j)}{\atiling_{\alpha+1}(i,j)} \in \cTmulti$.
\item[\desclabel{(m-accept)}{amtp:m-accept}] For some $j \in \interval{0}{2^n-1}$, $\atiling_{n}(2^n-1,j) \in \cTacc$. 
\end{description}
An instance $\mathcal{I}$ for $\cPmulti$ made of 
 $n$, $\triple{\cT}{\cH}{\cV}$, $\cTzero,\cTacc \subseteq \cT$, 
$\cTmulti \subseteq \cT \times \cT$ is positive iff
for all~$w_1 \in \cTzero^{2^n}$, there is $w_2 \in \cTzero^{2^n}$ such that $\ldots$ 
for all $w_{n-1} \in \cTzero^{2^n}$, there is $w_{n} \in \cTzero^{2^n}$ such that
there is a solution  $(\atiling_1, \ldots, \atiling_n)$ 
for $(w_1, \ldots, w_n)$. 
Note that this sequence involves $n{-}1$ quantifier alternations.
\end{defi}

\noindent The alternating multi-tiling problem is shown to be \aexppol-complete in~\cite{Bozzellietal17,Molinari19}.


\section{What happens when trees are bounded?}\label{section-bounded-degree}

In this section, we study the satisfiability problem for \tQCTLEX{,\leq N} with $N \geq 1$, 
\ie{} \tQCTLEX{} over trees, where the degree of each node is bounded by a fixed
natural number $N \geq 1$. 
As~we have already mentioned in the introduction, the goal of this section is two-fold.
First, we would like to make the reader familiar with our proof techniques applied to a simplified scenario.
Second, our results show that to get $\TOWER$-hardness of \tQCTLEX{} we must focus on trees of arbitrarily large branching.


\subsection{A toolkit for introducing local nominals}\label{sec:local-nominals}
Below, we introduce formulae to simulate partially the use of nominals from hybrid modal logics~\cite{Areces&Blackburn&Marx01}. 
A nominal $\anominal$ is usually understood as a propositional variable true  at exactly one  world of the model (a global property). 
In \tQCTLEX{}, such a property cannot be enforced but it can be done with respect to nodes at a bounded depth from the evaluation node, whence the adjective `local' for the nominals. 
The use of local nominals is essential in all our hardness proofs, as it allows us to simulate first-order quantification on a given set of nodes of bounded depth.

Proofs of all lemmas from this section are rather straightforward and are shown by careful inspection of the semantics. Hence, we delegate them to Appendix~\ref{appendix:proof-of:lemma:nominal}--\ref{appendix:proof-of:lemma-atoflengthd}.

\begin{defi}
Given a tree model $\atreemodel$ and a node $\anode$, we say that the propositional variable $\anominal$ is a \defstyle{nominal for the depth $k \geq 1$ from $\anode$}
iff there is $\anode' \in V$ with $\anode E^{k} \anode'$ such that $\atreemodel, \anode' \models \anominal$, and for all worlds~$\anode'' \neq \anode'$ s.t.
$\anode E^{k} \anode''$, we have $\atreemodel, \anode'' \not \models \anominal$ ($E^k$ is the $k$-fold composition of~$E$).  
\end{defi}

The formula $\bindd{\anominal}{k}$ defined as 
$\EX^k \anominal \wedge \neg \exists \ \avarprop \ (\EX^k(\anominal \wedge \avarprop) \wedge \EX^k(\anominal \wedge \neg \avarprop))$,
where $\avarprop$ is distinct from $\anominal$, states that $\anominal$ is a nominal for the depth $k$ ($\EX^k$ denotes $k$ copies of $\EX$).

\begin{lem}\label{lemma:nominal}
$\anominal$ is a nominal for the depth $k \geq 1$ from $\anode$ iff $\atreemodel, \anode \models \bindd{\anominal}{k}$.
\end{lem}
%

\noindent Let us next define $\att{\anominal}{k} \aformula$ as the formula $\EX^{k}(\anominal \wedge \aformula)$
(usually assuming that $\bindd{\anominal}{k}$ holds). 

\begin{lem}\label{lemma:at}
Assuming that $\anominal$ is a nominal for the depth $k \geq 0$ from $\anode$ such that $\anode E^k \anode'$ and $\atreemodel, \anode' \models \anominal$, 
we have $\atreemodel, \anode \models \att{\anominal}{k} \aformula$ iff $\atreemodel, \anode' \models \aformula$. 
\end{lem}
%
Given $d \geq 1$ and propositional variables $\anominal_1$, \ldots, $\anominal_d$ (that play the role of nominals), 
we often write~$\att{\anominal_1, \ldots, \anominal_d}{} \aformula$ to denote 
the formula $\att{\anominal_1}{1} \att{\anominal_2}{1} \cdots \att{\anominal_d}{1} \aformula$
(usually assuming that  $\bindd{\anominal_1}{1}$ holds and for all $i \in \interval{2}{d}$, 
$\att{\anominal_1}{1} \att{\anominal_2}{1} \cdots \att{\anominal_{i-1}}{1} \bindd{\anominal_i}{1}$ holds true). 
We also use $\att{\bar{\anominal}}{} \aformula$ instead of $\att{\anominal_1, \ldots, \anominal_d}{} \aformula$ 
when $\bar{\anominal}$ is understood as $\anominal_1, \ldots, \anominal_d$. 
Given a node $\anode_0$ such that 
\[
\atreemodel, \anode_0 \models \bindd{\anominal_1}{1} \wedge  
\bigwedge_{i \in \interval{2}{d}} \att{\anominal_1}{1} \att{\anominal_2}{1} \cdots \att{\anominal_{i-1}}{1} \bindd{\anominal_i}{1},
\]
we write $\anode_1, \ldots, \anode_d$ to denote the unique sequence of nodes such that
for all $i \in \interval{1}{d}$, we have both~$\anode_{i-1} E \anode_i$ and $\atreemodel, \anode_i \models \anominal_i$. 
The existence and uniqueness of the nodes $\anode_1, \ldots, \anode_d$ 
follow from Lemma~\ref{lemma:nominal} and Lemma~\ref{lemma:at}. 
Here is another useful lemma justifying the use of the introduced abbreviations. 

\begin{lem}\label{lemma:atoflengthd}
Assume that $\atreemodel, \anode_0 \models \bindd{\anominal_1}{1} \wedge  \bigwedge_{i \in \interval{2}{d}}
\att{\anominal_1}{1} \att{\anominal_2}{1} \cdots \att{\anominal_{i-1}}{1} \bindd{\anominal_i}{1}$
and the sequence $\anode_1, \ldots, \anode_d$ is associated to $\bar{\anominal} = \anominal_1, \ldots \anominal_d$.
Then, $\atreemodel, \anode_0 \models \att{\bar{\anominal}}{} \aformula$ iff $\atreemodel, \anode_d \models \aformula$.
\end{lem}
%
%
Let $\distinctbindd{\anominal_1, \ldots, \anominal_{\alpha}}{k}$ be an abbreviation of
$\bigwedge_{i \in \interval{1}{\alpha}} \bindd{\anominal_i}{k} \wedge
\bigwedge_{j \in \interval{1}{\alpha}} \bigwedge_{i < j} \neg \att{\anominal_i}{k} \anominal_j$.
It allows us to name $\alpha$ distinct nodes at the depth $k$. Hence, the respective nodes
interpreting the nominals~$\anominal_1$, \dots, $\anominal_{\alpha}$ are pairwise distinct.
It is summarised with the lemma below, which proof is a slight variant the proof of Lemma~\ref{lemma:nominal}.

\begin{lem} 
\label{lemma:distinct-nominals}
Given a tree model $\atreemodel$ and a node $\anode$, we have 
$\atreemodel, \anode \models \distinctbindd{\anominal_1, \ldots, \anominal_{\alpha}}{k}$ iff
there are~$\alpha$ distinct nodes $\anode_1,\ldots, \anode_{\alpha}$ 
such that for all $i \in \interval{1}{\alpha}$, 
$\anominal_i$ is a nominal for the depth $k \geq 1$ from $\anode$.
\end{lem}

Let us illustrate the use of local nominals and $\distinctbindd{\anominal_1, \ldots, \anominal_{\alpha}}{k}$ 
to specify that a node has at most $2^n$ children, 
with a formula of polynomial size in $n$. 
This is exactly the type of properties that can be expressed in graded modal logics~\cite{Fine72b,FattorosiBarnaba&DeCaro85,Tobies00}.
Given a finite set $\aset$ of propositional variables,
we design a formula stating that no pair of distinct children
agree on all propositional variables from $\aset$, as done in~\cite{DavidLM16}.
It is given below:
\[
\mathrm{Uni}(\aset) \ \egdef \  
\forall \anominal, \anominalbis \ 
\distinctbindd{\anominal,\anominalbis}{1}
\rightarrow
\neg
\left(
\bigwedge_{\avarprop \in \aset}  (\att{\anominal}{1} \avarprop  \leftrightarrow \att{\anominalbis}{1} \avarprop)
\right)
\]
Thus, the formula $\Diamond_{\leq 2^n} \top$ from graded modal logics can be expressed in \tQCTLEX{}
with 
\[
\Diamond_{\leq 2^n} \top \ \egdef \  \exists \ \avarprop_0, \ldots, \avarprop_{n-1} \ \mathrm{Uni}(\set{\avarprop_0, \ldots, \avarprop_{n-1}}).
\]

In Section~\ref{section-tower-hardness}, we 
show how to succinctly express hyperexponential bounds.


\subsection{Beyond the \texorpdfstring{\expspace}{ExpSpace} upper bound: \texorpdfstring{\aexppol}{AExpPol}}\label{sec:upper_bd} 

In order to solve \satproblem{\tQCTLEX{,\leq N}}, little is needed if the \expspace upper bound is aimed. 
Indeed, given a formula $\aformula$ in \tQCTLEX{,\leq N}, it is clear that for an $N$-bounded tree model $\atreemodel$ satisfying $\aformula$
at its root node $\aroot$, it is irrelevant what happens at nodes of depth strictly more than $\md{\aformula}$. 
Hence, the formula $\aformula$ is satisfiable iff there is a finite $N$-bounded tree structure $\atreemodel$ with all the branches of length 
exactly $\md{\aformula}$ satisfying $\aformula$ at its root $\aroot$ (as the branches of tree models are infinite, 
we need to consider branches of length exactly $\md{\aformula}$). 
Thus, $\atreemodel$ has at most $\length{\aformula} N^{\length{\aformula}}$ nodes.
To get an algorithm working in \nexpspace, guess such an exponential-size finite tree structure,
and perform model-checking on it with an algorithm inherently in \pspace (as
model-checking finite  structures with MSO is \pspace-complete~\cite{Stockmeyer74,Vardi82} 
and $\aformula$ can be translated
to MSO in the standard way), which leads to \nexpspace. 
By Savitch's Theorem~\cite{Savitch70}, we get the \expspace upper bound.

This  bound is not completely satisfactory as it does not use much of \tQCTLEX{,\leq N} and more importantly, 
Section~\ref{section-aexppol-hardness} proves \aexppol-hardness of  \satproblem{\tQCTLEX{,\leq N}} as long as~$N \geq 2$. 
Hence, the goal of this section is to establish an \aexppol upper bound. 
The tight upper bound for \satproblem{\tQCTLEX{,\leq N}} relies on the  following ingredients.
\begin{itemize}[left=4mm]
  \item[(i)] Every formula $\aformula$ of \tQCTLEX{} is logically equivalent to a \tQCTLEX{} formula $\aformula'$ 
  in prenex normal form  (PNF) such that $\aformula'$ can be computed in polynomial-time in $\length{\aformula}$. 
  Formulae in PNF are of the form $\mathcal{Q}_1 \ \avarprop_1 \ \cdots \mathcal{Q}_{\beta} \ \avarprop_{\beta}\ 
  \aformulabis$ where $\set{\mathcal{Q}_1, \ldots, \mathcal{Q}_{\beta}} \subseteq \set{\exists,\forall}$ and 
  $\aformulabis$ is quantifier-free. 
  \item[(ii)] Existence of an $N$-bounded tree model for  $\aformula$ is equivalent to the existence
  of an $N$-bounded finite tree structure such that all branches are of length $\md{\aformula}$.
  Then, we simply guess a finite tree of a small (exponential) size with the help of the \defstyle{shallow model property} - 
  such a tree will be later unravelled to become an infinite tree model.

  \item[(iii)] Checking whether $\atreemodel, \aroot \models \mathcal{Q}_1 \ \avarprop_1 \ \cdots \mathcal{Q}_{\beta} \ \avarprop_{\beta} \
  \aformulabis$ (involving an $N$-bounded finite tree with branches of length $\md{\aformulabis}$ and the input formula in PNF)
  can be done with an alternating Turing machine in time $\mathcal{O}((\length{\aformulabis}+\beta) \length{\atreemodel})$ 
  and with at most $\beta$ alternations. 
\end{itemize}
\noindent
To establish (i), we cannot rely directly on PNF for $\QCTL$ from~\cite[Prop. 3.1]{LaroussinieM14} as the translation
in~\cite[Prop. 3.1]{LaroussinieM14} involves temporal operators beyond the language of \tQCTLEX. 
\begin{lem}\label{lemma:pnf}
For every formula $\aformula$ in \tQCTLEX{}, one can compute in polynomial-time in
$\length{\aformula}$ a logically equivalent formula in PNF
$\mathcal{Q}_1 \ \avarprop_1 \ \cdots \mathcal{Q}_{\beta} \ \avarprop_{\beta}\ \aformulabis$
with $\beta \leq \length{\aformula}$.
\end{lem}
\begin{proof}
On tree models, the following formulae are tautologies, 
assuming that $\avarprop$ does not occur in $\aformulabis'$ (where $\mathcal{Q}$ is either $\exists$ or $\forall$):
\[
\EX \ \mathcal{Q} \ \avarprop \ \aformulabis \leftrightarrow 
\mathcal{Q} \ \avarprop \ \EX \aformulabis \ \ \ \
(\mathcal{Q} \ \avarprop \ \aformulabis) \wedge \aformulabis' \leftrightarrow \mathcal{Q} \ \avarprop 
\ (\aformulabis \wedge \aformulabis') 
\ \ \ \
\neg \exists  \ \avarprop \ \aformulabis \leftrightarrow \forall \ \avarprop \ \neg \aformulabis \ \ \ \ 
\neg \forall  \ \avarprop \ \aformulabis \leftrightarrow \exists \ \avarprop \ \neg \aformulabis
\]
Hence, by employing the above formulae and rewriting the input, we conclude the lemma. 
For a more detailed explanation consult Apppendix~\ref{appendix:proof-of:lemma-pnf} 
\end{proof}


Now we proceed with the second property. Let us be a bit more precise.
Given a tree model $\atreemodel$, we write $\arestrictedtreemodel{n}$ to denote
its subtree obtained by taking only nodes on the depth at most $n$ from the root. 
A \defstyle{completion} of a finite tree $\atreemodel'$ of maximal depth~$n$ is an infinite tree $\atreemodel$ 
(finite-branching and all the maximal branches are infinite)  such that $\atreemodel' = \arestrictedtreemodel{n}$.
By the \defstyle{naive completion} of $\atreemodel'$ of maximal depth $n$, we refer to 
the unique completion achieved by replacing each node $\anode$ at depth $n$ from 
$\atreemodel$ by an infinite chain of copies of itself.

A \defstyle{shallow model property} states that what matters 
for a $\tQCTLEX{}$ formula $\aformula$ in its infinite tree model is a relatively 
small finite part, with paths bounded by the modal depth of~$\aformula$. 

\begin{lem}[Shallow Model Property] \label{lemma:completion}
Let $\atreemodel,\aroot$ be a model for the \tQCTLEX{}--formula $\aformula$. 
Then, any completion of~$\arestrictedtreemodel{\md{\aformula}},\aroot$ is also a model for~$\aformula$.
\end{lem}
\begin{proof}
The construction is standard and goes in exactly the same way as in the 
literature, \eg~\cite[Theorem~2.3]{blackburn_rijke_venema_2001}. 
\end{proof}

Since we are interested in the satisfiability problem over $N$-bounded trees, 
the overall size of a structure $\arestrictedtreemodel{\md{\aformula}}$ 
can be easily bounded. A simple estimation can be obtained by counting the number of 
nodes with a certain distance from the root, namely:
\[
\length{\arestrictedtreemodel{\md{\aformula}}} \leq 1 + N + N^2 + \ldots + 
N^{\md{\aformula}} < \length{\aformula} \cdot N^{\length{\aformula}}
\]

As a direct consequence of the above estimation and Lemma~\ref{lemma:completion}, 
we obtain:
\begin{lem} \label{lemma:small_model}
For any formula $\aformula$, $\aformula$ is satisfiable for \tQCTLEX{,\leq N}
iff $\aformula$ is satisfiable in a finite $N$-bounded tree structure of size bounded by 
$\length{\aformula} N^{\length{\aformula}}$ and each branch is of length $\md{\aformula}$. 
\end{lem}

In order to establish (iii), the details are omitted but  
we apply the naive model-checking algorithm for MSO with an ATM:
existential (resp. universal) quantification $\exists \avarprop$ (resp. $\forall \avarprop$) requires time $\mathcal{O}(\length{\atreemodel})$ and the machine enters a sequence of existential (resp. universal) states.
The quantifier-free formula $\aformulabis$ is evaluated as a first-order formula
by the standard translation for modal logic. Note also that checking $\atreemodel', \anode \models \aformulabis$
can be done in polynomial time in $\length{\aformulabis} + \length{\atreemodel'}$ (see~\cite{ClarkeE81,Schnoebelen03}).
By combining (i)-(iii) we establish an improved upper bound.

\begin{thm} 
\label{thm:aexppol-upper}
For any $N \geq 1$, the satisfiability problem for \tQCTLEX{,\leq N} is in \aexppol.
\end{thm}

\noindent When $N=1$, the upper bound can be improved as the number of alternations is linear and the size of the finite witness ``tree'' is  polynomial in $\length{\aformula}$,
and therefore the whole procedure can be implemented with a polynomial-time 
alternating Turing machine (thus in $\pspace$~\cite{ChandraKS81}). 
The matching lower bound is inherited from quantified propositional logic QBF, 
see \eg~\cite{Meyer&Stockmeyer73}.
\begin{cor} \label{corollary:qctl1bounded}
The satisfiability problem for \tQCTLEX{,\leq 1} is \pspace-complete.
\end{cor}
 

\subsection{A reduction from \texorpdfstring{$\cPmulti$}{AMTP} (with fixed  \texorpdfstring{$N \geq 2$}{N at least 2})}\label{section-aexppol-hardness}

Let $N \geq 2$ and let us consider the satisfiability problem for \tQCTLEX{, \leq N} in which the structures
are tree models where all the maximal branches are infinite and each node has at most $N$ children 
(and at least one child). 
In order to show that the problem is \aexppol-hard, we define below a reduction from the alternating multi-line
tiling problem $\cPmulti$  presented in Section~\ref{section-preliminaries} and introduced in~\cite{Bozzellietal17}.

To define a grid  $\interval{0}{2^n-1} \times  \interval{0}{2^n-1}$, a major part in the solution
of an instance of $\cPmulti$, we specify a tree such that every node at a distance less than $2n$ from the 
root $\aroot$ has exactly two children, implying that there are exactly $2^{2n}$ nodes at a distance $2n$ from $\aroot$.
Moreover, each node at a distance~$2n$ encodes a position $(\mathfrak{H},\mathfrak{V})$ 
in  $\interval{0}{2^n-1} \times  \interval{0}{2^n-1}$, by making the propositional variables~$\ahvarprop_{n-1}, \ldots, \ahvarprop_{0}$ and  $\avvarprop_{n-1}, \ldots, \avvarprop_{0}$ to be responsible, respectively, for
the horizontal and vertical axes.
The $i$-th bit of $\mathfrak{H}$ (resp. $\mathfrak{V}$) is taken care of by the truth value of $\ahvarprop_i$ (resp. $\avvarprop_i$)
and by convention 
the least significant bit is encoded by $\ahvarprop_0$ (resp. $\avvarprop_0$). 

The forthcoming formula $\mathrm{grid}(2n)$ is dedicated to encoding such a grid.
\begin{multline*}
\mathrm{grid}(2n) \ \egdef \ \left( \bigwedge_{i \in \interval{0}{2n-1}} \AX^i \ \EX_{=2} \top)  \right) \wedge\\
\forall \anominal, \anominalbis \ \distinctbindd{\anominal, \anominalbis}{2n} \rightarrow
\left( \bigvee_{j \in \interval{0}{n-1}} \neg (\att{\anominal}{2n} \ahvarprop_j \leftrightarrow \att{\anominalbis}{2n} \ahvarprop_j)
\vee \neg (\att{\anominal}{2n} \avvarprop_j \leftrightarrow \att{\anominalbis}{2n} \avvarprop_j)
\right),
\end{multline*}
where $\EX_{=2} \top \egdef \exists \ \anominal_1,  \anominal_2 \
\distinctbindd{\anominal_1, \anominal_2}{1} \wedge \AX (\anominal_1 \vee  \anominal_2)$
 states that there are exactly two children.
Moreover, $\AX^0 \aformulabis \egdef \aformulabis$ and
$\AX^{i+1} \aformulabis \egdef \AX \AX^i \aformulabis$. 

Note
that the upper part of $\mathrm{grid}(2n)$ enforces that there are exactly $2^{2n}$ descendants at a distance~$2n$ from the root, while the lower part imposes that any two such descendants differ by at least one  propositional variable from $\ahvarprop_{n-1}, \ldots, \ahvarprop_{0}, \avvarprop_{n-1}, \ldots, \avvarprop_{0}$.
Hence, the full grid $\interval{0}{2^n-1} \times  \interval{0}{2^n-1}$  is encoded with $\mathrm{grid}(2n)$. 
The correctness of $\mathrm{grid}(2n)$ follows from Lemma~\ref{lemma:at} and Lemma~\ref{lemma:distinct-nominals}.
\begin{cor}\label{corr:grid-2n-correctness}
$\atreemodel, \anode \models \mathrm{grid}(2n)$ iff $\arestrictedtreemodel{2n}$ is a binary tree in which there are exactly $2^{2n}$ nodes $\anode'$ satisfying~$E^{2n}(\anode, \anode')$ and each of such distinct
nodes $\anode', \anode''$ is labelled by a different subset of atomic propositions from~$\{ \avvarprop_0, \avvarprop_1, \ldots, \avvarprop_{n-1}, \ahvarprop_0, \ahvarprop_1, \ldots, \ahvarprop_{n-1} \}$.
\end{cor}

Let $\triple{\cT}{\cH}{\cV}$ be a triple from an instance of  $\cPmulti$ and let $j \in \Nat$. 
Each tile type  $t \in \cT$ will be represented by a fresh propositional variable $t^j$. 
Hence, $\set{t^j: t \in \cT}$ (written below~$\cT^j$) is a set of propositional variables used to provide a tile type for each position of the grid $\interval{0}{2^n-1} \times  \interval{0}{2^n-1}$, 
while the superscript `$j$' is handy to remember that this concerns the $j$-th tiling
(as several tilings are involved in $\cPmulti$ instances).  

We first define the formulae $\aformula_{\rm cov}^j$,  $\aformula_{\cH}^j$ and  $\aformula_{\cV}^j$
whose conjunction states that every position of the grid has a unique tile type in $\cT^j$, 
and the horizontal and vertical matching conditions are satisfied. 
Hence, we have a valid tiling of the grid made from $\cT^j$.

The formula $\aformula_{\rm cov}^j$ expresses that every position of the grid
has a unique tile type:
\[
\aformula_{\rm cov}^j \egdef 
\forall \anominal \ \bindd{\anominal}{2n}
\rightarrow \att{\anominal}{2n} \left( 
\bigvee_{t \in \cT} t^j \wedge \bigwedge_{t \neq t' \in \cT} \neg (t^j \wedge t'^j)
\right). 
\]
For the horizontal matching constraints, we need to express when two nodes at a distance $2n$ interpreted
respectively by $\anominal$ and $\anominalbis$ and representing respectively the position $(\mathfrak{H},\mathfrak{V})$ and $(\mathfrak{H}', \mathfrak{V}')$, satisfy
$\mathfrak{V} = \mathfrak{V}'$ and $\mathfrak{H}' = \mathfrak{H}+1$. 
The formula ${\rm HN}(\anominal,\anominalbis)$ (`HN' for 'horizontal neighbours') 
does the job using a standard arithmetical reasoning on binary numbers.

The intuition
is that we treat $\ahvarprop_i$ propositions as bits and to verify that the number encoded on~$\anominalbis$ is equal to the number encoded on $\anominal$ plus $1$, 
we need to (i) find an index $i$ on which the $i$-th bit is switched on for $\anominalbis$ but switched off for $\anominal$, (ii) check that all bits on more significant positions after $i$ for $\anominal$ and $\anominalbis$ are equal and
(iii) ensure that all less significant bits are switched on for $\anominal$ while switched off for $\anominalbis$. This idea is formalised as follows:
\begin{multline*}
{\rm HN}(\anominal,\anominalbis) \egdef
\left(  
\bigwedge_{\alpha \in \interval{0}{n-1}} \att{\anominal}{2n} \avvarprop_{\alpha} \leftrightarrow \att{\anominalbis}{2n} 
\avvarprop_\alpha
\right)
\wedge 
\bigvee_{i \in \interval{0}{n-1}} \Bigg(
\att{\anominal}{2n} \neg \ahvarprop_i \wedge
\att{\anominalbis}{2n} \ahvarprop_i \; \wedge\\
\wedge \; \bigwedge_{\alpha \in \interval{0}{i-1}} \left( \att{\anominal}{2n} \ahvarprop_{\alpha} \wedge  \att{\anominalbis}{2n} \neg \ahvarprop_{\alpha} \right)
\wedge
(\bigwedge_{\alpha \in \interval{i+1}{n}} \left( \att{\anominal}{2n} \ahvarprop_{\alpha} \leftrightarrow  \att{\anominalbis}{2n} \ahvarprop_{\alpha}) \right)
\Bigg).
\end{multline*}
Employing ${\rm HN}(\anominal,\anominalbis)$ we can provide a formula $\aformula_{\cH}^j$ that encodes horizontal matching constraints:
\[
\aformula_{\cH}^j \egdef 
\forall \anominal,  \anominalbis \ 
\left(
\bindd{\anominal}{2n} \wedge
\bindd{\anominalbis}{2n} \wedge {\rm HN}(\anominal,\anominalbis)
\right)
\rightarrow
\bigvee_{\pair{t}{t'} \in \mathcal{H}}  \left( \att{\anominal}{2n} \ t^j \wedge \att{\anominalbis}{2n} \ t'^j \right).
\]
Let ${\rm VN}(\anominal,\anominalbis)$ (where `VN' stands for `vertical neighbours') 
be the formula obtained from  ${\rm HN}(\anominal,\anominalbis)$ by replacing each occurrence of $\ahvarprop_{\alpha}$ (resp.  $\avvarprop_{\alpha}$) by $\avvarprop_{\alpha}$ (resp.  $\ahvarprop_{\alpha}$). 

The following formula $\aformula_{\cV}^j$ encodes the vertical matching constraints:
\[
\aformula_{\cV}^j \egdef 
\forall \anominal,  \anominalbis \ 
\left(
\bindd{\anominal}{2n} \wedge
\bindd{\anominalbis}{2n} \wedge {\rm VN}(\anominal,\anominalbis)
\right)
\rightarrow
\bigvee_{\pair{t}{t'} \in \mathcal{V}}  \left( \att{\anominal}{2n} \ t^j \wedge \att{\anominalbis}{2n} \ t'^j \right).
\]
To state that a root satisfying ${\rm grid}(2n)$ encodes a tiling with 
respect to $\cT^j$, we consider the formula 
\[
\aformula_{\rm tiling}^j  \egdef 
\aformula_{\rm cov}^j \wedge \aformula_{\cH}^j  \wedge \aformula_{\cV}^j.
\] 

\begin{lem}\label{lemma:tiling-ok}
Assume that $\atreemodel, \anode \models \mathrm{grid}(2n)$ holds. Then:
\begin{itemize}
    \item If $\atreemodel, \anode \models \aformula_{\rm tiling}^j$ then  $\tau : \interval{0}{2^n-1} \times \interval{0}{2^n-1} \to \cT^j$, defined as $\tau(x,y) = t^j$ for $\atreemodel, \anode_{\pair{x}{y}} \models t^j$, where $\anode_{\pair{x}{y}}$ is the unique encoding 
of the position $\pair{x}{y}$ satisfying~$E^{2n}(\anode, \anode_{\pair{x}{y}})$, is a tiling. 
    \item If $\tau : \interval{0}{2^n-1} \times \interval{0}{2^n-1} \to \cT^j$ is a tiling, then there exists a tree $\atreemodel', \anode'$ satisfying $\aformula_{\rm tiling}^j \wedge \mathrm{grid}(2n)$ and being a $\cT^j$-variant of $\atreemodel$.
\end{itemize}
\end{lem}
\begin{proof}
By careful inspection of the semantics and of presented formulae, \cf Appendix~\ref{appendix:proof-of:tiling-ok}.
\end{proof}

In order to 
encode an instance of $\cPmulti$, there are still  properties that need to be expressed. 
Let us assume that the root node $\aroot$ satisfies ${\rm grid}(2n)$.
\begin{itemize}
\itemsep 0 cm 
\item Given  
the set of initial tile types 
$\cTzero \subseteq \cT$, 
let us express that
for each position of the first row of the grid, exactly one tile type in  $\cTzero^j$
holds.
\[
\aformula_{\rm init}^j \egdef
\forall\anominal \ 
\left( 
\bindd{\anominal}{2n} \wedge \att{\anominal}{2n}(\bigwedge_{\alpha \in \interval{0}{n-1}} \neg \ahvarprop_{\alpha})
\right)
\rightarrow
\att{\anominal}{2n} \left( \bigvee_{t \in \cTzero} t^j \wedge \bigwedge_{t \neq t' \in \cTzero} \neg (t^j \wedge t'^j) \right)
\]

\item Assuming  that $\aroot$  satisfies  $\aformula_{\rm tiling}^j \wedge \aformula_{\rm init}^{j'}$,
we express that
for  each position of the first row of the grid, the tile type in $\cTzero^j$
coincides with the tile type in $\cTzero^{j'}$ (corresponding to \ref{amtp:m-init}):
\[
\aformula_{\rm coinci}^{j,j'} \egdef
\forall \anominal \left( 
(\bindd{\anominal}{2n} \wedge \att{\anominal}{2n}(\bigwedge_{\alpha \in \interval{0}{n-1}} \neg \ahvarprop_{\alpha}) \right)
\rightarrow 
\att{\anominal}{2n} \left( \bigvee_{t \in \cTzero} (t^j \wedge t^{j'}) \right). 
\]
\item Given the set of accepting tile types $\cTacc \subseteq \cT$ and assuming that $\aroot$ satisfies  $\aformula_{\rm tiling}^j$, we state that
there is a position on the last row with a tile type in $\cTacc$ (satisfying \ref{amtp:m-accept}):
\[
\aformula_{\rm acc}^j \egdef
\exists \anominal \ 
\bindd{\anominal}{2n} \wedge \att{\anominal}{2n} \left(  
\left( \bigwedge_{\alpha \in \interval{0}{n-1}} \ahvarprop_{\alpha} \right)
\wedge \bigvee_{t \in \cTacc} t^j \right). 
\]
\item Given the multi-matching tiling relation $\cTmulti \subseteq \cT \times \cT$, 
and assuming  $\aroot$ satisfies $\aformula_{\rm tiling}^j \wedge \aformula_{\rm tiling}^{j+1}$, 
on every position, the tile type in $\cT^j$ and the tile type in $\cT^{j+1}$ 
are in the relation $\cTmulti$ (fulfilling the requirements of \ref{amtp:m-multi}):
\[
\aformula_{\rm multi}^j \egdef
\forall \anominal \ 
\bindd{\anominal}{2n} \rightarrow
\att{\anominal}{2n} \left( \bigvee_{\pair{t}{t'} \in \cTmulti} ( t^j  \wedge t'^{j+1}) \right).
\]
\end{itemize}
It is time to wrap up. 
Given a finite set of propositional variables $\aset = \set{\avarpropter_1, \ldots, \avarpropter_{\beta}}$,
we write $\exists \aset \ \aformulabis$ to denote 
the formula $\exists \avarpropter_1 \ \exists \avarpropter_2 \
\cdots \exists \avarpropter_{\beta} \ \aformulabis$.  $\forall \aset \ \aformulabis$ is defined similarly. 
Given an instance $\mathcal{I}$ of  $\cPmulti$ made of $n$, 
$\triple{\cT}{\cH}{\cV}$, $\cTzero$, $\cTacc$, $\cTmulti$, let us define
the formula 
$\aformula_{\mathcal{I}}$ below:%
\begin{multline*}
\aformula_{\mathcal{I}} \egdef 
{\rm grid}(2n) \wedge 
\forall \cTzero^1 \ 
\exists \cTzero^2 \ 
\forall \cTzero^3 \  
\ldots 
\forall \cTzero^{n-1} \  
\exists \cTzero^n
\bigwedge_{j \in \interval{1}{n}} \aformula_{\rm init}^{j}
\rightarrow 
\Bigg(
\exists \set{t^j: t \in \cT, j \in \interval{n+1}{2n}} \\ 
\left(\bigwedge_{j \in \interval{n+1}{2n}} \aformula_{\rm tiling}^{j} \wedge \aformula_{\rm coinci}^{j,(j-n)} \right)
\wedge 
\left( \bigwedge_{j \in \interval{n+1}{2n-1}} \aformula_{\rm multi}^{j} \right)
\wedge
\aformula_{\rm acc}^{2n}
\Bigg).
\end{multline*}

Now, we can state the correctness of the reduction.
\begin{lem} \label{lemma:aexppol-hardness}
$\mathcal{I}$ is a positive instance of $\cPmulti$ iff 
$\aformula_{\mathcal{I}}$ is satisfiable in \tQCTLEX{, \leq N}.
\end{lem}
\begin{proof}
The proof is a bit tedious but has no serious difficulties, as all the conditions
for being a solution of $\mathcal{I}$ can be easily expressed, as soon as the grid  $\interval{0}{2^n-1} \times  \interval{0}{2^n-1}$
is encoded. Moreover, the quantifications involved in $\cPmulti$ are straightforwardly taken care of
in  \tQCTLEX{, \leq N} thanks to the presence of propositional quantification. 
Consult Appendix~\ref{appendix:proof-of:proof-lemma-aexppol-hardness}.
\end{proof}

The above lemma leads us to one of the main results of the paper.
\begin{thm} 
\label{theorem:aexppol-hardness}
For all $N \geq 2$, the satisfiability problem for \tQCTLEX{, \leq N} is \aexppol-hard.
\end{thm}


\section{\tower-hardness of the satisfiability problem \texorpdfstring{\tQCTLEX{}}{tQCTLEX} }\label{section-tower-hardness}

We are back to the (general) satisfiability problem for \tQCTLEX{}, \ie with no further restrictions on the number of children per node.

\subsection{Overview of the method}

In order to show \tower-hardness, we shall reduce the $k$-\nexptime-complete tiling problem $\mathtt{Tiling}_k$
introduced in Section~\ref{sec:complexity} to \satproblem{\tQCTLEX{}} and this should be
done in a uniform way so that \tower-hardness can be concluded (see the discussion in~\cite[Section 3.1.2]{Schmitz16} and
in Section~\ref{sec:complexity}).
Hence, we need to encode concisely a grid $\tow(k,n) \times \tow(k,n)$ and to do so, the main task consists in enforcing
that a node has $\tow(k,n)$ children, using a formula of elementary size in $k+n$ (bounded by a tower of exponentials of fixed height). Actually, our method produces
a formula of exponential size in $k+n$. Of course, this is not  the end of the story as we need to encode 
the grid $\tow(k,n) \times \tow(k,n)$ and to express on it constraints about the tiling $\atiling$. 
First, let us explain how to enforce that a node has exactly $\tow(k,n)$ children,
by partly taking  advantage of the proof technique of local nominals (see Section~\ref{sec:local-nominals}).

We recall that $\tow(1,n) {=} 2^n$ and $\tow(k+1,n) {=} 2^{\tow(k,n)}$ for $k > 0$.
For the forthcoming subsections, we assume that $n$ is fixed.
Below, we classify the nodes of a tree model by their \defstyle{type} (a value in $\Nat$)
such that any node is of type 0, and if a node is of type $k > 0$, then it has exactly
$\tow(k,n)$ children and all the children are of type $k-1$. To be more precise, a node may have
two types (as one of them is always zero). 
So, a node of type 1 has exactly $2^n$ children,  a node of type 2 has exactly $2^{2^n}$ children, etc.
Therefore if a node is of type $k>0$, then the value $k$ is unique. 
Additional conditions apply for being of type $k > 0$ but we can already notice
that a node $\anode$ of type $k$ implicitly defines a balanced subtree of depth $k$ with root $\anode$. 
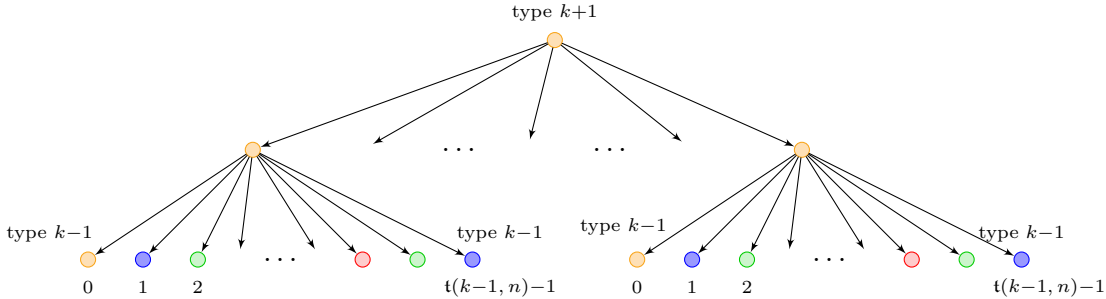
 \begin{center}
    \begin{tikzpicture}
      
      \begin{scope}[scale=.73,line cap=round]
      

        \draw (1,0) node[minirond, jaune] (AA) {};
        \draw (-2,-2) node[minirond, jaune] (AB1) {};
        \draw (-1,-2) node[minirond, bleu] (AB2) {};
        \draw (0,-2) node[minirond, vert] (AB3) {};
        \draw (0.75,-2) node[text=black] (AC1) {};
        \draw (1.5,-2) node[text=black] (AC2) {$\ldots$};
        \draw (2.25,-2) node[text=black] (AC3) {};
        \draw (3,-2) node[minirond, rouge] (AB4) {};
        \draw (4,-2) node[minirond, vert] (AB5) {};
        \draw (5,-2) node[minirond, bleu] (AB6) {};
        
        \draw (-2,-2.5) node[text=black] (AD1) {\tiny{0}};
        \draw (-1,-2.5) node[text=black] (AD2) {\tiny{1}};
        \draw (0,-2.5) node[text=black] (AD3) {\tiny{2}};
        \draw (3,-2.5) node[text=black] (AD4) {};
        \draw (4,-2.5) node[text=black] (AD5) {};
        \draw (5.5,-2.5) node[text=black] (AD6) {\tiny{$\tow(k{-}1,n){-}1$}};

        \draw (-2.7,-1.52) node[text=black] (AE1) {\tiny{\text{type} $k{-}1$}};
        \draw (5.5,-1.52) node[text=black] (AE2) {\tiny{\text{type} $k{-}1$}};

        \foreach \i/\j in {AA/AB1, AA/AB2, AA/AB3, AA/AB4, AA/AB5, AA/AB6, AA/AC1, AA/AC3}
                 {\draw[-latex'] (\i) -- (\j);}
                 

        \draw (10+1,0) node[minirond, jaune] (BA) {};
        \draw (10+-2,-2) node[minirond, jaune] (BB1) {};
        \draw (10+-1,-2) node[minirond, bleu] (BB2) {};
        \draw (10+0,-2) node[minirond, vert] (BB3) {};
        \draw (10+0.75,-2) node[text=black] (BC1) {};
        \draw (10+1.5,-2) node[text=black] (BC2) {$\ldots$};
        \draw (10+2.25,-2) node[text=black] (BC3) {};
        \draw (10+3,-2) node[minirond, rouge] (BB4) {};
        \draw (10+4,-2) node[minirond, vert] (BB5) {};
        \draw (10+5,-2) node[minirond, bleu] (BB6) {};
        
        \draw (10+-2,-2.5) node[text=black] (BD1) {\tiny{0}};
        \draw (10+-1,-2.5) node[text=black] (BD2) {\tiny{1}};
        \draw (10+0,-2.5) node[text=black] (BD3) {\tiny{2}};
        \draw (10+3,-2.5) node[text=black] (BD4) {};
        \draw (10+4,-2.5) node[text=black] (BD5) {};
        \draw (10+5.5,-2.5) node[text=black] (BD6) {\tiny{$\tow(k{-}1,n){-}1$}};

        \draw (10+-2.25,-1.4) node[text=black] (BE1) {\tiny{\text{type} $k{-}1$}};
        \draw (10+5,-1.52) node[text=black] (BE2) {\tiny{\text{type} $k{-}1$}};

        \foreach \i/\j in {BA/BB1, BA/BB2, BA/BB3, BA/BB4, BA/BB5, BA/BB6, BA/BC1, BA/BC3}
                 {\draw[-latex'] (\i) -- (\j);}

             
      \draw (6.5,2) node[minirond, jaune] (C) {};
      \draw (6.5, 2.5) node[text=black] (AE2) {\tiny{\text{type} $k{+}1$}};
      
      \draw (3,0) node[text=black] (X) {};
      
      \draw (4.75,0) node[text=black] (XX) {$\ldots$};

      \draw (6,0) node[text=black] (Y) {};
      
      \draw (7.5,0) node[text=black] (YY) {$\ldots$};
      
      \draw (9,0) node[text=black] (Z) {};

      \foreach \i/\j in {C/AA, C/BA, C/X, C/Y, C/Z} {\draw[-latex'] (\i) -- (\j);}   
                 
      \end{scope}
      \path[use as bounding box]  (0,0);

    \end{tikzpicture}  
  \end{center}
In order to enforce that a node is of type $k  \geq 1$ (this is a trivial property for $k=0$), 
and therefore has exactly $\tow(k,n)$ children, 
with each node $\anode$ of type $k \geq 0$ is associated a \defstyle{number} in $\interval{0}{\tow(k+1,n)-1}$.
Such a number is written $\semnumber{\atreemodel}{\anode}$. 
The subscript `$\atreemodel$' may be omitted when the context is clear.
Similarly, when 
 $\atreemodel, \anode_0 \models \bindd{\anominal_1}{1} \wedge 
\bigwedge_{i \in \interval{2}{d}}
\att{\anominal_1}{1} \att{\anominal_2}{1} \cdots \att{\anominal_{i-1}}{1} \bindd{\anominal_i}{1}$
and
the nodes  $\anode_1, \ldots, \anode_d$ are associated with $\anominal_1, \ldots \anominal_d$,
we write $\semnumber{\atreemodel}{\anominal_i}$ instead of $\semnumber{\atreemodel}{\anode_i}$ 
for all $i \in \interval{1}{d}$ 
(in general, by a slight abuse of notation, we may refer to a node
by its local nominal when it exists).

When the type of the node $\anode$ is zero, its number is defined as the unique $m$ 
such that the number represented by the truth values of  
$\avarprop_{n-1}, \avarprop_{n-2}, \ldots \avarprop_0$ is equal to $m$.
As usual, the propositional variable~$\avarprop_i$ is responsible for the $i$th bit of 
the number $m$ and by convention, the least significant bit is encoded by the truth value of $\avarprop_0$.
We illustrate the encoding below:
  
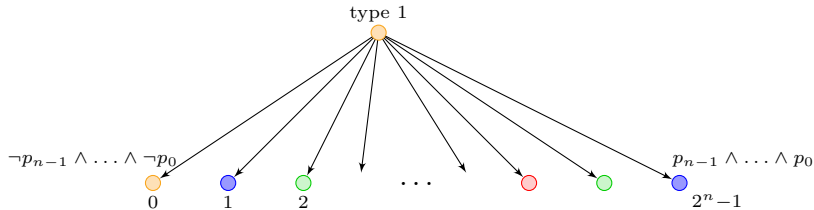
\begin{center}
    \begin{tikzpicture}
      
      \begin{scope}[line cap=round]

        \draw (1,0.25) node[text=black] (XXX) {\tiny{type $1$}};
        \draw (1,0) node[minirond, jaune] (A) {};
        \draw (-2,-2) node[minirond, jaune] (B1) {};
        \draw (-1,-2) node[minirond, bleu] (B2) {};
        \draw (0,-2) node[minirond, vert] (B3) {};
        \draw (0.75,-2) node[text=black] (C1) {};
        \draw (1.5,-2) node[text=black] (C2) {$\ldots$};
        \draw (2.25,-2) node[text=black] (C3) {};
        \draw (3,-2) node[minirond, rouge] (B4) {};
        \draw (4,-2) node[minirond, vert] (B5) {};
        \draw (5,-2) node[minirond, bleu] (B6) {};
        
        \draw (-2,-2.25) node[text=black] (D1) {\tiny{0}};
        \draw (-1,-2.25) node[text=black] (D2) {\tiny{1}};
        \draw (0,-2.25) node[text=black] (D3) {\tiny{2}};
        \draw (3,-2.25) node[text=black] (D4) {};
        \draw (4,-2.25) node[text=black] (D5) {};
        \draw (5.5,-2.25) node[text=black] (D6) {\tiny{$2^n{-}1$}};

        \draw (-2.8,-1.7) node[text=black] (E1) {\tiny{$\neg p_{n-1} \wedge \ldots \wedge \neg p_{0}$}};
        \draw (5.85,-1.7) node[text=black] (E2) {\tiny{$p_{n-1} \wedge \ldots \wedge p_0$}};

        \foreach \i/\j in {A/B1, A/B2, A/B3, A/B4, A/B5, A/B6, A/C1,  A/C3}
                 {\draw[-latex'] (\i) -- (\j);}
      \end{scope}
      \path[use as bounding box]  (0,0);

    \end{tikzpicture}
  \end{center}

Otherwise, when the type of $\anode$ is equal to some $k > 0$, its number is
represented by the binary encoding of
the propositional variable $\val{k-1}$ on its children assuming that there are $\tow(k,n)$ children
whose respective (bit) numbers span all over $\interval{0}{\tow(k,n)-1}$ and therefore all the children are implicitly ordered. 
This principle makes sense conceptually but it remains to express it in \tQCTLEX{}, similarly
to the \tower-hardness proof in~\cite{Prattetal19} for the fluted fragment in which counters
with high values have to be enforced within a restricted language (see also~\cite{Stockmeyer74}). 
That is why, in Table~\ref{table-auxiliary-formulae},  
we present a list of formulae to be defined. 
All of them are interpreted on a node~$\anode_0$ of type $k \geq 0$, $1 \leq d \leq k$.

\begin{table}[h]
    \begin{tabular}{|c|c|}
        \hline
         Formulae to be defined & Intuitive meaning \\ \hline
        $\phitype{k}$ & $\anode_0$ is of type $k$  \\ \hline
        $\phifirst{k}$ &  $\semnumber{}{\anode_0} = 0$ \\ \hline
        $\philast{k}$ &  $\semnumber{}{\anode_0} = \tow(k+1,n)-1$ \\ \hline
        $\phiunique{k}$ &  $\forall \anode, \anode' \ ((\anode_0 E \anode) \ \& \  (\anode_0 E \anode') \ \& \ 
        \anode \neq \anode') \ \rightarrow $ \\
        & $\semnumber{}{\anode} \neq \semnumber{}{\anode'}$ \\ \hline
        $\phipopulate{k}$ &   $\forall \anode \ ((\anode_0 E \anode) \ \& \ \semnumber{}{\anode} < \tow(k,n)-1)
                              \ \rightarrow$\\
       & $\exists \ \anode'  \  (\anode_0 E \anode') \ \wedge \ \semnumber{}{\anode'} = \semnumber{}{\anode} + 1$ \\ \hline
        $\eqk{k}{\bar{\anominal}}{\bar{\anominalbis}}$ & 
        $\semnumber{}{\anode_d} = \semnumber{}{\anode_d'}$\\ \hline
        $\succk{k}{\bar{\anominal}}{\bar{\anominalbis}}$  & 
         $\semnumber{}{\anode_d'} = 1 + \semnumber{}{\anode_d}$\\ \hline
        $\gk{k}{\bar{\anominal}}{\bar{\anominalbis}}$ & 
        $\semnumber{}{\anode_d} < \semnumber{}{\anode_d'}$\\ \hline
    \end{tabular}
\caption{Family of auxiliary formulae.}
\label{table-auxiliary-formulae}
\end{table}
In the last 3 lines of the table, the subscript
 `$k$' below `$=$' and `$<$' allows us
to remember that the formula is evaluated at a node of type $k$; as $\bar{\anominal}$ and  $\bar{\anominalbis}$
are of length $d$, the number comparison is done on nodes of type $k-d$ and the numbers can take 
values in $\interval{0}{\tow(k-d+1,n)-1}$.  Though most of the intuitive meanings are straightforward,
let us notice that $\phiunique{k}$ is  intended to express that two distinct children of $\anode_0$ have
distinct numbers. Similarly, $\phipopulate{k}$ is intended to express that 
if a child of $\anode_0$ has a number $n$ strictly less than 
$\tow(k,n)-1$, then $\anode_0$ has also another child with number equal to $n+1$
(`compl' in $\phipopulate{k}$ stands for `complete').

In what follows we will also employ $\widehat{\bar{\anominal},\bar{\anominalbis}}$ to denote the formula
\[
\bindd{\anominal_1}{1} \wedge \bindd{\anominalbis_1}{1} \wedge \bigwedge_{i \in \interval{2}{d}}  
(\att{\anominal_1}{1} \att{\anominal_2}{1} \cdots \att{\anominal_{i-1}}{1} \bindd{\anominal_i}{1}
\wedge 
\att{\anominalbis_1}{1} \att{\anominalbis_2}{1} \cdots \att{\anominalbis_{i-1}}{1} \bindd{\anominalbis_i}{1}).
\]
Assuming that the nodes $\anode_1, \ldots, \anode_d$ are associated with $\bar{\anominal} = \anominal_1, \ldots, \anominal_d$
(resp.  $\anode_1', \ldots, \anode_d'$ are associated with $\bar{\anominalbis} = \anominalbis_1, \ldots, \anominalbis_d$),
$\anode_0, \anode_1, \ldots, \anode_d$ and $\anode_0, \anode_1', \ldots, \anode_d'$ can be understood
as two branches rooted at $\anode_0$ ending at the node $\anode_d$ and at the node $\anode_d'$ respectively.
The formula $\widehat{\bar{\anominal},\bar{\anominalbis}}$ uses subformulae introduced in Section~\ref{sec:local-nominals}
and the wide hat symbol in $\widehat{\bar{\anominal},\bar{\anominalbis}}$ above $\bar{\anominal}$ and  $\bar{\anominalbis}$ 
is a graphical reminder of these two branches. 
By contrast, the specific  formula $\widehat{\bar{\anominal},\bar{\anominal}}$ states the existence of a single branch 
with nodes $\anode_0, \ldots, \anode_d$.

In order to define $\phitype{k}$ ($k \geq 1$), we specify that every child is of type $k-1$,
there is a child with number equal to zero, and if a child has number $m < \tow(k,n)-1$, then 
there is a child with number equal to $m+1$. 
Moreover, two distinct children have distinct numbers in $\interval{0}{\tow(k,n)-1}$.
Satisfying these conditions  guarantees that
the number for the children span all over $\interval{0}{\tow(k,n)-1}$. 
The formula $\phitype{k}$ (for $k \geq 1$) is defined as
\[ \phitype{k} \deff  \AX(\phitype{k-1}) \wedge \EX(\phifirst{k-1}) \wedge \phiunique{k} \wedge \phipopulate{k}. \]
Note that the above formula is intended to be built 
over the propositional variables $\avarprop_0, \ldots, \avarprop_{n-1}, \val{}$ (only).  

Let us first explain how we proceed to define all the mentioned formulae.
For successive values~$N \in \Nat$, we define inductively the formulae:
\begin{itemize}
\item $\phitype{N}$, $\phifirst{N}$ and $\philast{N}$, 
\item $\eqk{k}{\anominal_1, \ldots, \anominal_d}{\anominalbis_1, \ldots, \anominalbis_d}$,
\item $\gk{k}{\anominal_1, \ldots, \anominal_d}{\anominalbis_1, \ldots, \anominalbis_d}$,
\item $\succk{k}{\anominal_1, \ldots, \anominal_d}{\anominalbis_1, \ldots, \anominalbis_d}$
\end{itemize}
for all $k-d = N-1$. 
For $N = 0$, only the formulae  $\phitype{0}$, $\phifirst{0}$ and $\philast{0}$ make sense.
The case  $N = 1$ is not yet an instance of the general case. 
We first treat the cases for $N \in \set{0,1}$ and then we proceed with the general case $N \geq 2$.

\subsection{Formulae for types zero and one}
\label{section-Nzeroone}
When $k= 0$ (thus $k = N = 0$),  only the intended properties for the formulae
 $\phitype{0}$, $\phifirst{0}$ and $\philast{0}$ are meaningful.
Let us define them, in the simplest way.
\cut{
\begin{itemize}
  \item $\phitype{0} \egdef \top$
  \item $\phifirst{0} \egdef \neg \avarprop_{n-1} \wedge \cdots \wedge \neg \avarprop_{0}$ and
  \item $\philast{0} \egdef \avarprop_{n-1} \wedge \cdots \wedge \avarprop_{0}$. 
\end{itemize}
}
$$
\phitype{0} \egdef \top \ \ \ \ \ \ 
\phifirst{0} \egdef \neg \avarprop_{n-1} \wedge \cdots \wedge \neg \avarprop_{0} \ \ \ \ \ \ 
\philast{0} \egdef \avarprop_{n-1} \wedge \cdots \wedge \avarprop_{0}.
$$

It can be rapidly checked that $\atreemodel, \anode \models \phitype{0}$ iff $\anode$ is of type 0.
Moreover, assuming that $\anode$ is understood as a node of type $0$, we have  $\atreemodel, \anode \models \phifirst{0}$ iff $\semnumber{\atreemodel}{\anode} = 0$ and  
$\atreemodel, \anode \models \philast{0}$ iff  $\semnumber{\atreemodel}{\anode} = \tow(1,n)-1 = 2^n-1$. 

We next focus on the case when $N = 1$ and we define the formulae
$\phitype{k}$, $\phifirst{k}$ and $\philast{k}$
with $k {=} 1$ (\ie when $k {=} N {=} 1$) as well as $\gk{k}{\anominal_1, \ldots, \anominal_d}{\anominalbis_1, \ldots, \anominalbis_d}$
and $\succk{k}{\anominal_1, \ldots, \anominal_d}{\anominalbis_1, \ldots, \anominalbis_d}$
with $k-d = 0$ (that is when $k{-}d = N{-}1$  with $N{=}1$). 
We stress that $k$ and~$d$ can be arbitrarily large as long as $k = d$.

To start with the formula $\eqk{k}{\anominal_1, \ldots, \anominal_d}{\anominalbis_1, \ldots, \anominalbis_d}$, 
it can be easily defined in terms of $\gk{k}{\anominal_1, \ldots, \anominal_d}{\anominalbis_1, \ldots, \anominalbis_d}$ as:
\[\neg \left( \gk{k}{\anominal_1, \ldots, \anominal_d}{\anominalbis_1, \ldots, \anominalbis_d} \right)  \wedge \neg \left( \gk{k}{\anominalbis_1, \ldots, \anominalbis_d}{\anominal_1, \ldots, \anominal_d} \right). \]

Second, we turn our attention to the formula  $\phitype{1}$.
It states that there is a child with number equal to zero, if a child has number $m < 2^n-1$, then there is a child with number equal to $m+1$, all the children are of type $0$, 
and  two distinct children have distinct numbers in $\interval{0}{2^n-1}$.
Remember that the number of each child (of type 0) is computed from the propositional variables in $\set{\avarprop_{n-1}, \ldots, \avarprop_0}$.
The arithmetical reasoning between children, leading to the fact that there are exactly $2^n$ children whose numbers
span all over $\interval{0}{2^n-1}$ takes advantage of standard arithmetical properties on numbers encoded by $n$ bits. 
Here is the formula $\phitype{1}$: 
\[ \phitype{1} \egdef \AX(\phitype{0}) \wedge \EX(\phifirst{0}) \wedge \phiunique{1} \wedge \phipopulate{1}. \]

It remains to specify what exactly the formulae $\phiunique{1}$ and $\phipopulate{1}$ are. 
In order to define $\phiunique{1}$, responsible for enforcing the uniqueness among the children's numbering,
we simply state that there are no two distinct children (of type $0$) having the same number:
\[\forall{\anominal,\anominalbis} \ \distinctbindd{\anominal,\anominalbis}{1} \rightarrow 
\neg (\eqk{1}{\anominal}{\anominalbis}).\]
Note that the formula $\distinctbindd{\anominal,\anominalbis}{1}$ guarantees that we pick two distinct children
and the nominals $\anominal$ and $\anominalbis$ allow us to access them (and check the values of the
propositional variables in~$\set{\avarprop_{n-1}, \ldots, \avarprop_0}$). Consult Lemma~\ref{lemma:distinct-nominals} for the correctness.

The formula $\phipopulate{1}$ below states that for each child $\anode$ (of type $0$) that is not the last one (\ie does not have the highest possible number among all other nodes),  
there is also a child $\anode'$ (also of type~$0$)  such that $\semnumber{\atreemodel}{\anode'} = \semnumber{\atreemodel}{\anode} +1$. 
Here is the formula $\phipopulate{1}$: 
\[
\forall{\anominal} \ (\bindd{\anominal}{1} \wedge \att{\anominal}{1}(\neg \philast{0})) \rightarrow
\exists{\anominalbis} \ \bindd{\anominalbis}{1} \wedge \succk{1}{\anominal}{\anominalbis}. \]

Finally, it remains to define the formulae $\eqk{1}{\anominal}{\anominalbis}$ and $\succk{1}{\anominal}{\anominalbis}$
used respectively in $\phiunique{1}$ and in $\phipopulate{1}$. 
Below, we treat the more general situation with $k = d$ ($k$ is not necessarily equal to 1), and 
$\eqk{1}{\anominal}{\anominalbis}$ and $\succk{1}{\anominal}{\anominalbis}$ are  specific instances with $k=d=1$.
Let assume that $\bar{\anominal} = \anominal_1, \ldots, \anominal_k$ and $\bar{\anominalbis} = 
\anominalbis_1, \ldots, \anominalbis_k$ (thus $d = k$).
The forthcoming definitions are standard and rely on elementary operations on binary encoding of natural numbers  with $n$ bits 
(again, the least significant
bit is represented by the truth value of $\avarprop_0$):
\begin{itemize}
\itemsep 0 cm    
\item $\succk{k}{\anominal_1, \ldots, \anominal_k}{\anominalbis_1, \ldots, \anominalbis_k}$
is defined as
\[
\bigvee_{i=0}^{n-1} 
  \left(
  \underbrace{
     \att{\bar{\anominal}}{} 
  \left(
  (\neg \avarprop_i \wedge \bigwedge_{j=0}^{i-1} \avarprop_j)
 \right)
  }_{\text{look for the first zero bit}}
    \wedge 
     \underbrace{
        \att{\bar{\anominalbis}}{} 
    \left(
    \bigwedge_{j=0}^{i-1} \neg \avarprop_j \wedge \avarprop_i
    \right)
     }_{\text{reset previous bits, set } \avarprop_i}
      \wedge 
    \underbrace{ 
     \left( 
    \bigwedge_{j=i+1}^{n-1}
      \att{\bar{\anominal}}{} \avarprop_j \Iff  \att{\bar{\anominalbis}}{} 
\avarprop_j
      \right)
      }_{\text{rewrite other bits}}
   \right)
\]
\item $\gk{k}{\anominal_1, \ldots, \anominal_k}{\anominalbis_1, \ldots, \anominalbis_k}$
is defined as
\[
\bigvee_{i=0}^{n-1} 
\left(
\underbrace{\att{\bar{\anominalbis}}{} \avarprop_i \wedge 
\att{\bar{\anominal}}{} \neg \avarprop_i \wedge}_{\text{find the first bit that differs}}
     \underbrace{
    (\bigwedge_{j=i+1}^{n-1}
      \att{\bar{\anominal}}{} \avarprop_j \Iff \att{\bar{\anominalbis}}{} \avarprop_j)
     }_{\text{equality of more significant bits}}
\right) 
\]
\end{itemize}

For the sake of completeness, we define $\phifirst{1} \egdef \AX (\neg \val{0})$ and $\philast{1} \egdef \AX (\val{0})$. 

The lemma below states that we have properly proceeded for the
binary encoding of numbers with the variables in $\avarprop_{n-1}, \ldots,
\avarprop_{0}$ (and Lemmas~\ref{lemma:atoflengthd}--\ref{lemma:distinct-nominals} need to be used). 
 
\begin{lem}\label{lemma:kd-comparison-zero}
Let $\atreemodel$ be a tree model and $\anode$ be one of its nodes such that $\anode$ satisfies $\widehat{\bar{\anominal},\bar{\anominalbis}}$ ($\bar{\anominal}$ 
and $\bar{\anominalbis}$ are both of length $k$) and, $\anode_k$ and $\anode_k'$ are  understood as nodes of type $0$. 
\begin{description}
\item[(I)] $\atreemodel, \anode \models \succk{k}{\bar{\anominal}}{\bar{\anominalbis}}$ iff $\semnumber{\atreemodel}{\anode_k'} = 1 + \semnumber{\atreemodel}{\anode_k}$.
\item[(II)] $\atreemodel, \anode \models \gk{k}{\bar{\anominal}}{\bar{\anominalbis}}$ iff $\semnumber{\atreemodel}{\anode_k} < \semnumber{\atreemodel}{\anode_k'}$.
\item[(III)] $\atreemodel, \anode \models \eqk{k}{\bar{\anominal}}{\bar{\anominalbis}}$ iff $\semnumber{\atreemodel}{\anode_k} = \semnumber{\atreemodel}{\anode_k'}$.
\end{description}
\end{lem}
\begin{proof}
By careful inspection of the presented formulae, \cf Appendix~\ref{appendix:proof-of:lemma:kd-comparison-zero}.
\end{proof}

We conclude by presenting the main lemma that gathers all established formulae.

\begin{lem}\label{lemma:typeone}
Let $\atreemodel$ be a tree model and let $\anode$ be any of its nodes. 
The following hold:
\begin{description}
\item[(I)] $\atreemodel, \anode \models \phitype{1}$ iff $\anode$ is of type $1$,
\item[(II)] Assuming $\anode$ satisfies $\phitype{1}$, we have $\atreemodel, \anode \models \phifirst{1}$ iff $\semnumber{\atreemodel}{\anode} = 0$.
\item[(III)] Assuming $\anode$ satisfies $\phitype{1}$, we have $\atreemodel, \anode \models \philast{1}$ iff $\semnumber{\atreemodel}{\anode} = \tow(2,n)-1$.
\end{description}
\end{lem}
\begin{proof}
The properties (II)--(III) follow immediately from the way we encode numbers.
Check Appendix~\ref{appendix:proof-of:lemma:typeone} for a more detailed explanation.
\end{proof}

\subsection{Formulae with arbitrary \texorpdfstring{$N \geq 2$}{N >= 2}}

Let us consider the arbitrary case $N \geq 2$. 
Below, we define the formulae  $\phitype{N}$, $\phifirst{N}$ and~$\philast{N}$
as well as $\gk{k}{\anominal_1, \ldots, \anominal_d}{\anominalbis_1, \ldots, \anominalbis_d}$,
and $\succk{k}{\anominal_1, \ldots, \anominal_d}{\anominalbis_1, \ldots, \anominalbis_d}$ with~$k{-}d = N{-}1$.
We also consider the formula $\eqk{k}{\anominal_1, \ldots, \anominal_d}{\anominalbis_1, \ldots, \anominalbis_d}$, that is defined as follows:
\[\neg (\gk{k}{\anominal_1, \ldots, \anominal_d}{\anominalbis_1, \ldots, \anominalbis_d}) \wedge 
\neg (\gk{k}{\anominalbis_1, \ldots, \anominalbis_d}{\anominal_1, \ldots, \anominal_d}).\]
We assume that for all $k < N$, the formulae $\phitype{k}$, $\philast{k}$ and $\phifirst{k}$
are already defined and for  $k-d \leq N-2$,
the formulae $\gk{k}{\bar{\anominal}}{\bar{\anominalbis}}$,  $\eqk{k}{\bar{\anominal}}{\bar{\anominalbis}}$ 
and  $\succk{k}{\bar{\anominal}}{\bar{\anominalbis}}$ are already defined too ($\bar{\anominal}$ and 
$\bar{\anominalbis}$ are of length $d$). This can be understood as an implicit induction hypothesis
when proving the correctness of the formulae built for $N$. 

As for the case $N=1$, the formula $\phitype{N}$ follows the general schema: 
it states that there is a child with number equal to zero, 
if a child has number $m < \tow(N,n)-1$, then there is a child with number equal to $m+1$, 
and two distinct children have distinct numbers in $\interval{0}{\tow(N,n)-1}$.
Of course, all the children are enforced to be of type $N-1$. 
We present the claimed formula $\phitype{N}$ below.
\[
\phitype{N} \deff \AX(\phitype{N-1}) \wedge
\EX(\phifirst{N-1}) \wedge \phiunique{N} \wedge \phipopulate{N}.
\]

Again, it remains to specify what $\phiunique{N}$ and $\phipopulate{N}$ are. 
In order to define $\phiunique{N}$, we simply  state that there are no
two distinct children (of type $N-1$) with the same number:
\[
\phiunique{N} \deff 
\forall{\anominal,\anominalbis} \ \distinctbindd{\anominal,\anominalbis}{1} \rightarrow \neg  (\eqk{N}{\anominal}{\anominalbis}).
\]
Again, the formula $\distinctbindd{\anominal,\anominalbis}{1}$ allows us to select two distinct children (of type $N-1$).
The formula $\phipopulate{N}$ below states that for each  child $\anode$ (of type $N-1$) that is not the last one,  
there is also a child $\anode'$ (of type $N-1$ too) such that $\semnumber{\atreemodel}{\anode'} = \semnumber{\atreemodel}{\anode} +1$. 
Here it is: 
\[
\phipopulate{N} \egdef 
\forall{\anominal} \ (\bindd{\anominal}{1} \wedge \att{\anominal}{1}(\neg \philast{N-1})) \rightarrow
\exists{\anominalbis} \ \bindd{\anominalbis}{1} \wedge \succk{N}{\anominal}{\anominalbis}.
\]
It remains to define $\eqk{N}{\anominal}{\anominalbis}$ and $\succk{N}{\anominal}{\anominalbis}$
used respectively in $\phiunique{N}$ and in~$\phipopulate{N}$. This time, this requires 
much lengthier developments, apart from using the properties of the formulae constructed for $N-1$ and for smaller values
(implicit induction hypothesis). 
Below, we treat the more general situation with $k-d = N-1$, and 
$\eqk{N}{\anominal}{\anominalbis}$ and $\succk{N}{\anominal}{\anominalbis}$ are just particular instances with $k=N$ and $d = 1$. 
Thus, let $\bar{\anominal} = \anominal_1, \ldots, \anominal_d$ and~$\bar{\anominalbis} = 
\anominalbis_1, \ldots, \anominalbis_d$.

For defining $\succk{k}{\anominal_1, \ldots, \anominal_k}{\anominalbis_1, \ldots, \anominalbis_k}$ 
(see Section~\ref{section-Nzeroone}), we have compared the respective truth values of 
the propositional variables $\avarprop_{n-1}, \ldots, \avarprop_{0}$ 
for the node $\anode_k$ (the interpretation of  $\anominal_k$) and for the node $\anode'_k$ 
(the interpretation of $\anominalbis_k$). 
The same principle applied for defining~$\eqk{k}{\anominal_1, \ldots, \anominal_k}{\anominalbis_1, \ldots, \anominalbis_k}$.
Typically, $\succk{k}{\anominal_1, \ldots, \anominal_k}{\anominalbis_1, \ldots, \anominalbis_k}$ holds iff
there is an index~$i \in \interval{0}{n-1}$, such that
\begin{itemize}
\itemsep 0 cm 
\item for every $j \in \interval{i+1}{n-1}$, $\anode_k$ and $\anode_k'$ agree on $\avarprop_j$,
\item $\anode_k$ does not satisfy $\avarprop_i$ and  $\anode_k'$  satisfies $\avarprop_i$,
\item for $j \in \interval{0}{i-1}$, $\anode_k$  satisfies $\avarprop_j$, $\anode_k'$  does not satisfy $\avarprop_j$. 
\end{itemize}
Thus, we needed to define a partition $\set{\interval{i+1}{n-1}, \set{i}, \interval{0}{i-1}}$ 
of $\interval{0}{n-1}$ (understood as the set of bit numbers to encode a value in $\interval{0}{2^n-1}$).
The same   principle  applies when the bit numbers are among $\interval{0}{\tow(k{-}d,n) - 1}$ to encode a value in  
$\interval{0}{\tow(k{-}d{+}1,n) - 1}$. This needs to be done concisely as we cannot go through all the $\tow(k{-}d,n)$ bit numbers 
 because the whole reduction has to be of elementary complexity. 

Now, when attempting to define  $\succk{k}{\anominal_1, \ldots, \anominal_d}{\anominalbis_1, \ldots, \anominalbis_d}$,
the nodes $\anode_d$ and $\anode_d'$ are of type $k-d > 0$ with $\tow(k-d,n)$ children each. The truth values of $\val{}$ on their 
respective children determine precisely the numbers $\semnumber{}{\anode_d}$ and $\semnumber{}{\anode_d'}$.
Below, we describe what $\anode_d$'s children look like.
\begin{center}
    \begin{tikzpicture}
      
      \begin{scope}[scale=.9,line cap=round]
      

        \draw (1,0) node[minirond, jaune] (AA) {};
        \draw (-2,-2) node[minirond, jaune] (AB1) {};
        \draw (-1,-2) node[minirond, bleu] (AB2) {};
        \draw (0,-2) node[minirond, vert] (AB3) {};
        \draw (0.75,-2) node[text=black] (AC1) {};
        \draw (1.5,-2) node[text=black] (AC2) {$\ldots$};
        \draw (2.25,-2) node[text=black] (AC3) {};
        \draw (3,-2) node[minirond, rouge] (AB4) {};
        \draw (4,-2) node[minirond, vert] (AB5) {};
        \draw (5,-2) node[minirond, bleu] (AB6) {};
        
        \draw (-2,-2.5) node[text=black] (AD1) {\small{0}};
        \draw (-2,-3) node[text=black] (Text1) {\small{\textit{val}}};

        \draw (-1,-2.5) node[text=black] (AD2) {\small{1}};
        \draw (-1,-3) node[text=black] (Text1) {\small{$\neg$\textit{val}}};

        \draw (0,-2.5) node[text=black] (AD3) {\small{2}};
        \draw (0,-3) node[text=black] (Text1) {\small{$\neg$\textit{val}}};

        \draw (3,-2.5) node[text=black] (AD4) {};
        \draw (3,-3) node[text=black] (Text1) {\small{$\ldots$}};

        \draw (4,-2.5) node[text=black] (AD5) {};        
        \draw (4,-3) node[text=black] (Text1) {\small{\textit{val}}};

        \draw (5.5,-2.5) node[text=black] (AD6) {\small{$\tow(k{-}d,n){-}1$}};
        \draw (5.5,-3) node[text=black] (Text1) {\small{\textit{val}}};

        \draw (-2.7,-1.52) node[text=black] (AE1) {\small{\text{type} $k{-}d{-}1$}};
        \draw (5.5,-1.52) node[text=black] (AE2) {\small{\text{type} $k{-}d{-}1$}};

        \foreach \i/\j in {AA/AB1, AA/AB2, AA/AB3, AA/AB4, AA/AB5, AA/AB6, AA/AC1, AA/AC3}
                 {\draw[-latex'] (\i) -- (\j);}

      \end{scope}
      \path[use as bounding box]  (0,0);

    \end{tikzpicture}  
  \end{center}

Let $\anodebis_{\tow(k-d,n)-1}, \ldots, \anodebis_0$ be the children of $\anode_d$ such that $\semnumber{}{\anodebis_j} = j$ for all $j$.
Similarly, let $\anodebis_{\tow(k-d,n)-1}', \ldots, \anodebis_0'$ be the children of $\anode_k'$ such that
$\semnumber{}{\anodebis_j'} = j$ for all $j$. Hence, $\succk{k}{\anominal_1, \ldots, \anominal_d}{\anominalbis_1, \ldots, \anominalbis_d}$ holds iff there is a position $i \in \interval{0}{\tow(k-d,n)-1}$ satisfying
\begin{itemize}
\item for  $j \in \interval{i+1}{\tow(k-d,n)-1}$, $\anodebis_j$ and $\anodebis_j'$ agree on $\val{}$,
\item $\anodebis_i$ does not satisfy $\val{}$ and  $\anodebis_i'$  satisfies $\val{}$,
\item for every  $j \in \interval{0}{i-1}$, $\anodebis_j$  satisfies $\val{}$ and $\anodebis_j'$  does not satisfy $\val{}$. 
\end{itemize}
We have to define a partition $\set{\anodebis_{i+1}, \ldots, \anodebis_{\tow(k-d,n)-1}}, \set{\anodebis_i}, \set{\anodebis_0, \ldots, \anodebis_{i-1}}$ of  {\small $\set{\anodebis_{\tow(k-d,n)-1}, \ldots, \anodebis_0}$} (and similarly for the children 
of $\anode_d'$). To do so, we employ the existential quantification
on the fresh propositional variables $\mathsf{l}$ (left), $\mathsf{s}$ (selected bit), $\mathsf{r}$ (right) such that 
\begin{enumerate}
\itemsep 0 cm 
\item[(a)] for every child of $\anode_d$ (resp. $\anode_d'$), exactly one propositional
variable among $\set{\mathsf{l},\mathsf{s},\mathsf{r}}$ holds true,
\item[(b)] exactly one child of $\anode_d$ (resp. $\anode_d'$) satisfies $\mathsf{s}$,
\item[(c)] if $\anode$ is a child of $\anode_d$  satisfying $\mathsf{l}$ (resp. $\mathsf{s}$)
and $\anode'$ is child of $\anode_d$  satisfying $\mathsf{s}$ (resp. $\mathsf{r}$), then $\semnumber{\atreemodel}{\anode} < 
\semnumber{\atreemodel}{\anode'}$. The same condition holds with $\anode_d'$. 
\end{enumerate}
Below, we illustrate how the propositional variables 
$\mathsf{l}$, $\mathsf{s}$ and $\mathsf{r}$ are distributed. 

 \begin{center}
    \begin{tikzpicture}
      
      \begin{scope}[scale=.9,line cap=round]

        \draw (5,1) node[minirond, vert] (Root) {};


        \draw (1,0) node[minirond, jaune] (AA) {};
        \draw (-2,-2) node[minirond, jaune] (AB1) {};
        \draw (-1,-2) node[minirond, bleu] (AB2) {};

        \draw (0,-2) node[text=black] (AC1) {$\ldots$};
        
        \draw (1.5,-2) node[minirond, vert]  (AC2) {};
        \draw (1.4,-2) node[]  (AC2l) {};
        \draw (1.6,-2) node[]  (AC2r) {};

        \draw (3,-2) node[text=black] (AC3) {$\ldots$};

        \draw (4,-2) node[minirond, jaune] (AB5) {};
        \draw (5,-2) node[minirond, bleu] (AB6) {};
        
        \draw (-2,-2.3) node[text=black] (AD1) {$\small{\mathsf{l}}$};
        \draw (-1,-2.3) node[text=black] (AD2) {$\small{\mathsf{l}}$};
        \draw (1.5,-2.3) node[text=black] (AD3) {$\small{\mathsf{s}}$};
        \draw (4,-2.3) node[text=black] (AD4) {$\small{\mathsf{r}}$};
        \draw (5,-2.3) node[text=black] (AD5) {$\small{\mathsf{r}}$};

        \foreach \i/\j in {AA/AB1, AA/AB2, AA/AB3, AA/AB4, AA/AB5, AA/AB6, AA/AC2}
                 {\draw[-latex'] (\i) -- (\j);}

        \scoped[on background layer] \filldraw [blue!10, line width=2.1em, line join=round,] (AB1.center) -- (AC1.center) -- cycle;
        \scoped[on background layer] \filldraw [vert!10, line width=2.1em, line join=round,] (AC2l.center) -- (AC2r.center) -- cycle;
        \scoped[on background layer] \filldraw [rouge!10, line width=2.1em, line join=round,] (AC3.center) -- (AB6.center) -- cycle;


        \draw (9,0) node[minirond, jaune] (BB) {};
        \draw (6,-2) node[minirond, jaune] (BB1) {};
        \draw (7,-2) node[minirond, bleu] (BB2) {};

        \draw (8,-2) node[text=black] (BC1) {$\ldots$};
        \draw (9.5,-2) node[minirond, vert]  (BC2) {};
        \draw (9.4,-2) node[]  (BC2l) {};
        \draw (9.6,-2) node[]  (BC2r) {};
        \draw (11,-2) node[text=black] (BC3) {$\ldots$};

        \draw (12,-2) node[minirond, jaune] (BB5) {};
        \draw (13,-2) node[minirond, bleu] (BB6) {};
        
        \draw (6,-2.3) node[text=black] (BD1) {$\small{\mathsf{l}}$};
        \draw (7,-2.3) node[text=black] (BD2) {$\small{\mathsf{l}}$};
        \draw (9.5,-2.3) node[text=black] (BD3) {$\small{\mathsf{s}}$};
        \draw (12,-2.3) node[text=black] (BD4) {$\small{\mathsf{r}}$};
        \draw (13,-2.3) node[text=black] (BD5) {$\small{\mathsf{r}}$};
        
        \foreach \i/\j in {BB/BB1, BB/BB2, BB/BB3, BB/BB4, BB/BB5, BB/BB6, BB/BC2}
                 {\draw[-latex'] (\i) -- (\j);}

        \draw[-latex'] (Root) -- (AA);
        \draw[-latex'] (Root) -- (BB);

        \scoped[on background layer] \filldraw [blue!10, line width=2.1em, line join=round,] (BB1.center) -- (BC1.center) -- cycle;
        \scoped[on background layer] \filldraw [vert!10, line width=2.1em, line join=round,] (BC2l.center) -- (BC2r.center) -- cycle;
        \scoped[on background layer] \filldraw [rouge!10, line width=2.1em, line join=round,] (BC3.center) -- (BB6.center) -- cycle;

      \end{scope}
      \path[use as bounding box]  (0,0);

    \end{tikzpicture}  
  \end{center}

Additional arithmetical constraints are needed to relate the partition
of $\bar{\anominal}$ with the partition of~$\bar{\anominalbis}$ (see below the details)
but in a way, it is independent of the partition itself. 
For instance, the unique child of $\anode_d$ satisfying $\mathsf{s}$ and the  unique child of $\anode_d'$ satisfying $\mathsf{s}$
should have the same (bit) number. 
Nevertheless, it is clear that we need, at least,
to be able 
to state in \tQCTLEX{} the existence of a partition satisfying the conditions (a), (b) and (c). 
In the sequel, such partitions are called \defstyle{$\mathsf{l}\mathsf{s}\mathsf{r}$-partitions}. 
The forthcoming formula $\lsr{k}{\bar{\anominal}}$ does the job for $\anode_d$ (then use $\lsr{k}{\bar{\anominalbis}}$ for $\anode_d'$). 

Take $\bar{\anominal} = \anominal_1, \ldots, \anominal_d$.
In the context of the definition of $\lsr{k}{\bar{\anominal}}$, 
we allow the limit case $d=0$, with empty sequence $\varepsilon$, assuming that $\att{\varepsilon}{} \aformulabis \egdef \aformulabis$
and $\widehat{\varepsilon,\varepsilon} \egdef \top$. 
Below $0 \leq d < k$ and the formula~$\lsr{k}{\bar{\anominal}}$  is defined as the conjunction $\lsrone{k}{\bar{\anominal}} \wedge 
\lsrtwo{k}{\bar{\anominal}} \wedge \lsrthree{k}{\bar{\anominal}}$ and 
is  interpreted on a node~$\anode_0$ of type $k$ satisfying
$\widehat{\bar{\anominal},\bar{\anominal}}$, and therefore this satisfaction is witnessed by the branch~$\anode_0, \ldots, \anode_d$
(notations for developments below). 

First, $\lsrtwo{k}{\bar{\anominal}} \egdef \att{\bar{\anominal}}{}(\EXOne(\mathsf{s}))$, 
with 
$\EXOne(\aformulabis)$ defined as  $\EX \aformulabis \wedge \neg \exists \ \avarprop \ 
(\EX(\aformulabis  \wedge \avarprop) \wedge  \EX(\aformulabis  \wedge \neg \avarprop))$ for a fresh $\avarprop$.
Note that the formula $\lsrtwo{k}{\bar{\anominal}}$ simply states that there is a unique child of $\anode_d$ satisfying $\mathsf{s}$.
Next, let $\lsrone{k}{\bar{\anominal}}$ be defined below, stating that for every child of $\anode_d$, exactly one 
propositional variable among $\set{\mathsf{l},\mathsf{s},\mathsf{r}}$ holds true:
\[
\lsrone{k}{\bar{\anominal}} \egdef 
 \att{\bar{\anominal}}{} 
\left( \AX((\mathsf{s} \vee \mathsf{l} \vee \mathsf{r})  \wedge 
 \neg (\mathsf{s} \wedge \mathsf{l}) \wedge 
\neg (\mathsf{s} \wedge \mathsf{r}) \wedge \neg (\mathsf{l} \wedge \mathsf{r}))
\right).
\]
Finally, $\lsrthree{k}{\bar{\anominal}}$ is defined as follows.
\[
\att{\bar{\anominal}}{} 
( 
\forall{w} \forall{w'} \ \distinctbindd{w,w'}{1} \wedge ((\att{w}{1}(\mathsf{s}) \wedge 
\att{w'}{1}(\mathsf{r})) \vee 
(\att{w}{1}(\mathsf{l}) \wedge \att{w'}{1}(\mathsf{s}))) \rightarrow 
\gk{k-d}{w'}{w}
).
\]
The formula $\lsrthree{k}{\bar{\anominal}}$ states if $\anode$ is a child of $\anode_d$ satisfying $\mathsf{l}$ (resp. $\mathsf{s}$)
and $\anode'$ is another child of $\anode_d$ satisfying $\mathsf{s}$ (resp. $\mathsf{r}$), then $\semnumber{\atreemodel}{\anode} < 
\semnumber{\atreemodel}{\anode'}$. The nodes $\anode$ and $\anode'$ are obviously of type $k-d-1$ and their respective
numbers belong to $\interval{0}{\tow(k-d,n)-1}$. It is important to observe that $\gk{k-d}{w'}{w}$ is well-defined recursively as soon as
$k-d \leq N-2$. 

\begin{lem} 
\label{lemma:lsr-partition}
Let $\atreemodel$ be a tree model, $\anode_0$ be a node of type $k \geq 0$, and 
 $\bar{\anominal}$ be a (possibly empty) sequence of nominals $\anominal_1, \ldots, \anominal_d$ for some $d \in \interval{0}{k-1}$
such that  $k-d \leq N-1$, $\atreemodel, \anode_0 \models \widehat{\bar{\anominal},\bar{\anominal}}$, and its witness branch is $\anode_0, \ldots, \anode_d$.
Then $\atreemodel, \anode_0 \models \lsr{k}{\bar{\anominal}}$ iff the conditions below hold:
\begin{enumerate}
\itemsep 0 cm 
\item[(a)] For every child of $\anode_d$, exactly one propositional
variable among $\set{\mathsf{l},\mathsf{s},\mathsf{r}}$ holds true.
\item[(b)] Exactly one child of $\anode_d$ satisfies $\mathsf{s}$. 
\item[(c)] If $\anode$ is a child of $\anode_d$ satisfying $\mathsf{l}$ (resp. $\mathsf{s}$)
and $\anode'$ is a child of $\anode_d$ satisfying $\mathsf{s}$ (resp. $\mathsf{r}$), then $\semnumber{\atreemodel}{\anode} < 
\semnumber{\atreemodel}{\anode'}$.
\end{enumerate}
\end{lem}
\begin{proof}
By careful inspection of the presented formulae, \cf Appendix~\ref{appendix:proof-of:lemma:lsr-partition}.
\end{proof}

We come back to the question of defining formulae expressing number comparisons. 
The formula $\succk{k}{\bar{\anominal}}{\bar{\anominalbis}}$ (remember $k-d=N-1$) is defined as the expression
\[
\exists \  \mathsf{l}, \mathsf{s}, \mathsf{r} \ 
\lsr{k}{\bar{\anominal}} \wedge \lsr{k}{\bar{\anominalbis}} \wedge
\aformula_{{\tiny \rm left}}(k) \wedge \aformula_{{\tiny \rm select}}(k) \wedge \aformula_{{\tiny \rm right}}(k).
\]

The conjunction $\aformula_{{\tiny \rm left}}(k) \wedge \aformula_{{\tiny \rm select}}(k) \wedge \aformula_{{\tiny \rm right}}(k)$ takes care of the arithmetical constraints. 
The formula $\aformula_{{\tiny \rm select}}(k)$ states that the $\anode_d$'s unique child  satisfying $\mathsf{s}$ (whose
number is the pivot bit $i$) does not satisfy $\val{}$, and the $\anode_d'$'s unique child  satisfying $\mathsf{s}$ (whose number is also the pivot bit $i$) satisfies $\val{}$:
\[
\aformula_{{\tiny \rm select}}(k)  \egdef 
\att{\bar{\anominal}}{} \left( \AX(\mathsf{s} \rightarrow \neg \val{k-d-1}) \right) \wedge 
\att{\bar{\anominalbis}}{} \left( \AX(\mathsf{s} \rightarrow  \val{k-d-1}) \right).
\]
The formula $\aformula_{{\tiny \rm right}}(k)$ states that for all the children of $\anode_d$ satisfying $\mathsf{r}$ (and therefore with bit number strictly smaller than $i$), the bit value is 1, and  for all the children of $\anode_d'$ satisfying
$\mathsf{r}$ (and therefore with bit number strictly smaller than $i$), the bit value is $0$. 
\[
\aformula_{{\tiny \rm right}}(k)  \egdef \att{\bar{\anominal}}{} \left( \AX(\mathsf{r} \rightarrow  \val{k-d-1}) \right) \wedge 
\att{\bar{\anominalbis}}{} \left( \AX(\mathsf{r} \rightarrow  \neg \val{k-d-1}) \right).
\]
The formula $\aformula_{{\tiny \rm left}}(k)$ states that the children of $\anode_d$ satisfying $\mathsf{l}$ induce a set of
bit numbers equal to the set of bit numbers induced by the children of $\anode_d'$ satisfying  $\mathsf{l}$. 
This entails also that the unique respective children of $\anode_d$ and $\anode_d'$ satisfying $\mathsf{s}$ have the same (bit) number. 
Moreover, we require that children with the same bit number satisfying $\mathsf{l}$ (taken from $\anode_d$ and from $\anode_d'$) have the same bit value witnessed by the truth value of $\val{}$. 
The formula $\aformula_{{\tiny \rm left}}(k)$ is equal to $\aformula_{{\tiny \rm left}}^{\bar{\anominal}, \bar{\anominalbis}}(k)
\wedge \aformula_{{\tiny \rm left}}^{\bar{\anominalbis}, \bar{\anominal}}(k)$  with $\aformula_{{\tiny \rm left}}^{\bar{\anominal}, \bar{\anominalbis}}(k)$ 
defined below:
\[
\aformula_{{\tiny \rm left}}^{\bar{\anominal}, \bar{\anominalbis}}(k) \egdef
\forall{w} \; \att{\bar{\anominal}}{}(\bindd{w}{1} \wedge \att{w}{1} (\mathsf{l})) \rightarrow
\]
\[
\big(\exists{w'} \; 
\att{\bar{\anominalbis}}{}(\bindd{w'}{1} \wedge \att{w'}{1} (\mathsf{l})) \wedge 
\eqk{k}{\bar{\anominal},w}{\bar{\anominalbis},w'} \wedge (\att{\bar{\anominal},w}{} \val{k-d-1} \leftrightarrow
\att{\bar{\anominalbis},w'}{} \val{k-d-1})\big).
\]
Note that $\eqk{k}{\bar{\anominal},w}{\bar{\anominalbis},w'}$ is well-defined as $k-(d+1) \leq N-2$. 
Below, we define the formula $\gk{k}{\bar{\anominal}}{\bar{\anominalbis}}$ with $k-d=N-1$, $\bar{\anominal} = \anominal_1, \ldots,\anominal_d$ and $\bar{\anominalbis} = \anominalbis_1, \ldots,\anominalbis_d$.
Based on previous developments and on standard arithmetical properties of numbers encoded in binary with $k-d$ bits, we define the formula $\gk{k}{\bar{\anominal}}{\bar{\anominalbis}}$ as the expression 
\[
\gk{k}{\bar{\anominal}}{\bar{\anominalbis}} \egdef 
\exists \ \mathsf{s}, \mathsf{l}, \mathsf{r} \ 
\lsr{k}{\bar{\anominal}} \wedge \lsr{k}{\bar{\anominalbis}} \wedge
\aformula_{{\tiny \rm left}}(k) \wedge \aformula_{{\tiny \rm select}}(k).
\]

As previously, the formula $\eqk{k}{\bar{\anominal}}{\bar{\anominalbis}}$  is defined as follows:
\[
\eqk{k}{\bar{\anominal}}{\bar{\anominalbis}}
\egdef
\neg (\gk{k}{\bar{\anominal}}{\bar{\anominalbis}})
\wedge
\neg (\gk{k}{\bar{\anominalbis}}{\bar{\anominal}}).
\]

\begin{lem} 
\label{lemma:kd-comparison}
Let $\atreemodel$ be a tree model, $\anode$ be a node satisfying $\widehat{\bar{\anominal},\bar{\anominalbis}}$ and, $\anode_{d}$, $\anode'_{d}$ are of type $k - d$. 
\begin{description}
\item[(I)] We have $\atreemodel, \anode \models \succk{k}{\bar{\anominal}}{\bar{\anominalbis}}$ iff $\semnumber{\atreemodel}{\anode_d'} = 1+ \semnumber{\atreemodel}{\anode_d}$.
\item[(II)] We have $\atreemodel, \anode \models \gk{k}{\bar{\anominal}}{\bar{\anominalbis}}$ iff $\semnumber{\atreemodel}{\anode_d} < \semnumber{\atreemodel}{\anode_d'}$.
\item[(III)] We have $\atreemodel, \anode \models \eqk{k}{\bar{\anominal}}{\bar{\anominalbis}}$ iff $\semnumber{\atreemodel}{\anode_d} = \semnumber{\atreemodel}{\anode_d'}$.
\end{description}
\end{lem}
\begin{proof}
By careful inspection of the presented formulae, \cf Appendix~\ref{appendix:proof-of:lemma:kd-comparison}.
\end{proof}
Mainly, this allows us to prove Lemma~\ref{lemma:typeN}(I).
\begin{lem}\label{lemma:typeN}
Let $N \geq 2$, $\atreemodel$ be a tree model and $\anode$ be one of its nodes.
\begin{description}
\item[(I)] $\atreemodel, \anode \models \phitype{N}$ iff $\anode$ is of type $N$,
\item[(II)] Assuming that $\anode$ satisfies $\phitype{N}$, we have $\atreemodel, \anode \models \phifirst{N}$ iff $\semnumber{\atreemodel}{\anode} = 0$. 
\item[(III)] Assuming that $\anode$ satisfies $\phitype{N}$, we have $\atreemodel, \anode \models \philast{N}$ iff $\semnumber{\atreemodel}{\anode} = \tow(N+1,n)-1$.
\end{description}
\end{lem}
\begin{proof}
Similar to the proof of Lemma~\ref{lemma:typeone}, see Appendix~\ref{appendix:proof-of:lemma:typeN} for more details.
\end{proof}

Consequently, for all $k \geq 0$, $\phitype{k}$, $\phifirst{k}$ and $\philast{k}$ 
characterise exactly the discussed properties, and similarly for the formulae
of form $\succk{k}{\bar{\anominal}}{\bar{\anominalbis}}$, 
$\gk{k}{\bar{\anominal}}{\bar{\anominalbis}}$,
$\eqk{k}{\bar{\anominal}}{\bar{\anominalbis}}$
where $\bar{\anominal}$ and $\bar{\anominalbis}$ are of length $d$ in $\interval{1}{k}$ and $k \geq 1$. 
It is natural to wonder what is the size of $\phitype{k}$, using a reasonably succinct encoding for formulae.
As the definition of $\phitype{k}$ requires the subformulae $\phitype{k-1}$,  
$\eqk{k}{\bar{\anominal}}{\bar{\anominalbis}}$ and $\succk{k}{\bar{\anominal}}{\bar{\anominalbis}}$
(the other subformulae  are of constant size), and the formula $\phitype{1}$ is quadratic in $n$, one can show that $\phitype{k}$ is of size $2^{\mathcal{O}(k+n)}$.
This is sufficient for our purposes.

\subsection{Uniform reduction leading to \tower-hardness}\label{section-uniform-reduction}
Let $(\cP, c)$, where $\cP = (\cT, \cH, \cV)$ and $c = t_0,t_1, \ldots, t_{n-1}$ 
be an instance of  $\mathtt{Tiling}_k$ known to be $\kNExpTime$-complete 
(see also~\cite[Chapter 11]{Demri&Goranko&Lange16}). We reduce the existence of a tiling $\tau: \interval{0}{\tow(k,n)-1} \times
\interval{0}{\tow(k,n)-1} \rightarrow \cT$, respecting the initial condition
and the horizontal and vertical matching conditions, 
to the satisfiability of a formula $\aformula$ in \tQCTLEX{}. 

To encode the grid  $\interval{0}{\tow(k,n)-1} \times \interval{0}{\tow(k,n)-1}$, we consider a root node $\aroot$ of type $k+1$, and we distinguish $\tow(k,n)$ children among all of its $\tow(k+1,n)$ children. 
Each child of $\aroot$ has itself exactly $\tow(k,n)$ children as it is a node of type $k$.
In order to identify the $\tow(k,n)$ first children of $\aroot$, we use an $\mathsf{l} \mathsf{s} \mathsf{r}$-partition
so that the unique child satisfying $\mathsf{s}$ has precisely the number $\tow(k,n)$. 
This guarantees that exactly the children of $\aroot$ whose numbers are in $\interval{0}{\tow(k,n)-1}$
satisfy $\mathsf{r}$.  So, the $\mathsf{l} \mathsf{s} \mathsf{r}$-partition is used in a new context.

Below, we define the new formula $\anothereqk{k}$  that expresses that 
a node of type $k$ has number $\tow(k,n)$ ($k \geq 1$).
We recall that a node of type $k$ takes its values in $\interval{0}{\tow(k+1,n)-1}$.
Let us provide an inductive definition for $\anothereqk{k}$ with the base case  $k = 1$. 
\begin{itemize}
\item $\anothereqk{0} \; \egdef \neg \avarprop_{n-1} \wedge \cdots \wedge \neg \avarprop_0$.
\item $\anothereqk{1} \egdef \AX (\val{0} \leftrightarrow \aset_n = n)$, where $\aset_n = n$ is an abbreviation for the formula
stating that the truth values for $\avarprop_{n-1}, \ldots, \avarprop_0$ encode $n$ in binary.
\item 
For $k \geq 2$,  $(\anothereqk{k}) \egdef \AX (\val{k-1} \leftrightarrow (\anothereqk{k-1}))$.
\end{itemize}

\begin{lem}\label{lemma:nb-eq-towkn}
Assume that $\atreemodel, \anode \models \phitype{k}$. Then
$\semnumber{\atreemodel}{\anode} = \tow(k,n)$ iff $\atreemodel, \anode \models (\anothereqk{k})$.  
\end{lem}
\begin{proof}
By careful inspection of the semantics, \cf Appendix~\ref{appendix:proof-of:lemma:nb-eq-towkn}.
\end{proof}

Let $\aformula_\cP$ be the formula built from the instance $(\cP = (\cT, \cH, \cV),c)$ as follows:
\[
\phitype{k+1} \wedge \exists \ \mathsf{l}, \mathsf{s}, \mathsf{r} \ 
\lsr{k+1}{\varepsilon} \wedge 
\EX(\mathsf{s} \wedge (\anothereqk{k})) \wedge 
\aformula_{\rm cov} \wedge
\aformula_{\rm init}   \wedge
\aformula_{\cH} \wedge
\aformula_{\cV}. 
\]
Definitions and explanations for $\aformula_{\rm cov}$, $\aformula_{\rm init}$, $\aformula_{\cH}$
and $\aformula_{\cV}$ follow but observe that an $\mathsf{l}\mathsf{s}\mathsf{r}$-partition
is performed for a node of type $k+1$ and exactly $\tow(k,n)$ children satisfy $\mathsf{r}$ thanks
to the satisfaction of the subformula $\EX(\mathsf{s} \wedge (\anothereqk{k}))$. 
The formula $\aformula_{\rm cov}$ states that every position in $\interval{0}{\tow(k,n)-1} \times \interval{0}{\tow(k,n)-1}$ has a unique tile:
\[
\forall \ \anominal, \anominalbis \ 
\bindd{\anominal}{1} \wedge \att{\anominal}{1} \mathsf{r} \wedge 
\att{\anominal}{1} \bindd{\anominalbis}{1} 
\rightarrow
\att{\anominal, \anominalbis}{} (\bigvee_{t \in \cT} t \wedge \bigwedge_{t \neq t' \in \cT} \neg (t \wedge t')). 
\]
The tile types in $\cT$ are understood as propositional variables. 
In order to access the root node $\aroot$ to a node  encoding a position of the grid,
one needs first to access a child $\anode$ of $\aroot$ satisfying $\mathsf{r}$ (and this is done 
with the help of the local nominal $\anominal$) and then to access any child $\anode'$ of $\anode$
(done  with  the local nominal $\anominalbis$). Then, to reason propositionally on $\anode'$, it is sufficient
to consider subformulae of the form $\att{\anominal, \anominalbis}{}  \aformulabis$. This principle
is applied to all the formulae below. 
The formula $\aformula_{\cH}$ defined below encodes the horizontal matching constraints.
\[
\forall \ \anominal, \anominal', \anominalbis, \anominalbis' \ 
(\bindd{\anominal}{1} \wedge \att{\anominal}{1} \mathsf{r} \wedge 
\bindd{\anominal'}{1} \wedge \att{\anominal'}{1} \mathsf{r} \wedge 
\att{\anominal}{1} \bindd{\anominalbis}{1} \wedge \att{\anominal'}{1} \bindd{\anominalbis'}{1}
\wedge 
\]
\[
\succk{k+1}{\anominal}{\anominal'} \wedge
  \eqk{k+1}{\anominal,\anominalbis}{\anominal,\anominalbis'})
\rightarrow 
\bigvee_{\pair{t}{t'} \in \mathcal{H}} \att{\anominal,\anominalbis}{} t \wedge \att{\anominal',\anominalbis'}{} t'
\]
Similarly, the following formula $\aformula_{\cV}$ encodes the vertical matching constraints.
\[
\forall \ \anominal, \anominal', \anominalbis, \anominalbis' \ 
\bindd{\anominal}{1} \wedge \att{\anominal}{1} \mathsf{r} \wedge 
\bindd{\anominal'}{1} \wedge \att{\anominal'}{1} \mathsf{r} \wedge 
\att{\anominal}{1} \bindd{\anominalbis}{1} \wedge \att{\anominal'}{1} \bindd{\anominalbis'}{1}
\wedge
\]
\[
\eqk{k+1}{\anominal}{\anominal'} \wedge  \succk{k+1}{\anominal,\anominalbis}{\anominal,\anominalbis'}
\rightarrow 
\bigvee_{\pair{t}{t'} \in \mathcal{V}} \att{\anominal,\anominalbis}{} t \wedge \att{\anominal',\anominalbis'}{} t'
\]

It remains to express the initial conditions. It is sufficient to identify the $n$ first children
of the first child of $\aroot$ (identified by the satisfaction of $\phifirst{k}$). 
For example, to express that the $j$th child of the first child of $\aroot$ (say $\anode$ is this first child
of $\aroot$) satisfies $t_j$, 
perform an $\mathsf{l} \mathsf{s} \mathsf{r}$-partition on $\anode$, 
enforce that the unique child satisfying $\mathsf{s}$ also satisfies $t_j$
and express that there are exactly $j-1$ children of  $\anode$  satisfying $\mathsf{r}$. 
This is a condition from graded modal logic that is easy to express. 
Let $\EX_{=i} \ \aformulabis$ be the formula below stating that exactly $i \geq 1$ children satisfy
$\aformulabis$:
\[
\exists \ \avarpropbis_1, \ldots, \avarpropbis_i \
\distinctbindd{\avarpropbis_1, \ldots, \avarpropbis_i}{1} \wedge
\AX ((\avarpropbis_1 \vee \cdots \vee \avarpropbis_i) \leftrightarrow \aformulabis),
\]
where $\avarpropbis_1, \ldots, \avarpropbis_i$ are fresh 
propositional 
variables. 
By convention $\EX_{=0} \ \aformulabis$ is defined as $\AX \neg \aformulabis$. 
The formula $\aformula_{\rm init}$ is defined as 
\[
\forall \ \anominal \  ( \bindd{\anominal}{1} \wedge \att{\anominal}{1} (\phifirst{k}))
\rightarrow \att{\anominal}{1} (
\bigwedge_{i \in \interval{0}{n-1}} \exists \ \mathsf{l}, \mathsf{s}, \mathsf{r} \ 
\lsr{k}{\varepsilon} \wedge \EX_{=i} \ \mathsf{r} \wedge \EX(\mathsf{s} \wedge t_i)
).
\]
The correctness of the reduction is stated below.

\begin{lem} 
\label{lemma:correctness-tiling-to-ex}
$\cP = (\cT, \cH, \cV)$, $c = t_0,t_1, \ldots, t_{n-1}$  is a positive instance
of $\mathtt{Tiling}_k$
iff $\aformula_\cP$ is satisfiable in \tQCTLEX{}.  
\end{lem}
\begin{proof}
Given an instance $\cP = (\cT, \cH, \cV)$, $c = t_0,t_1, \ldots, t_{n-1}$ 
 of the tiling problem $\mathtt{Tiling}_k$, let  $\tau: \interval{0}{\tow(k,n)-1} \times
\interval{0}{\tow(k,n)-1} \rightarrow \cT$ be a tiling respecting all the constraints \ref{tiling:init}, \ref{tiling:hori} and \ref{tiling:verti}. 
Below, we build a tree model $\atreemodel = \triple{V}{E}{l}$ such that
$\atreemodel, \aroot \models \aformula_{\cP}$ where $\aroot$ is the root node of $\atreemodel$. 
Let $V$ be the following subset of $\Nat^*$ (set of finite sequences over $\Nat$):
\begin{itemize}
\item $\aroot$ is the empty string and it belongs to $V$,
\item For all $j \in \interval{1}{k+1}$, 
      $\interval{0}{\tow(k+1,n)-1} \times \cdots \times \interval{0}{\tow(j,n)-1} \subseteq V$.
\item $\interval{0}{\tow(k+1,n)-1} \times \cdots \times \interval{0}{\tow(1,n)-1} \times 0^+ \subseteq V$.
\end{itemize}
The binary relation $E$ is simply defined as: $\anode E \anode'$ $\equivdef$ $\anode$ is a prefix of $\anode'$, and $\anode \cdot \alpha
= \anode'$ for some $\alpha \in \Nat$. So, $\pair{V}{E}$ is a finite-branching tree such that all the maximal branches are infinite. 
The labelling map $l$ is defined in  a way so that $\atreemodel, \aroot \models \phitype{k+1}$. 
For instance, any node $m_{k+1}, \ldots, m_{2} \in V \cap \Nat^{k-1}$ has $2^n$ children, and their numbers should span
over $\interval{0}{2^n-1}$. This is easy to realise by setting properly the truth values for $\avarprop_{n-1}, \ldots, \avarprop_0$. 
Similarly, any node $m_{k+1}, \ldots, m_{3} \in V \cap \Nat^{k-2}$ has $\tow(2,n)$ children, and their numbers should span
over $\interval{0}{\tow(2,n)-1}$. Again, this is easy to realise by setting properly the truth values of $\val{}$
on the nodes in $\interval{0}{\tow(k+1,n)-1} \times \cdots \times \interval{0}{\tow(1,n)-1}$.
So, it remains to take care of the propositional variables dedicated to the tile types.
\begin{itemize}
\item For all $\pair{i}{j} \in \interval{0}{\tow(k,n)-1} \times \interval{0}{\tow(k,n)-1}$, 
     $l(\pair{i}{j}) \cap \cT$ is equal to $\set{\tau(i,j)}$ by definition. 
     In particular, this means that for all $\pair{i}{j} \in \interval{0}{\tow(k,n)-1} \times \interval{0}{\tow(k,n)-1}$, 
     there is exactly one tile type (understood as a propositional variable) satisfied by $\pair{i}{j}$.  
\item For all $\anode \in V \setminus \interval{0}{\tow(k,n)-1} \times \interval{0}{\tow(k,n)-1}$, 
    the value of the set $l(\anode) \cap \cT$ is irrelevant. 
\end{itemize}

It remains to check that
\[
\atreemodel, \aroot \models \phitype{k+1} \wedge \exists \ \mathsf{l}, \textit{s}, \mathsf{r} \ 
\lsr{k+1}{\varepsilon} \wedge 
\EX(\mathsf{s}\wedge \anothereqk{k}) \wedge 
\aformula_{\rm cov} \wedge
\aformula_{\rm init}   \wedge
\aformula_{\cH} \wedge
\aformula_{\cV},
\]
which is routine and done in Appendix~\ref{appendix:lemma:correctness-tiling-to-ex}.
\end{proof}

We are ready to conclude the main theorem of this paper.
\begin{thm} 
\label{theorem:tQCTLEX-tower}
\satproblem{\tQCTLEX{}} is \tower-complete.
\end{thm}
Theorem~\ref{theorem:tQCTLEX-tower} significantly 
improves  the \tower lower bound from~\cite[Cor. 5.6]{LaroussinieM14} 
by considering as only temporal operator, the (local) modality $\EX$. 
\tower-hardness can be also obtained with arbitrary countable trees.
In Section~\ref{section-collection} below, we show that
this entails more \tower-hardness results for other fragments of \tQCTLEX{} and for modal logics
with propositional quantification under appropriate tree semantics.
\section{A harvest of \tower-complete modal and temporal logics}\label{section-collection}

We capitalise on the \tower-hardness of the satisfiability problem for \tQCTLEX{}, 
by showing \tower-hardness of other fragments of \tQCTL{} that involve only $\EF$ or its strict
variant $\EX \EF$ (Section~\ref{section-qctlef}). 
\tower-hardness is obtained by reduction from \satproblem{\tQCTLEX{}} by introducing propositional variables that enforce layers from the root in the tree model and therefore this allows us to simulate $\EX$. 
In Section~\ref{section-modal-logics}, we consider well-known modal logics that are complete 
for classes of tree-like Kripke structures, and we show that their extension with propositional
quantification for such classes of tree-like Kripke structures 
is decidable in \tower, but more importantly \tower-hard. 
Some of such classes involve finite trees and therefore, we also take the opportunity to
study \ftQCTLEX{} and \ftQCTLEFplus{} that happen, for instance, to be closely related
to the  modal logics \QKt and \QGLt, respectively. 

\subsection{The satisfiability problems for \texorpdfstring{\tQCTLEF{}}{tQCTLEF} and \texorpdfstring{\tQCTLEFplus{}}{tQCTLEFplus} are \tower-hard!}\label{section-qctlef}

The fragment \tQCTLEF{} of \tQCTL{} is defined according to the following grammar
\[
\aformula ::= \avarprop \ | \ \neg \aformula \ 
| \  \aformula \wedge \aformula \ 
| \ \EF \aformula \ 
|  \ \ \exists{\avarprop} \; \aformula.
\]
We recall the standard semantics for $\EF$-formulae:
$\atreemodel, \anode \models \EF \aformula$ $\equivdef$
there is $j \geq 0$ such that $\anode E^j \anode'$ and  $\atreemodel, \anode' \models \aformula$, and 
as usual, $\AG \aformula \egdef \neg \EF \neg \aformula$. 

In order to show that \satproblem{\tQCTLEF{}} is \tower-hard, 
we design a logarithmic-space many-one reduction from \satproblem{\tQCTLEX{}}.
A more sophisticated analysis is also possible to establish \tower-hardness for even smaller fragments, see
the recent work~\cite{Mansutti20bis}.
   
Let $\aformula$ be a formula in \tQCTLEX{} with modal depth $\md{\aformula} = k$. 
Without loss of generality, we assume that $\aformula$ may contain occurrences 
of $\EX$ and no occurrences of $\AX$. Let us define 
the formula $\aformula' = \textit{trans}(k,\aformula) \wedge {\rm shape}(k)$
in \tQCTLEF{},  where the formula ${\rm shape}(k)$ enforces a discipline for layers (explained below) 
and $ \textit{trans}(k,\aformula)$ admits a recursive definition, by relativising the occurrences of $\EX$.
We consider the set of propositional variables 
$\asetbis_{k} = \set{\layer{-1}, \layer{0}, \ldots, \layer{k}}$
with the intended meaning that a node satisfying $\layer{i}$ is of ``layer $i$'', the root node being of layer~$k$.
Indeed, there is a need for such propositional variables, as unlike with 
the formulae in \tQCTLEX{}, we have to enforce that moving with $\EF$ leads
to a lower layer. 

The formula ${\rm shape}(k)$ is defined as the conjunction of the following formulae.
\begin{itemize}
\itemsep 0 cm 
    \item Every node satisfies exactly one propositional variable from $\asetbis_k$ (layer unicity):
         \[
         \AG \left( (\layer{-1} \vee \layer{0} \vee \cdots \vee \layer{k} ) 
         \wedge \bigwedge_{-1 \leq i\neq j \leq k} \neg (\layer{i} \wedge \layer{j}) \right) 
         \] 
    \item When a node satisfies $\layer{i}$ with $-1 \leq i \leq k$, none of its 
        descendants satisfies $\layer{j}$ with $j > i$ (monotonicity of layer numbers):
        \[
        \bigwedge_{-1 \leq i \leq k} \AG(\layer{i} \rightarrow 
        \AG (\layer{-1} \vee \layer{0} \vee \cdots \vee \layer{i}))
        \]
    \item When a node satisfies $\layer{i}$ with $0 \leq i \leq k$,  there is a descendant
    satisfying $\layer{i-1}$ (weak progress):
    \[
    \bigwedge_{0 \leq i \leq k} \AG(\layer{i} \rightarrow \EF \ \layer{i-1})
    \]

    \item When a node satisfies $\layer{i}$ with $0 \leq i \leq k$, it has no (strict) descendant satisfying $\layer{i}$ (no stuttering). This type of constraints does not apply to $\layer{-1}$.
        \[
        \bigwedge_{0 \leq i \leq k}
        \AG(\layer{i} \rightarrow \neg \exists{\avarprop} \ (\avarprop  \wedge \EF(\layer{i} \wedge \neg \avarprop)))\]
         
    \item The root node is at layer $k$: $\layer{k}$. 
\end{itemize}

\noindent A tree model $\atreemodel = \triple{V}{E}{l}$ for \tQCTL{} with root $\aroot$
is \defstyle{$k$-layered} iff the conditions below~hold:
\begin{enumerate}
\itemsep 0 cm 
\item[(a)] For every node $\anode \in V$, $\card{l(\anode) \cap \asetbis_{k}} = 1$. 
\item[(b)] For all $\anode \in V$ such that $\layer{j} \in l(\anode)$ for some $j \in \interval{-1}{k}$,
           \begin{itemize}
           \item if $j \geq 0$, then there is $\anode'$ such that $\anode E \anode'$ and $\layer{j-1} \in l(\anode')$ and,
           \item for all $\anode'$ such that $\anode E^+ \anode'$ and $\layer{j'} \in l(\anode')$, we have $j' \leq j$. 
           \end{itemize}
\item[(c)] For all $j \in \interval{0}{k}$, there are no distinct nodes $\anode$  and $\anode'$
           such that $\anode E^+ \anode'$ and $\layer{j} \in l(\anode) \cap l(\anode')$. 
\item[(d)] $\layer{k} \in l(\aroot)$, where $\aroot$ is the root of the tree model.
\end{enumerate}

This means that the only propositional variable from $\asetbis_k$ 
satisfied by a node reachable in $j \in \interval{1}{k}$ steps from $\aroot$ is 
$\layer{m}$ for some $m \leq k-j$, and the only propositional variable from $\asetbis_k$ 
satisfied by a node reachable in strictly more than $k$ steps
from $\aroot$  is $\layer{-1}$. Moreover, once $\layer{-1}$ holds true, 
it holds for all its descendants.
Actually, the formula ${\rm shape}(k)$ characterises $k$-layered structures. 

\begin{lem}\label{lemma:k-layered}
Let $\atreemodel = \triple{V}{E}{l}$ be a tree model for \tQCTL{} with the root node $\aroot$.
We have $\atreemodel, \aroot \models {\rm shape}(k)$ holds if and only if $\atreemodel$ is $k$-layered.
\end{lem}
\begin{proof}
By careful inspection of the semantics, \cf Appendix~\ref{appendix:proof-of:lemma:k-layered}.
\end{proof}

\noindent To define $\textit{trans}(k,\aformula)$, we define 
 inductively 
$\textit{trans}(i,\aformulabis)$ where
$\aformulabis$ is a subformula of $\aformula$ and $\md{\aformulabis} \leq i$.
\begin{itemize}
\itemsep 0 cm 
\item $ \textit{trans}(i, \avarprop) \egdef \avarprop$ for all propositional variables $\avarprop$, 
\item $ \textit{trans}$ is homomorphic for Boolean connectives and  $ \textit{trans}(i, \exists \ \avarprop \ \aformulabis) \egdef \exists \ \avarprop \ \textit{trans}(i, \aformulabis)$,
\item $ \textit{trans}(i, \EX \aformulabis) \egdef \EF \left(\layer{i{-}1} \wedge  \textit{trans}(i{-}1, \aformulabis)\right)$.
\end{itemize}
Note that $ \textit{trans}(k,\aformula)$ has no occurrence of $\layer{-1}$ 
since $\md{\aformula}=k$ and that translating an $\EX$-formula decreases 
the index of the layer by exactly one. 
The correctness of the reduction can be now stated as follows. 

\begin{lem}
\label{lemma:correctness-tQCTLEF}
A formula $\aformula$ is satisfiable for \tQCTLEX{} iff
$\textit{trans}(k,\aformula) \wedge {\rm shape}(k)$
is satisfiable for \tQCTLEF{}. 
\end{lem}
\begin{proof}
By induction. Consult Appendix~\ref{appendix:proof:of:lemma:correctness-tQCTLEF}.
\end{proof}

Hence we conclude yet another important result. 

\begin{thm}\label{theorem:tQCTLEF}
The satisfiability problem for \tQCTLEF{} is \tower-complete.
\end{thm}

The \tower upper bound  is established for 
the full logic \tQCTL{} in~\cite{LaroussinieM14} 
and in particular for~\tQCTLEF{}. Theorem~\ref{theorem:tQCTLEF} also admits a variant in which we only allow
to move to proper descendants. It amounts to replacing $\EF$ by $\EX \EF$ (treated here as a single modality) in  $\tQCTLEF{}$, leading to
the variant  \tQCTLEFplus{}, with formulae obtained~from
\[
\aformula ::= \avarprop \ | \ \neg \aformula \ 
| \  \aformula \wedge \aformula \ 
| \ \EX \EF \aformula \ 
|  \ \ \exists{\avarprop} \; \aformula
\]
As usual, we  write $\AX \AG \; \aformulabis$ as an abbreviation of $\neg \EX \EF \neg \aformula$.
\begin{thm}\label{theorem:tQCTLEFplus}
The satisfiability problem for \tQCTLEFplus{} is \tower-complete.
\end{thm}
As above, the \tower upper bound for \satproblem{\tQCTLEFplus{}} is inherited from \satproblem{\tQCTL{}}~\cite{LaroussinieM14}. 
Note that in \tQCTL{}, the formula $\EF \avarprop$ is logically equivalent to $\avarprop \vee \EX \EF \avarprop$. Thus, we get a (possibly exponential, which is sufficient for us) reduction from \satproblem{\tQCTLEF{}} to \satproblem{\tQCTLEFplus{}}, 
whence the \tower-hardness of~\satproblem{\tQCTLEFplus{}}.  Recall that \tower-hardness is defined with respect
to elementary reductions and therefore an exponential-time reduction is fine to establish \tower-hardness,
see e.g.~\cite{Schmitz16} and Section~\ref{sec:complexity}.  

\subsection{Modal logics with propositional quantification on trees}
\label{section-modal-logics}

Numerous well-known  modal logics are complete (a.k.a. determined) for classes of tree-like structures. 
A modal logic $\alogic$ (defined from the Hilbert-style system $\mathcal{H}\alogic$)
is complete for a class of Kripke structures $\mathcal{C}$ iff
the theoremhood of $\aformula$ in $\mathcal{H}\alogic$ is equivalent
to the validity of $\aformula$ in $\mathcal{C}$ (i.e. for all $\kripkeK \in \mathcal{C}$,
for all $w$, we have $\kripkeK, w \models \aformula$). 
For instance, the  (propositional) modal logic \K is complete for the class of finite 
trees~\cite{Segerberg71,blackburn_rijke_venema_2001}. 
It is worth noting that a given modal logic can be complete for different classes of Kripke models
(\eg \K is complete for the class of all the Kripke models, but also complete for the class of finite Kripke models)
and their extension to propositional quantification may lead to distinct logics.
Typically, \K with propositional
quantification under the structure semantics is undecidable~\cite{Fine70} whereas it is shown below to be 
\tower-complete under the finite tree semantics. 

Below, for the propositional modal logics  $\alogic$ in \K,  \KD, \GL, \Kfour and \Sfour we define
an extension $\mathsf{Q}\alogic^t$ with propositional quantification under a class of tree-like models
that is complete  for the logic $\alogic$. 

In order to avoid too many notations, the modalities for each logic $\mathsf{Q}\alogic^t$
are $\EX$ and $\AX$ (instead of the more standard modal operators $\Diamond$ and $\Box$) 
and therefore $\mathsf{Q}\alogic^t$ formulae are built from the grammar below:
$
\aformula ::= 
\avarprop \mid  \neg \aformula \mid \aformula \wedge \aformula \mid
\EX \aformula \mid \AX \aformula 
\mid \exists{\avarprop} \  \aformula
$.
\begin{itemize}
\item The propositional modal logic \K is complete for the class of finite trees and we define \QKt as the modal logic
with propositional quantification over the class of finite trees. 
\item The propositional modal logic \KD (\K with seriality, a.k.a. totality) is known to be complete for the class
of finite-branching trees for which all the maximal branches are infinite.
Indeed, \KD is complete for the class of finite total Kripke models. 
Therefore by using the unfolding
construction, completeness applies also for 
 the class
of finite-branching trees for which all the maximal branches are infinite, \ie  the models for \tQCTL{}. 
Let \QKDt be the modal logic
with propositional quantification over the class 
of finite-branching trees for which all the maximal branches are infinite.
The satisfiability problem for \QKDt  is exactly the problem for \tQCTLEX{}.  
\item The modal logic \GL is known to be complete for the class of finite transitive trees
(\GL is complete with respect to finite irreflexive transitive Kripke models~\cite{Smorynski85}), \ie the class
of Kripke structures $\triple{V}{E^+}{l}$ such that $\triple{V}{E}{l}$ is a finite tree model, see \eg~\cite{blackburn_rijke_venema_2001}. 
Let \QGLt be the modal logic with propositional quantification over the
class of finite transitive trees, which is precisely \ftQCTLEFplus{}.
It is worth noting that adding propositional quantification to \GL is studied in~\cite{Artemov&Beklemishev93},
where a fragment is shown to be decidable by translation into the weak monadic second-order logic of one successor 
WS1S~\cite{Buchi60}.  
\item The modal logic \Kfour is complete for the class of Kripke structures
$\triple{V}{E^+}{l}$ such that $\triple{V}{E}{l}$ is a finite-branching tree 
model (some branches may be infinite, some others not).
Let \QKfourt be the modal logic with propositional quantification over the
class of finite-branching trees. 
\item The modal logic \Sfour is complete for the class of finite Kripke structures such that
the accessibility relation is reflexive and transitive, and therefore complete
for the class of structures
$\triple{V}{E^*}{l}$ such that $\triple{V}{E}{l}$ is a finite-branching tree 
model in which all the  branches are infinite (by unfolding). 
Let \QSfourt be the modal logic with propositional quantification over the
class of finite-branching trees in which all the  branches are infinite (precisely 
the class of models for \tQCTL{}). The satisfiability problem for \QSfourt  happens to be exactly the problem 
for \tQCTLEF{}, modulo the fact that $\EX$ in \QSfourt corresponds to $\EF$ in \tQCTLEF{}.
\end{itemize}

In this section, we show that the satisfiability problem for
the logics \QKt, \QKDt, \QGLt, \QKfourt and \QSfourt
and whose models are tree-like Kripke structures is \tower-complete.
For instance, \QKt corresponds to the modal logic \K interpreted on finite trees
with propositional quantification, which is precisely \ftQCTLEX{}, \ie \ftQCTL{}
restricted to the~$\EX$ operator.

\begin{thm} 
\label{theorem:K}
\satproblem{\QKt} is \tower-complete.
\end{thm}
\begin{proof}
The satisfiability problem in Theorem~\ref{theorem:K} is exactly \satproblem{\ftQCTLEX{}}. An easy translation can be found in Appendix~\ref{appendix:proof-of:theorem:K}.
\end{proof}

As a corollary of Theorem~\ref{theorem:aexppol-hardness}, 
for all $N \geq 2$, the satisfiability problem for $\QKt_{\leq N}$ 
is \aexppol-complete too where $\QKt_{\leq N}$ is interpreted on finite trees whose 
branching degree is at most $N$. The \aexppol lower bound can be obtained using the same
reduction as for \tQCTLEX{, \leq N} whereas the \aexppol upper bound uses also the same arguments
as for \tQCTLEX{, \leq N}.

We have seen that \QKDt is actually defined
as \tQCTLEX{} (\KD is characterised by total and finite Kripke structures whose
unfoldings generate finite-branching trees in which all the branches are infinite).
Consequently:

\begin{thm} 
\label{theorem:QKDt}
\satproblem{\QKDt} is \tower-complete. 
\end{thm}

Recall that the modal logic \GL is known to be complete for the class of finite transitive trees.
The logic \QGLt extends it with propositional quantification. We obtain the following:

\begin{thm}\label{theorem:GL}
\satproblem{\QGLt} is \tower-complete.
\end{thm}
\begin{proof}
The logic can be shown to be precisely \ftQCTLEFplus{} when $\EX$ is translated into
$\EX\EF$, whence we get a \tower upper bound.
For the \tower-hardness proof, it is very similar to the one for \tQCTLEF{} (actually, it is a bit simpler).
Consult Appendix~\ref{appendix:proof-of:theorem:GL}.
\end{proof}

\QKfourt is defined as the modal logic with propositional quantification over the
class of finite-branching trees. Our next theorem is as follows.
\begin{thm}\label{theorem:Kfour}
\satproblem{\QKfourt} is \tower-complete.
\end{thm}
\begin{proof}
Full proof is in Appendix~\ref{appendix:proof-of:theorem:Kfour}.
We first show \satproblem{\QKfourt} and \satproblem{\gtQCTLEFplus{}}
are identical modulo the rewriting of $\EX$ into $\EX \EF$. 
This yields the \tower upper bound. 
As far as \tower-hardness is concerned, for any formula $\aformula$ in 
\tQCTLEFplus{}, one can show that $\aformula$ is satisfiable for \tQCTLEFplus{}
iff $\aformula \wedge \EX \ \EF \ \top \wedge \AX \AG \ \EX  \EF \ \top$ is satisfiable
in \gtQCTLEFplus{}. 
\end{proof}

Finally, by noting that \QSfourt is equal to \tQCTLEF{} modulo that $\EX$ is rewritten into $\EF$, 
using Theorem~\ref{theorem:tQCTLEF}, we get the following complexity characterisation. 

\begin{thm}  
\label{theorem:Sfour}
\satproblem{\QSfourt} is \tower-complete.
\end{thm}


\section{Conclusion}\label{section-conclusion}

In the paper, we have developed a relatively simple proof method to show that the satisfiability problems for 
\tQCTLEX{}, \tQCTLEF{} and \tQCTLEFplus{} are \tower-complete, see also similar methods in~\cite{Stockmeyer74,Prattetal19}. 
Our contribution is to establish \tower-hardness, which could be also
shown for several modal logics with propositional quantification 
whose respective classes of models are tree-like structures.  
Moreover, in the case of fixed degree, we have shown that for all $N \geq 2$, 
the satisfiability problem for the variant $\tQCTLEX{,\leq N}$ is \aexppol-complete. 
Whereas \aexppol-hardness is established by reducing the alternating multi-tiling problem 
recently introduced in~\cite{Bozzellietal17}, the \tower-hardness of \satproblem{\tQCTLEX{}} 
is essentially based on the fact that one can enforce concisely that a node has a number 
of children equal to some tower of exponentials.

Section~\ref{section-collection} deals with the \tower-completeness of 
\satproblem{\ftQCTLEX{}} and \satproblem{\ftQCTLEFplus{}}, as well as \tower-completeness for 
the well-known modal logics \K, \KD, \GL, \Kfour and \Sfour 
extended with propositional quantification 
but with adequate classes of tree-like structures. 
Though the \tower upper bound  for decision problems on trees should not come as a real surprise, 
all our \tower-hardness results significantly improve the current state-of-the-art
regarding the fragments of \tQCTL{} and for the above-mentioned modal logics. 
In particular, our proof technique for \tower-hardness of 
\satproblem{\tQCTLEX{}} (and therefore for \QKt on finite trees)
is simple enough so that it could be further reused or adapted, see \eg 
a recent refinement of the proof in~\cite{Bednarczyketal20}. 

This work can be continued in several directions.
For instance, \tower-hardness of \satproblem{\tQCTLEF{}} is recently refined 
in~\cite{Mansutti20,Mansutti20bis} by establishing that already 
\tQCTLEF{}  restricted to formulae of modal/tem\-po\-ral depth two is also
\tower-hard. 
Among the several directions, one of them would be to characterise the expressiveness 
of \tQCTLEX{} or \tQCTLEF{} along the lines of~\cite{DavidLM16} or~\cite{Kuusisto15}, see 
also~\cite{Aucher&vanBenthem&Grossi18}.
More generally, we believe that standard modal logics with propositional quantification, but under the 
tree semantics, deserve to be much better understood.

\section*{Acknowledgements}
We would like to thank the reviewers for the numerous suggestions that help
us to improve the quality of the document.
Bartosz Bednarczyk was supported by the Polish Ministry of Science and 
Higher Education programme ``Diamentowy Grant'' no. DI2017 006447.
St\'ephane Demri is supported by the Centre National de la Recherche
Scientifique (CNRS). Furthermore, we would like to thank Raul Fervari 
(University of Cordoba) and Alessio Mansutti (University of Oxford) 
for many insightful discussions about modal logics with propositional
quantification, modal separation logics and related topics. 

\bibliographystyle{alphaurl}
\bibliography{bibliography}

\clearpage
\appendix


\section{Proofs  from Section~\ref{section-bounded-degree}}

\subsection{Proof of Lemma~\ref{lemma:nominal}}\label{appendix:proof-of:lemma:nominal}

\begin{proof} First, suppose that $\anominal$ is a nominal for the depth $k \geq 0$ from $\anode$. 
By definition, this means that there is $\anode' \in V$ satisfying $\anode E^{k} \anode'$ such that $\atreemodel, \anode' \models \anominal$  
and for all $\anode'' \neq \anode'$ satisfying $\anode E^{k} \anode''$ we have $\atreemodel, \anode'' \not \models \anominal$. 
Obviously, we have $\atreemodel, \anode \models \EX^k \anominal$. 
\emph{Ad absurdum}, suppose that 
$\atreemodel, \anode \models \exists \avarprop \ (\EX^k(\anominal \wedge \avarprop) \wedge 
\EX^k(\anominal \wedge \neg \avarprop))$ with $\avarprop$ being a fresh propositional variable.
So, there is a $(\PVAR \setminus \set{\avarprop})$-variant $\atreemodel'$ of~$\atreemodel$ (\ie $\atreemodel \approx_{\AP \setminus \set{\avarprop}} \atreemodel'$) 
such that $\atreemodel', \anode \models \EX^k(\anominal \wedge \avarprop) \wedge \EX^k(\anominal \wedge \neg \avarprop)$. 
Therefore, there is $\anode'$ such that~$\anode E^k \anode'$ and $\atreemodel', \anode' \models \anominal \wedge \avarprop$.
Similarly, there is $\anode''$ such that $\anode E^k \anode''$ and $\atreemodel', \anode'' \models \anominal \wedge \neg \avarprop$.
Because of the constraint on the satisfaction of the propositional variable $\avarprop$, 
the nodes $\anode'$ and $\anode''$ are distinct, which leads to a contradiction.

Conversely, suppose that $\atreemodel, \anode \models \bindd{\anominal}{k}$.
As $\atreemodel, \anode \models \EX^k \anominal$, there exists $\anode'$ such that $\anode E^k \anode'$
and~$\atreemodel, \anode' \models \anominal$. The uniqueness of $\anode'$ can be concluded from
the satisfaction $\neg \exists \avarprop \ (\EX^k(\anominal \wedge \avarprop) \wedge 
\EX^k(\anominal \wedge \neg \avarprop))$ (as above), since that formula characterises exactly
the property that there are no two distinct nodes reachable in $k$ steps from $\anode$ satisfying $\anominal$. 
\end{proof}

\subsection{Proof of Lemma~\ref{lemma:at}}\label{appendix:proof-of:lemma-at}

\begin{proof} 
By assumption, there is a node $\anode' \in V$ satisfying $\anode E^{k} \anode'$ such that $\atreemodel, \anode' \models \anominal$, 
and for all~$\anode'' \neq \anode'$ satisfying~$\anode E^{k} \anode''$, we have $\atreemodel, \anode'' \not \models \anominal$. 
First, suppose that $\atreemodel, \anode \models \att{\anominal}{k} \aformula$, \ie $\atreemodel, \anode \models \EX^{k}(\anominal \wedge \aformula)$.
Hence, there exists $\anode''$ such that $\anode E^k \anode''$ and $\atreemodel, \anode'' \models \anominal \wedge \aformula$.
As $\atreemodel, \anode'' \models \anominal$, the node $\anode''$ is necessarily equal to~$\anode'$ and therefore  $\atreemodel, \anode' \models \aformula$.
For the opposite direction, suppose that $\atreemodel, \anode' \models \aformula$. 
As $\anode E^{k} \anode'$ and  $\atreemodel, \anode' \models \anominal$, by the definition of the satisfaction relation~$\models$, we conclude that 
$\atreemodel, \anode \models \EX^{k}(\anominal \wedge \aformula)$ and $\EX^{k}(\anominal \wedge \aformula)$ is equal to~$\att{\anominal}{k} \aformula$.
\end{proof}

\subsection{Proof of Lemma~\ref{lemma:atoflengthd}}\label{appendix:proof-of:lemma-atoflengthd}

\begin{proof} 
The proof is by induction on $d$. For the base case, suppose that $d=1$.
Thus, we have that~$\atreemodel, \anode_0 \models \bindd{\anominal_1}{1}$, and that $\anode_1$ is the unique
child of $\anode_0$ such that $\anode_0 E \anode_1$ and $\atreemodel, \anode_1 \models \anominal_1$. 
By~Lemma~\ref{lemma:at}, we have $\atreemodel, \anode_0 \models \att{\anominal_1}{1} \aformula$ iff
$\atreemodel, \anode_1 \models \aformula$ and $ \att{\anominal_1}{} \aformula$ is equal to  
$\att{\anominal_1}{1} \aformula$, so we are done.
For the induction step with $d \geq 2$, let us assume that 
\[
\atreemodel, \anode_0 \models \bindd{\anominal_1}{1} \wedge  \bigwedge_{i \in \interval{2}{d}}
\att{\anominal_1}{1} \att{\anominal_2}{1} \cdots \att{\anominal_{i-1}}{1} \bindd{\anominal_i}{1},
\]
and $\anode_1, \ldots, \anode_d$ is associated to $\anominal_1, \ldots \anominal_d$.
The propositions below are equivalent
\begin{itemize}
\item $\atreemodel, \anode_0 \models \att{\bar{\anominal}}{} \aformula$,
\item  $\atreemodel, \anode_0 \models  
       \att{\anominal_1}{1} \att{\anominal_2}{1} \cdots \att{\anominal_d}{1} \aformula$ (by definition of $\att{\bar{\anominal}}{}$),
\item $\atreemodel, \anode_{d-1} \models \att{\anominal_d}{1} \aformula$ (by the induction hypothesis),
\item $\atreemodel, \anode_{d} \models \aformula$ (by Lemma~\ref{lemma:at}). 
\end{itemize}
Hence, $\atreemodel, \anode_0 \models \att{\bar{\anominal}}{} \aformula$ holds iff $\atreemodel, \anode_{d} \models \aformula$ holds, finishing the proof.
\end{proof}

\subsection{Proof of Lemma~\ref{lemma:pnf}}\label{appendix:proof-of:lemma-pnf}

\begin{proof} 
By way of example, we prove that 
$\EX \ \forall \ \avarprop \ \aformulabis \leftrightarrow 
\forall \ \avarprop \ \EX \aformulabis$ is valid.
(Other formulae, \eg~$\EX \ \exists \ \avarprop \ \aformulabis \leftrightarrow 
\exists \ \avarprop \ \EX \aformulabis$ can be proved to be valid in an even simpler way). 
Before doing so, note that then the valid equivalences below
provide a rewriting system (by reading the equivalences from left to right) that pushes the propositional quantification
outside (in the usual way), leading to formulae in PNF in polynomial-time ($\mathcal{Q} \in \set{\exists,\forall}$).
$$
\EX \ \mathcal{Q} \ \avarprop \ \aformulabis \leftrightarrow 
\mathcal{Q} \ \avarprop \ \EX \aformulabis \ \ \ \
(\mathcal{Q} \ \avarprop \ \aformulabis) \wedge \aformulabis' \leftrightarrow \mathcal{Q} \ \avarprop 
\ (\aformulabis \wedge \aformulabis') 
\ \ \ \
\neg \exists  \ \avarprop \ \aformulabis \leftrightarrow \forall \ \avarprop \ \neg \aformulabis \ \ \ \ 
\neg \forall  \ \avarprop \ \aformulabis \leftrightarrow \exists \ \avarprop \ \neg \aformulabis
$$
assuming that $\avarprop$ does not occur in $\aformulabis'$ (otherwise, rename the quantified variable). 

First, assume that $\atreemodel, \anode \models \EX \ \forall \ \avarprop \ \aformulabis$ with $\atreemodel = \triple{V}{E}{l}$. 
Thus, there exists $\anode'$ such that $\anode E \anode'$ and~$\atreemodel, \anode' \models \forall \ \avarprop \ \aformulabis$. 
Hence, for all $\atreemodel' \approx_{\AP \setminus \set{\avarprop}}
\atreemodel$, we have $\atreemodel', \anode' \models \aformulabis$. As $\anode'$ remains a child of~$\anode$ for all
such variants $\atreemodel'$, we have that for 
all $\atreemodel' \approx_{\AP \setminus \set{\avarprop}}
\atreemodel$, $\anode E'  \anode'$ and  $\atreemodel', \anode' \models \aformulabis$.
Thus, for all~$\atreemodel' \approx_{\AP \setminus \set{\avarprop}} \atreemodel$ we have $\atreemodel', \anode \models \EX \aformulabis$. 
We conclude $\atreemodel, \anode \models \forall \ \avarprop \ \EX \aformulabis$. 

Conversely (and this is the place where the tree structure is essential), 
assume $\atreemodel, \anode \models \forall \ \avarprop \ \EX \aformulabis$.
It~means that for all $\atreemodel' \approx_{\AP \setminus \set{\avarprop}}
\atreemodel$, we have $\atreemodel', \anode \models \EX \aformulabis$.
\emph{Ad absurdum}, suppose~$\atreemodel, \anode \not \models \EX \ \forall \ \avarprop \ \aformulabis$.
Equivalently, we have $\atreemodel, \anode \models \AX \ \exists \ \avarprop \ \neg \aformulabis$.
Thus, for all $\anode'$ such that $\anode E \anode'$, there is 
$\atreemodel_{\anode'} \approx_{\AP \setminus \set{\avarprop}}
\atreemodel$ (say $\atreemodel_{\anode'} = \triple{V}{E}{l_{\anode'}}$)
such that $\atreemodel_{\anode'}, \anode' \not \models \aformulabis$. 
Let $\atreemodel_{\anode'}' =  \triple{V_{\anode'}}{E_{\anode'}}{l_{\anode'}'}$ be the restriction
of $\atreemodel_{\anode'}$ to the subtree whose root is $\anode'$. Obviously 
$\atreemodel_{\anode'}', \anode' \not \models \aformulabis$
too. Two distinct children $\anode_1$ and $\anode_2$ of $\anode$ in $\atreemodel$, 
lead to tree models $\atreemodel_{\anode_1}' =  \triple{V_{\anode_1}}{E_{\anode_1}}{l_{\anode_1}'}$ and
$\atreemodel_{\anode_2}' =  \triple{V_{\anode_2}}{E_{\anode'}}{l_{\anode_2}'}$  with $V_{\anode_1} \cap V_{\anode_2} = \emptyset$.
Indeed, $\atreemodel$ is a tree. Let $\atreemodel^{\star}$ be defined as 
$\triple{V}{E}{(\uplus_{\anode' \in E(\anode)} l_{\anode'}') \uplus \set{\anode \mapsto l(\anode)}}$, where 
$\uplus$ denotes the disjoint sum. 
One can show that $\atreemodel^{\star} \approx_{\AP \setminus \set{\avarprop}} \atreemodel$ and for all $\anode'$ 
such that $\anode E \anode'$, $\atreemodel^{\star}, \anode' \not \models \aformulabis$, which leads to a contradiction.
\end{proof}

\subsection{Proof of Lemma~\ref{lemma:tiling-ok}}\label{appendix:proof-of:tiling-ok}

\begin{proof}
We start with the first item of the lemma. The definition of $\tau$ is correct, since:
\begin{itemize}
    \item for each position $\pair{x}{y} \in \interval{0}{2^n-1} \times \interval{0}{2^n-1}$ there is a unique node $\anode'$ in 
     the distance $2n$ from $\anode$ encoding the position $\pair{x}{y}$ 
(follows from Corollary~\ref{corr:grid-2n-correctness}),
    \item for each $\anode'$, as defined above, there is a unique tile proposition $t^j$ such that $\atreemodel, \anode' 
\models t^j$ holds (it is a direct consequence of the satisfaction of $\aformula_{\rm cov}^j$ at $\anode$).
\end{itemize}
The satisfaction of the conditions \ref{tiling:hori} and \ref{tiling:verti} follows from the satisfaction of $\aformula_{\cH}^j  \wedge \aformula_{\cV}^j$.
Indeed, let us discuss the condition \ref{tiling:hori} only, since \ref{tiling:verti} is analogous.
Take any two consecutive positions $\pair{a}{b}$ and~$\pair{a}{b+1}$ such that $\tau(a,b) = t$ and $\tau(a,b+1) = t'$.
Then, let $\anode_t, \anode_{t'}$ be nodes at the distance $2n$ from~$\anode$, representing the positions  $\pair{a}{b}$ and $\pair{a}{b+1}$.
Set the local nominals $\anominal$ and $\anominalbis$ at $\anode_{t}$ and $\anode_{t'}$.
Then, note that the formula ${\rm HN}(\anominal,\anominalbis)$ is satisfied at $\anode$ 
by elementary operations on binary encodings of numbers. Hence, by the right-hand side of the implication in $\aformula_{\cH}^j$ we conclude that $(t,t') \in \cH$.

To show the second item, we take $\atreemodel'= \triple{V}{E}{l'}$ obtained from $\atreemodel = \triple{V}{E}{l}$
by setting $l'(\anode_{\pair{a}{b}}) = (l(\anode_{\pair{a}{b}}) \setminus \cT^{j}) \cup \{ \tau(a,b) \}$ for the unique node 
$\anode_{\pair{a}{b}}$ at the distance $2n$ from $\anode$ corresponding to $\pair{a}{b}$ in the grid. 
Otherwise $l$ and~$l'$ coincide.
Since the formula $\mathrm{grid}(2n)$ does not employ propositions from $\cT^{j}$ we conclude~$\atreemodel', \anode \models \mathrm{grid}(2n)$. By the definition of $\tau$ we know that each node at the distance $2n$ from $\anode$ is labelled with exactly one tile proposition from $\cT^{j}$ and hence, $\atreemodel', \anode \models \aformula_{\rm cov}^j$.
Checking that $\atreemodel', \anode$ satisfies $\aformula_{\cH}^j  \wedge \aformula_{\cV}^j$ is routine and follows from the fact that $\tau$ is a tiling (so it satisfies \ref{tiling:hori} and \ref{tiling:verti}).
\end{proof}

\subsection{Proof of Lemma~\ref{lemma:aexppol-hardness}}\label{appendix:proof-of:proof-lemma-aexppol-hardness}

\begin{proof}  Let $\mathcal{I}$ be the instance $n$,  
$\triple{\cT}{\cH}{\cV}$, $\cTzero$, $\cTacc$, $\cTmulti$ of $\cPmulti$.
Before showing that the reduction is correct, we need to state  preliminary properties.
Moreover, in the proof below, we repeat several formula definitions to follow more smoothly
the technical developments. 

\textbf{(GRID)} We recall that ${\rm grid}(2n)$ is defined as the formula below
\[
(\bigwedge_{i \in \interval{0}{2n-1}} \AX^i \ \EX_{=2} \top)  \ \wedge
\forall \anominal, \anominalbis \ \distinctbindd{\anominal, \anominalbis}{2n} \rightarrow
\]
\[
(\bigvee_{j \in \interval{0}{n-1}} \neg (\att{\anominal}{2n} \ahvarprop_j \leftrightarrow \att{\anominalbis}{2n} \ahvarprop_j)
\vee \neg (\att{\anominal}{2n} \avvarprop_j \leftrightarrow \att{\anominalbis}{2n} \avvarprop_j)
).
\]
Given a tree model $\atreemodel$ and a root node $\aroot$, one can show that $\atreemodel, \aroot \models {\rm grid}(2n)$ iff
the properties below hold:
\begin{enumerate}
\item[(a)] For all $j \in \interval{0}{2n-1}$, we have that $\aroot E^j \anode$ implies that $\anode$ has exactly two children.
\item[(b)] For all distinct nodes $\anode, \anode'$  such that $\aroot E^{2n} \anode$ and $\aroot E^{2n} \anode'$,
there is some propositional variable $\avarpropter$ in $\set{\ahvarprop_{n-1}, \ldots,\ahvarprop_{0},\avvarprop_{n-1}, \ldots,\avvarprop_{0}}$,
such that $\anode$ satisfies $\avarpropter$ iff $\anode'$ does not satisfy $\avarpropter$. 
\end{enumerate}
Satisfaction of (a) is essentially due to the fact that $\EX_{=2} \top$ is defined as 
\[
\exists \ \anominal_1,  \anominal_2 \
\distinctbindd{\anominal_1, \anominal_2}{1} \wedge \AX (\anominal_1 \vee  \anominal_2)
\] 
and one can check that it holds true
on nodes having exactly two children. 
For the satisfaction of (b), we need to invoke Lemma~\ref{lemma:nominal}, Lemma~\ref{lemma:at} and Lemma~\ref{lemma:distinct-nominals}.
Assuming that $\distinctbindd{\anominal, \anominalbis}{2n}$ holds, (b) is equivalent to have two distinct nodes that are the 
respective interpretations of the nominals $\anominal$ and $\anominalbis$ for the depth $2n$. 
More generally, Lemma~\ref{lemma:tiling-ok} states the main properties that are used about ${\rm grid}(2n)$.

\textbf{(TILING)} Let us state a few properties about the conjunction $\aformula_{\rm cov}^j \wedge \aformula_{\cH}^j \wedge \aformula_{\cV}^j$
assuming that $\atreemodel, \aroot \models {\rm grid}(2n)$ (and therefore the set of nodes of distance $2n$ from the
root encodes the grid $\interval{0}{2^n-1} \times \interval{0}{2^n-1}$). 
The formula $\aformula_{\rm cov}^j$ defined as 
\[
\forall \ \anominal \ \bindd{\anominal}{2n}
\rightarrow \att{\anominal}{2n} \ (\bigvee_{t \in \cT} t^j \wedge \bigwedge_{t \neq t' \in \cT} \neg (t^j \wedge t'^j)) \]
states that all the nodes at distance $2n$ from the root satisfy exactly one tile type from $\cT^j$. This is again a consequence
of Lemma~\ref{lemma:nominal} and Lemma~\ref{lemma:at}. 
The formula  $\aformula_{\cH}^j$ is defined as:
$$
\forall \ \anominal,  \anominalbis \ 
(
\bindd{\anominal}{2n} \wedge
\bindd{\anominalbis}{2n} \wedge {\rm HN}(\anominal,\anominalbis)
)
\rightarrow 
\bigvee_{\pair{t}{t'} \in \mathcal{H}} \att{\anominal}{2n} \ t^j \wedge \att{\anominalbis}{2n} \ t'^j,
$$
where ${\rm HN}(\anominal,\anominalbis)$ is the formula below:
\begin{multline*}
{\rm HN}(\anominal,\anominalbis) \egdef
\left(  
\bigwedge_{\alpha \in \interval{0}{n-1}} \att{\anominal}{2n} \avvarprop_{\alpha} \leftrightarrow \att{\anominalbis}{2n} 
\avvarprop_\alpha
\right)
\wedge 
\bigvee_{i \in \interval{0}{n-1}} \Bigg(
\att{\anominal}{2n} \neg \ahvarprop_i \wedge
\att{\anominalbis}{2n} \ahvarprop_i \; \wedge\\
\wedge \; \bigwedge_{\alpha \in \interval{0}{i-1}} \left( \att{\anominal}{2n} \ahvarprop_{\alpha} \wedge  \att{\anominalbis}{2n} \neg \ahvarprop_{\alpha} \right)
\wedge
(\bigwedge_{\alpha \in \interval{i+1}{n}} \left( \att{\anominal}{2n} \ahvarprop_{\alpha} \leftrightarrow  \att{\anominalbis}{2n} \ahvarprop_{\alpha}) \right)
\Bigg).
\end{multline*}
The formula ${\rm HN}(\anominal,\anominalbis)$ expresses that, assuming that $\anominal$ and $\anominalbis$ for the depth $2n$, 
the two nodes at distance $2n$ interpreted
respectively by $\anominal$ and $\anominalbis$ and representing respectively the positions $(\mathfrak{H},\mathfrak{V})$ and 
$(\mathfrak{H}',\mathfrak{V}')$ of the grid, 
satisfies
$\mathfrak{V} = \mathfrak{V}'$ and $\mathfrak{H}' = \mathfrak{H}+1$.
So, $\aformula_{\cH}^j$ encodes the horizontal constraints for the set of tile types $\cT^j$. 
Similarly, $\aformula_{\cV}^j$ encodes the vertical constraints for the set of tile types $\cT^j$. 
Indeed, in ${\rm VN}(\anominal,\anominalbis)$ we swap the variable $\ahvarprop_{\alpha}$ with the variable $\avvarprop_{\alpha}$
(with respect to the definition for ${\rm HN}(\anominal,\anominalbis)$),
which amount to stating that assuming that $\anominal$ and $\anominalbis$ are nominals for the depth $2n$ from $\aroot$, 
the two nodes at distance $2n$ interpreted
respectively by $\anominal$ and $\anominalbis$ and representing respectively the positions $(\mathfrak{H},\mathfrak{V})$ and $(\mathfrak{H}',
\mathfrak{V}')$ of the grid, satisfies
$\mathfrak{V}' = \mathfrak{V}+1$ and $\mathfrak{H}' = \mathfrak{H}$.
So, assuming that $\atreemodel, \aroot \models {\rm grid}(2n)$, the formula
$\aformula_{\rm cov}^j \wedge \aformula_{\cH}^j \wedge \aformula_{\cV}^j$ expresses that 
the way the tile types from $\cT^j$ holds on the nodes at distance $2n$ from the root defines
a proper tiling.

\textbf{(INIT)} Let $\aformula_{\rm init}^j$ be the formula below:
$$
\forall \ \anominal \ 
(\bindd{\anominal}{2n} \wedge \att{\anominal}{2n}(\bigwedge_{\alpha \in \interval{0}{n-1}} \neg \ahvarprop_{\alpha}))
\rightarrow
\att{\anominal}{2n} \ (\bigvee_{t \in \cTzero} t^j \wedge \bigwedge_{t \neq t' \in \cTzero} \neg (t^j \wedge t'^j)). 
$$
Note that $(\bigwedge_{\alpha \in \interval{0}{n-1}} \neg \ahvarprop_{\alpha})$ states on a node $\anode$ at distance $2n$ from the root,
that all the propositional variables $\ahvarprop_{\alpha}$ are false and therefore this is a node on the row zero of the grid. 
$\aformula_{\rm init}^j$ therefore states (using again Lemma~\ref{lemma:nominal} and Lemma~\ref{lemma:at})
that all the nodes of the row zero have a unique tile type from $\cTzero$.

\textbf{(COINCI)} Let $\aformula_{\rm coinci}^{j,j'}$ be the formula below (with $j,j' \in \Nat$):
$$
\forall \ \anominal \ 
(\bindd{\anominal}{2n} \wedge \att{\anominal}{2n}(\bigwedge_{\alpha \in \interval{0}{n-1}} \neg \ahvarprop_{\alpha}))
\rightarrow 
\att{\anominal}{2n} \ (\bigvee_{t \in \cTzero} t^j \wedge t^{j'}). 
$$
Assuming that $\atreemodel, \aroot \models \aformula_{\rm tiling}^j \wedge \aformula_{\rm init}^{j'}$, 
the formula $\aformula_{\rm coinci}^{j,j'}$ states that for every node of the row zero, the tile type from $\cT^j$ is the same as
the tile type from $\cTzero^{j'}$. The reasoning is exactly the same as for \textbf{(INIT)}. 

\textbf{(ACCEPT)} Let $\aformula_{\rm acc}^j$ be the formula below
$$
\exists \ \anominal \ 
\bindd{\anominal}{2n} \wedge \att{\anominal}{2n}((\bigwedge_{\alpha \in \interval{0}{n-1}} \ahvarprop_{\alpha})
\wedge \bigvee_{t \in \cTacc} t^j). 
$$
Note that $(\bigwedge_{\alpha \in \interval{0}{n-1}} \ahvarprop_{\alpha})$ states on a node $\anode$ at distance $2n$ from the root,
that all the propositional variables $\ahvarprop_{\alpha}$ are true and therefore this is a node on the 
(last) row $2^n-1$ of the grid. 

Assuming that $\atreemodel, \aroot \models \aformula_{\rm tiling}^j$, the formula $\aformula_{\rm acc}^j$ therefore states
that there is a node encoding a position of the grid on the last row such that the tile type is in $\cTacc^j$. 

\textbf{(MULTI)} Finally, let $\aformula_{\rm multi}^j$ be the formula below:
$$
\forall \ \anominal \ 
\bindd{\anominal}{2n} \rightarrow
\att{\anominal}{2n}(\bigvee_{\pair{t}{t'} \in \cTmulti}  t^j  \wedge t'^{j+1}).
$$
As done in \textbf{(INIT)}, $\aformula_{\rm multi}^j$ (assuming that  $\atreemodel, \aroot \models \aformula_{\rm tiling}^j \wedge \aformula_{\rm tiling}^{j+1}$)
states that for all nodes at distance $2n$ from the root, the tile type from $\cT^j$ and the tile type from $\cT^{j+1}$ are in 
the relation $\cTmulti$. 
Consequently, 
\[
\atreemodel, \aroot \models {\rm grid}(2n) \wedge (\bigwedge_{j \in \interval{1}{n}} \aformula_{\rm init}^{j}) \wedge
(\bigwedge_{j \in \interval{n+1}{2n}} \aformula_{\rm tiling}^{j} \wedge \aformula_{\rm coinci}^{j,(j-n)})
\wedge 
(\bigwedge_{j \in \interval{n+1}{2n-1}} \aformula_{\rm multi}^{j})
\wedge
\aformula_{\rm acc}^{2n}
)
\]
if and only if 
the initial condition induced from the satisfaction of $(\bigwedge_{j \in \interval{1}{n}} \aformula_{\rm init}^{j})$
(see the definition of $\cPmulti$ in Section~\ref{sec:complexity}) of the form $ (w_1, \ldots, w_n) \in (\cTzero^{2^n})^{n}$,
and the multi-tiling $(\atiling_1, \ldots, \atiling_n)$ induced by the satisfaction of 
$(\bigwedge_{j \in \interval{n+1}{2n}} \aformula_{\rm tiling}^{j})$, entails  
that  $(\atiling_1, \ldots, \atiling_n)$ is a solution and satisfies the condition~\ref{amtp:m-init}, \ref{amtp:m-tiling}, \ref{amtp:m-multi} and \ref{amtp:m-accept}. 
Indeed, satisfying $\atreemodel, \aroot \models {\rm grid}(2n) \wedge (\bigwedge_{j \in \interval{n+1}{2n}} \aformula_{\rm tiling}^{j})$ defines 
a multi-tiling $(\atiling_1, \ldots, \atiling_n)$ (and reciprocally). Similarly, 
satisfying $\atreemodel, \aroot \models  {\rm grid}(2n) \wedge (\bigwedge_{j \in \interval{1}{n}} \aformula_{\rm init}^{j})$
defines an initial condition  $ (w_1, \ldots, w_n) \in (\cTzero^{2^n})^{n}$ (and reciprocally). 
The details are omitted but it does not pose any difficulty. 
So given an initial condition $c= (w_1, \ldots, w_n) \in (\cTzero^{2^n})^{n}$, we write
$\atreemodel_{c}$ to denote a tree model such that $\atreemodel_{c}, \aroot \models 
{\rm grid}(2n) \wedge (\bigwedge_{j \in \interval{1}{n}} \aformula_{\rm init}^{j})$  and 
on the grid induced by $\atreemodel_{c}$, we have precisely the initial condition $c$.
Similarly, given a multi-tiling $M= (\atiling_1, \ldots, \atiling_n)$, 
 we write
$\atreemodel_{M}$ to denote a tree model such that $\atreemodel_{M}, \aroot \models 
{\rm grid}(2n) \wedge (\bigwedge_{j \in \interval{n+1}{2n}} \aformula_{\rm tiling}^{j})$  and 
on the grid induced by $\atreemodel_{M}$, we have precisely the multi-tiling
$(\atiling_1, \ldots, \atiling_n)$. 
More generally, given $c$ and $M$, we write $\atreemodel_{c,M}$ to denote a tree model 
such that $\atreemodel_{c,M}, \aroot \models 
{\rm grid}(2n) \wedge (\bigwedge_{j \in \interval{1}{n}} \aformula_{\rm init}^{j}) \wedge (\bigwedge_{j \in \interval{n+1}{2n}} \aformula_{\rm tiling}^{j})$
and on the grid induced by $\atreemodel_{c,M}$, we have precisely the initial condition $c$ and the multi-tiling $M$.
Reciprocally, assuming a tree model $\atreemodel$ such that  $\atreemodel, \aroot \models 
{\rm grid}(2n) \wedge (\bigwedge_{j \in \interval{1}{n}} \aformula_{\rm init}^{j})
$, we write $c_{\atreemodel} = (w_1, \ldots, w_n)$ to denote the initial condition from the grid defined
by  $\atreemodel$. Similarly, 
assuming a tree model $\atreemodel$ such that  $\atreemodel, \aroot \models 
{\rm grid}(2n) \wedge  (\bigwedge_{j \in \interval{n+1}{2n}} \aformula_{\rm tiling}^{j}) 
$, we write $M_{\atreemodel} = (\atiling_1, \ldots, \atiling_n)$ to denote the multi-tiling from the grid defined
by  $\atreemodel$.

It remains to conclude by explaining how to handle the quantifications over the tuples $(w_1, \ldots, w_n) \in (\cTzero^{2^n})^{n}$.
The instance $\mathcal{I}$ is positive iff 
\begin{itemize}
\item 
(by definition) for all $w_1 \in \cTzero^{2^n}$, there is $w_2 \in \cTzero^{2^n}$ such that $\cdots$ 
for all $w_{n-1} \in \cTzero^{2^n}$, there is $w_{n} \in \cTzero^{2^n}$ such that
there is a solution  $(\atiling_1, \ldots, \atiling_n)$ 
for $(w_1, \ldots, w_n)$, iff
\item (by the above correspondences) there is a tree model $\atreemodel^0$ such that $\atreemodel^0, \aroot \models {\rm grid}(2n)$
such that for all $\atreemodel^1$ such that ($\atreemodel^1 \approx_{\PVAR \setminus \cTzero^1} \atreemodel^0$), $\atreemodel^1, \aroot
\models  \aformula_{\rm init}^{1}$ and $w_1^{\atreemodel^1} = w_1$ (where $w_1^{\atreemodel^1}$ is the new word induced by $\atreemodel^1$)
there is $\atreemodel^2$ such that ($\atreemodel^2 \approx_{\PVAR \setminus \cTzero^2} \atreemodel^1$), 
$\atreemodel^2, \aroot
\models  \aformula_{\rm init}^{2}$ and $w_2^{\atreemodel^2} = w_2$ \ \ldots \ 
for all $\atreemodel^{n-1}$ such that ($\atreemodel^{n-1} \approx_{\PVAR \setminus \cTzero^{n-1}} \atreemodel^{n-2}$), 
$\atreemodel^{n-1}, \aroot
\models  \aformula_{\rm init}^{n-1}$ and $w_{n-1}^{\atreemodel^{n-1}} = w_{n-1}$
there is $\atreemodel^{n}$ such that ($\atreemodel^{n} \approx_{\PVAR \setminus \cTzero^n} \atreemodel^{n-1}$),  $\atreemodel^{n}, \aroot
\models  \aformula_{\rm init}^{n}$ and $w_{n}^{\atreemodel^{n}} = w_{n}$,
there is $\atreemodel^{n+1}$ such that  ($\atreemodel^{n+1} \approx_{\PVAR \setminus \set{t^j: t \in \cT, j \in \interval{n+1}{2n}}} \atreemodel^{n}$)
such that $\atreemodel^{n+1}, \aroot \models  (\bigwedge_{j \in \interval{n+1}{2n}} \aformula_{\rm tiling}^{j})$
and $M_{\atreemodel^{n+1}}$ is a solution for $c_{\atreemodel^{n+1}}$, iff
\item (by first-order reasoning) there is a tree model $\atreemodel^0$ such that $\atreemodel^0, \aroot \models {\rm grid}(2n)$
such that for all $\atreemodel^1$ such that ($\atreemodel^1 \approx_{\PVAR \setminus \cTzero^1} \atreemodel^0$)
there is $\atreemodel^2$ such that ($\atreemodel^2 \approx_{\PVAR \setminus \cTzero^2} \atreemodel^1$) \ \ldots \ 
for all $\atreemodel^{n-1}$ such that ($\atreemodel^{n-1} \approx_{\PVAR \setminus \cTzero^{n-1}} \atreemodel^{n-2}$)
there is $\atreemodel^{n}$ such that ($\atreemodel^{n} \approx_{\PVAR \setminus \cTzero^n} \atreemodel^{n-1}$),
if $\atreemodel^{n}, \aroot \models (\bigwedge_{j \in \interval{1}{n}} \aformula_{\rm init}^{j})$, then
there is $\atreemodel^{n+1}$ such that  ($\atreemodel^{n+1} \approx_{\PVAR \setminus \set{t^j: t \in \cT, j \in \interval{n+1}{2n}}} \atreemodel^{n}$)
such that $\atreemodel^{n+1}, \aroot \models  (\bigwedge_{j \in \interval{n+1}{2n}} \aformula_{\rm tiling}^{j})$
and $M_{\atreemodel^{n+1}}$ is a solution for $c_{\atreemodel^{n+1}}$, iff
\item (by the encodings of (m-init), (m-coinci), (m-multi), (m-accept))
there is a tree model $\atreemodel^0$ such that $\atreemodel^0, \aroot \models {\rm grid}(2n)$
such that for all $\atreemodel^1$ such that ($\atreemodel^1 \approx_{\PVAR \setminus \cTzero^1} \atreemodel^0$)
there is $\atreemodel^2$ such that ($\atreemodel^2 \approx_{\PVAR \setminus \cTzero^2} \atreemodel^1$) \ \ldots \ 
for all $\atreemodel^{n-1}$ such that ($\atreemodel^{n-1} \approx_{\PVAR \setminus \cTzero^{n-1}} \atreemodel^{n-2}$)
there is $\atreemodel^{n}$ such that ($\atreemodel^{n} \approx_{\PVAR \setminus \cTzero^n} \atreemodel^{n-1}$),
if $\atreemodel^{n}, \aroot \models (\bigwedge_{j \in \interval{1}{n}} \aformula_{\rm init}^{j})$, then
there is $\atreemodel^{n+1}$ such that  ($\atreemodel^{n+1} \approx_{\PVAR \setminus \set{t^j: t \in \cT, j \in \interval{n+1}{2n}}} \atreemodel^{n}$)
such that $\atreemodel^{n+1}, \aroot \models  (\bigwedge_{j \in \interval{n+1}{2n}} \aformula_{\rm tiling}^{j})$
and $\atreemodel^{n+1}, \aroot \models (\bigwedge_{j \in \interval{n+1}{2n}} \aformula_{\rm coinci}^{j,(j-n)}) \wedge (\bigwedge_{j \in \interval{n+1}{2n-1}} \aformula_{\rm multi}^{j}) \wedge \aformula_{\rm acc}^{2n}$ iff
\item (by definition of $\models$) the formula below is satisfiable:
\[
{\rm grid}(2n) \wedge 
\forall \ \cTzero^1 \ 
\exists \ \cTzero^2 \ 
\forall \ \cTzero^3 \  
\cdots 
\exists \ \cTzero^n
\bigwedge_{j \in \interval{1}{n}} \aformula_{\rm init}^{j}
\rightarrow 
\]
\[
(
\exists \ \set{t^j: t \in \cT, j \in \interval{n+1}{2n}} \ 
(\bigwedge_{j \in \interval{n+1}{2n}} \aformula_{\rm tiling}^{j} \wedge \aformula_{\rm coinci}^{j,(j-n)})
\wedge 
(\bigwedge_{j \in \interval{n+1}{2n-1}} \aformula_{\rm multi}^{j})
\wedge
\aformula_{\rm acc}^{2n}
). 
\]
\end{itemize}
\end{proof}

\section{Proofs  from Section~\ref{section-tower-hardness}}

\subsection{Proof of Lemma~\ref{lemma:kd-comparison-zero}}\label{appendix:proof-of:lemma:kd-comparison-zero}

\begin{proof}
The property (III) is a direct consequence of (II). Since the proof of (II) is similar to (I), we focus on showing (I) only.
In order to define $\succk{k}{\anominal_1, \ldots, \anominal_k}{\anominalbis_1, \ldots, \anominalbis_k}$, 
and therefore to express that $\semnumber{\atreemodel}{\anode_k'} = \semnumber{\atreemodel}{\anode_k}+1$ with
numbers computed with the propositional variables $\avarprop_{n-1}, \ldots, \avarprop_{0}$, 
we perform a standard comparison of the respective truth values of  $\avarprop_{n-1}, \ldots, \avarprop_{0}$ 
for the node $\anode_k$ (the interpretation of $\anominal_k$) and for the node $\anode'_k$ (the interpretation of $\anominalbis_k$). 
Typically, $\semnumber{\atreemodel}{\anode_k'} = \semnumber{\atreemodel}{\anode_k}+1$ iff there is $i \in \interval{0}{n-1}$, such that
\begin{itemize}
\item for every $j \in \interval{i+1}{n-1}$, $\anode_k$ and $\anode_k'$ agree on $\avarprop_j$
      ($\semnumber{\atreemodel}{\anode_k'}$ and $\semnumber{\atreemodel}{\anode_k}$ agree on their $j$th bit),
\item $\anode_k$ does not satisfy $\avarprop_i$ and  $\anode_k'$  satisfies $\avarprop_i$
      ($\semnumber{\atreemodel}{\anode_k'}$ and $\semnumber{\atreemodel}{\anode_k}$ disagree on their $i$th bit
       and the $i$th bit of $\semnumber{\atreemodel}{\anode_k}$ is equal to zero),
\item for every $j \in \interval{0}{i-1}$, $\anode_k$  satisfies $\avarprop_j$ and $\anode_k'$  does not satisfy $\avarprop_j$. 
\end{itemize}
The formula $\succk{k}{\anominal_1, \ldots, \anominal_k}{\anominalbis_1, \ldots, \anominalbis_k}$ indeed
quantifies existentially on $i$ via a generalised disjunction and the three conditions are checked by three conjuncts
in the standard manner.
\end{proof}

\subsection{Proof of Lemma~\ref{lemma:typeone}}\label{appendix:proof-of:lemma:typeone}

\begin{proof} 
We focus on proving (I), the other properties can be shown analogously.

A node $\anode$ is of type $1$ iff it has exactly $2^n$ children, say $\anode_0, \ldots, \anode_{2^n-1}$, and
for all $j \in \interval{0}{2^n-1}$, the number associated to $\anode_j$ is precisely $j$
when encoded with the  truth values of the propositional variables $\avarprop_{n-1}, \ldots, \avarprop_{0}$. 
The latter proposition amounts to having the following conditions:
\begin{itemize}[left=3mm]
\item[(a)] The node $\anode$ has a child whose number is zero.
\item[(b)] If $\anode'$ is a child of $\anode$ with number $m < 2^n-1$, then
              $\anode$ has also a child with number $m+1$. 
\item[(c)] Two distinct children of $\anode$ have distinct numbers. 
\end{itemize}
Let us recall that $\phitype{1} = \AX(\phitype{0}) \wedge \EX(\phifirst{0}) \wedge \phiunique{1} \wedge \phipopulate{1}$. 
Obviously, $\AX(\phitype{0})$ always holds and $\EX(\phifirst{0})$ expresses exactly the condition (a).
It remains to show that $\phiunique{1}$ (resp. $\phipopulate{1}$) characterises the conditions (c) (resp. the condition (b)).
The formula $\phiunique{1}$ is equal to $\forall{\anominal,\anominalbis} \ \distinctbindd{\anominal,\anominalbis}{1} \rightarrow \neg  (\eqk{1}{\anominal}{\anominalbis})
$.
By Lemma~\ref{lemma:kd-comparison-zero}(III), and by Lemma~\ref{lemma:distinct-nominals}, 
when interpreted on a node of type $1$, $\phiunique{1}$ states that for any two distinct children,
their respective numbers are different, which is precisely the condition (c). 
Finally, let us recall the definition of the formula $\phipopulate{1}$:
\[
\forall{\anominal} \ (\bindd{\anominal}{1} \wedge \att{\anominal}{1}(\neg \philast{0})) \rightarrow
\exists{\anominalbis} \ \bindd{\anominalbis}{1} \wedge \succk{1}{\anominal}{\anominalbis}.
\]
By Lemma~\ref{lemma:kd-comparison-zero}(I), by the fact that $\philast{0}$ already characterises the nodes of type $0$ whose number is~$2^n-1$.
By Lemma~\ref{lemma:nominal} and Lemma~\ref{lemma:at},  the formula $\phipopulate{1}$, states that for all children whose number~$m$ is different from $2^n-1$, there is a child with number $m+1$, 
which is precisely the condition (b). 
This ends the proof as all the formulae $\AX(\phitype{0})$, $\EX(\phifirst{0})$, $\phiunique{1}$, $\phipopulate{1}$ capture exactly the properties specified above.  
\end{proof}

\subsection{Proof of Lemma~\ref{lemma:lsr-partition}}\label{appendix:proof-of:lemma:lsr-partition}

\begin{proof} Assume that $\anode_0$ is of type $k$, $\atreemodel, \anode_0 \models \widehat{\bar{\anominal},\bar{\anominal}}$
and the witness branch is $\anode_0, \ldots, \anode_d$.

First, suppose that $\atreemodel, \anode \models \lsr{k}{\bar{\anominal}}$. 
As $\atreemodel, \anode_0 \models \att{\bar{\anominal}}{}(\EXOne(\mathsf{s}))$, by Lemma~\ref{lemma:atoflengthd},
$\atreemodel, \anode_d \models \EXOne(\mathsf{s})$. It is easy to show that $\EXOne \ \aformulabis$ holds whenever
there is a unique child satisfying $\aformulabis$. Consequently, there is a unique child of $\anode_d$ satisfying 
$\mathsf{s}$, which corresponds to the satisfaction of (b). 
As $\atreemodel, \anode_0 \models \att{\bar{\anominal}}{} 
\left( \AX((\mathsf{s} \vee \mathsf{l} \vee \mathsf{r})  \wedge 
 \neg (\mathsf{s} \wedge \mathsf{l}) \wedge 
\neg (\mathsf{s} \wedge \mathsf{r}) \wedge \neg (\mathsf{l} \wedge \mathsf{r}))
\right)$, by Lemma~\ref{lemma:atoflengthd}, 
we have $\atreemodel, \anode_d \models \AX((\mathsf{s} \vee \mathsf{l} \vee \mathsf{r})  \wedge 
 \neg (\mathsf{s} \wedge \mathsf{l}) \wedge 
\neg (\mathsf{s} \wedge \mathsf{r}) \wedge \neg (\mathsf{l} \wedge \mathsf{r}))$ and therefore
for all children $\anode'$ of $\anode_d$, we have 
 $\atreemodel, \anode' \models (\mathsf{s} \vee \mathsf{l} \vee \mathsf{r})  \wedge 
 \neg (\mathsf{s} \wedge \mathsf{l}) \wedge 
\neg (\mathsf{s} \wedge \mathsf{r}) \wedge \neg (\mathsf{l} \wedge \mathsf{r})$.
As the formula $(\mathsf{s} \vee \mathsf{l} \vee \mathsf{r})  \wedge 
 \neg (\mathsf{s} \wedge \mathsf{l}) \wedge 
\neg (\mathsf{s} \wedge \mathsf{r}) \wedge \neg (\mathsf{l} \wedge \mathsf{r})$ precisely states that exactly one 
propositional
variable among $\set{\mathsf{l},\mathsf{s},\mathsf{r}}$ holds true, we can conclude that (a) is satisfied. 
Moreover, $\atreemodel, \anode_0 \models \lsrthree{k}{\bar{\anominal}}$ and by Lemma~\ref{lemma:atoflengthd},
we have~that 
\[
\atreemodel, \anode_d \models 
\forall{w} \forall{w'} \ \distinctbindd{w,w'}{1} \wedge ((\att{w}{1}(\mathsf{s}) \wedge 
\att{w'}{1}(\mathsf{r})) \vee (\att{w}{1}(\mathsf{l}) \wedge \att{w'}{1}(\mathsf{s}))) \rightarrow 
\gk{k-d}{w'}{w}. 
\]
As $k-d-1 \leq N-2$ by assumption, the satisfaction of the formula above on $\anode_d$ straightforwardly states
that for all children $\anode$ and $\anode'$ of $\anode_d$ such that
($\anode$ satisfies $\mathsf{s}$ and $\anode'$ satisfies~$\mathsf{r}$)
or
($\anode$ satisfies $\mathsf{l}$ and $\anode'$ satisfies $\mathsf{s}$),
we have $\semnumber{\atreemodel}{\anode} < \semnumber{\atreemodel}{\anode'}$ (here we use the induction hypothesis),
which corresponds precisely to the satisfaction of (c). \\
The proof in the other direction is quite similar as there are equivalences
between the formulae  $\lsrone{k}{\bar{\anominal}}$, $\lsrtwo{k}{\bar{\anominal}}$, and $\lsrthree{k}{\bar{\anominal}}$,
and the conditions (a), (b) and (c). 
\end{proof}

\subsection{Proof of Lemma~\ref{lemma:kd-comparison}}\label{appendix:proof-of:lemma:kd-comparison}

\begin{proof} Let $\atreemodel$ be a tree model and $\anode$ be such that
$\anode$ satisfies $\widehat{\bar{\anominal},\bar{\anominalbis}}$, $k-d = N-1$  and 
$\anode_{d}$ and $\anode'_{d}$ are of type $k-d$. 
We will focus on the proof of (I) only. The proof of (II) is similar to the proof of (I) and (III) is a direct consequence of (II). Hence, we omit the details.

As $\anode_d$ and $\anode_d'$ are of type $k-d$, both have $\tow(k-d,n)$ children.
Those children are ordered, let $\anodebis_0, \ldots, \anodebis_{\tow(k-d,n)-1}$ be the children of $\anode_k$ such that
$\semnumber{}{\anodebis_j} = j$ for all $j$. Similarly, 
let $\anodebis_0', \ldots, \anodebis_{\tow(k-d,n)-1}'$ be the children of $\anode_k'$ such that
$\semnumber{}{\anodebis_j'} = j$ for all $j$. By arithmetical reasoning, we have
$\semnumber{\atreemodel}{\anode_d'} = \semnumber{\atreemodel}{\anode_d} + 1$ iff 
there is $i \in \interval{0}{\tow(k-d,n)-1}$, such that
\begin{itemize}[left=5mm]
\item[(A)] for every $j \in \interval{i+1}{\tow(k-d,n)-1}$, 
$\semnumber{\atreemodel}{\anode_d'}$ and $\semnumber{\atreemodel}{\anode_d}$ agree on the
$j$th bit, which is equivalent to $\anodebis_j$ and $\anodebis_j'$ agree on $\val{}$,
\item[(B)] the $i$th bit of $\semnumber{\atreemodel}{\anode_d}$ is equal to $0$ and 
the $i$th bit of $\semnumber{\atreemodel}{\anode_d'}$ is equal to $1$, which is equivalent to
$\anodebis_i$ does not satisfy $\val{}$ and  $\anodebis_i'$  satisfies $\val{}$,
\item[(C)] for every $j \in \interval{0}{i-1}$, 
the $j$th bit of $\semnumber{\atreemodel}{\anode_d}$ is equal to $1$
and the $j$th bit of $\semnumber{\atreemodel}{\anode_d'}$ is equal to $0$,
which is equivalent to  $\anodebis_j$  satisfies $\val{}$ and $\anodebis_j'$  does not satisfy $\val{}$. 
\end{itemize}
By using Lemma~\ref{lemma:atoflengthd}, it is easy to check 
that the condition (A) (resp. (B), (C)) is taken care by $\aformula_{{\tiny \rm left}}(k)$ (resp. $\aformula_{{\tiny \rm select}}(k)$,
$\aformula_{{\tiny \rm right}}(k)$).
This is quite immediate for $\aformula_{{\tiny \rm select}}(k)$ and $\aformula_{{\tiny \rm right}}(k)$, as no induction
hypothesis is used, in particular no comparison between numbers is performed. Concerning the satisfaction of 
$\aformula_{{\tiny \rm left}}(k)$, we have to be more careful. Let us recall the definition of the formula
$\aformula_{{\tiny \rm left}}^{\bar{\anominal}, \bar{\anominalbis}}(k)$ below:
\[
\forall{w} \; \att{\bar{\anominal}}{}(\bindd{w}{1} \wedge \att{w}{1} (\mathsf{l})) \rightarrow
\]
\[ 
\big(\exists{w'} \; 
\att{\bar{\anominalbis}}{}(\bindd{w'}{1} \wedge \att{w'}{1} (\mathsf{l})) \wedge 
\eqk{k}{\bar{\anominal},w}{\bar{\anominalbis},w'} \wedge (\att{\bar{\anominal},w}{} \val{k-d-1} \Iff 
\att{\bar{\anominalbis},w'}{} \val{k-d-1})\big).
\]
We have defined $\aformula_{{\tiny \rm left}}(k)$ as $\aformula_{{\tiny \rm left}}^{\bar{\anominal}, \bar{\anominalbis}}(k)
\wedge \aformula_{{\tiny \rm left}}^{\bar{\anominalbis}, \bar{\anominal}}(k)$. 
By Lemma~\ref{lemma:atoflengthd}, and by the induction hypothesis (as $(k - (d+1))  \leq N-2$), 
$\aformula_{{\tiny \rm left}}^{\bar{\anominal}, \bar{\anominalbis}}(k)$ enforces that for all 
children $\anode$ of $\anode_d$ satisfying $\mathsf{l}$, there is a
child $\anode'$ of $\anode'_d$ satisfying $\mathsf{l}$, such that $\anode$ and $\anode'$ agree on $\val{}$.
Moreover, $\aformula_{{\tiny \rm left}}^{\bar{\anominalbis}, \bar{\anominal}}(k)$ enforces that for all 
children $\anode$ of $\anode_d'$ satisfying $\mathsf{l}$, there is a
child $\anode'$ of $\anode_d$ satisfying $\mathsf{l}$, such that $\anode$ and $\anode'$ agree on $\val{}$.
So, the set of numbers of the children of $\anode_d$ satisfying $\mathsf{l}$ is equal to the set of 
numbers of the children of $\anode_d'$ satisfying $\mathsf{l}$. This implies also 
that this property applies  for the children satisfying $\mathsf{r}$, 
and the number of the unique child of $\anode_d$  satisfying $\mathsf{s}$ is equal
to the number of the unique child of $\anode_d'$. This is a direct consequence of the properties
of $\mathsf{l} \mathsf{s} \mathsf{r}$-partitions (see Lemma~\ref{lemma:lsr-partition}). 
Consequently, 
$\semnumber{\atreemodel}{\anode_d'} = \semnumber{\atreemodel}{\anode_d} + 1$ implies 
$\succk{k}{\bar{\anominal}}{\bar{\anominalbis}}$.
The proof for the other direction is similar, as $\aformula_{{\tiny \rm left}}(k)$ is equivalent to (A),
$\aformula_{{\tiny \rm select}}(k)$ is equivalent to (B) and $\aformula_{{\tiny \rm right}}(k)$ is equivalent to (C).
Moreover, the existential quantification over $i$ corresponds to the existential quantification leading to
an $\mathsf{l} \mathsf{s} \mathsf{r}$-partition. \qedhere
\end{proof}

\subsection{Proof of Lemma~\ref{lemma:typeN}}\label{appendix:proof-of:lemma:typeN}
\begin{proof} 
The properties (II) and (III) are  easy to verify, hence let us focus on (I).
The proof is actually very similar to the proof of Lemma~\ref{lemma:typeone}.
A node $\anode$ is of type $N$ for some $N \geq 2$, iff the conditions below hold:
(a) it has exactly $\tow(N,n)$ children, (b) all its children are of type $N-1$ and (c) $\set{\semnumber{\atreemodel}{\anode'}: \anode E \anode'} = \interval{0}{\tow(N,n)-1}$.  
The mentioned conditions can be reformulated as follows so, by induction, (a$'$)+(b$'$)+(c$'$) is equivalent to (a)+(b)+(c):
\begin{itemize}[left=4mm]
\item[(a$'$)] The node $\anode$ has a child whose number is zero.
\item[(b$'$)] If $\anode'$ is a child of $\anode$ with number $m < \tow(N,n)-1$, then
              $\anode$ has also a child with number $m+1$. 
\item[(c$'$)] Two distinct children of $\anode$ have distinct numbers. 
\end{itemize}
Let us recall below the definition of the formula $\phitype{N}$:
\[
\AX(\phitype{N-1}) \wedge
\EX(\phifirst{N-1}) \wedge \phiunique{N} \wedge \phipopulate{N}.
\]
Obviously, the formula $\AX(\phitype{N-1})$ expresses exactly the condition (b), assuming that $\phitype{N-1}$ already characterised
the nodes of type $N-1$. Similarly, the formula $\EX(\phifirst{N-1})$ expresses exactly the condition (a$'$), assuming
that $\phifirst{N-1}$ already characterises the nodes of type $N-1$ whose number is zero.
It remains to show that $\phiunique{N}$ (resp. $\phipopulate{N}$) characterises the conditions (c$'$) (resp. the condition (b$'$)).

Let us recall below the definition of the formula $\phiunique{N}$:
\[
\forall{\anominal,\anominalbis} \ \distinctbindd{\anominal,\anominalbis}{1} \rightarrow \neg  (\eqk{N}{\anominal}{\anominalbis}).
\]
By Lemma~\ref{lemma:kd-comparison}(III), and by the fact that $\distinctbindd{\anominal,\anominalbis}{1}$ 
enforces that $\anominal$ and $\anominalbis$ are interpreted by two distinct children (see Lemma~\ref{lemma:distinct-nominals}), 
when interpreted on a node of type $N$, the formula $\phiunique{N}$ states that for any two distinct children,
their respective numbers are different, which is precisely the condition (c$'$). 
Finally, let us recall the formula $\phipopulate{N}$:
\[
\forall{\anominal} \ (\bindd{\anominal}{1} \wedge \att{\anominal}{1}(\neg \philast{N-1})) \rightarrow
\exists{\anominalbis} \ \bindd{\anominalbis}{1} \wedge \succk{N}{\anominal}{\anominalbis}.
\]
By Lemma~\ref{lemma:kd-comparison}(I), by the fact that 
$\philast{N-1}$ already characterises the nodes of type $N-1$ whose number is $\tow(N,n)-1$
and 
by the properties of the subformulae $\bindd{\anominal}{1}$ and $\bindd{\anominalbis}{1}$ 
enforcing local nominals $\anominal$ and $\anominalbis$, the formula $\phipopulate{N}$,
when interpreted on a node of type $N$, states that for all children 
whose number $m$ is different from $\tow(N,n)-1$, there is a child with number $m+1$, 
which is precisely the condition (b$'$). 
This ends the proof as the formulae 
$\AX(\phitype{N-1})$,
$\EX(\phifirst{N-1})$, $\phiunique{N}$, $\phipopulate{N}$ 
capture~(a$'$)+(b$'$)+(c$'$). 
\end{proof}

\subsection{Proof of Lemma~\ref{lemma:nb-eq-towkn}}\label{appendix:proof-of:lemma:nb-eq-towkn}

\begin{proof} 
The case of $k=0$ is trival, so take for the case case $k=1$.
Assuming that $\atreemodel, \anode \models \phitype{1}$, the node $\anode$ is of type $1$, and therefore $\anode$ has exactly $\tow(1,n) = 2^n$ children. 
The number $\semnumber{\atreemodel}{\anode}$ is determined by the truth values of $\val{}$ on its children, and the children
have themselves (bit) numbers spanning all over $\interval{0}{2^n-1}$ and the numbers are encoded by the truth values of $\avarprop_{n-1}, \ldots, \avarprop_{0}$. 
Consequently, $\semnumber{\atreemodel}{\anode} = \tow(1,n) = 2^n$ iff the unique child $\anode'$ of $\anode$ such that 
$\semnumber{\atreemodel}{\anode'} = n$ satisfies $\val{}$ and all the other children do not satisfy $\val{}$. 
If the value $n$ encoded with $n$ bits
is represented by the sequence of bits $b_{n-1} b_{n-2} \cdots b_{0}$, then $\aset_n = n$ is a shortcut
for $\bigwedge_{i \in \interval{0}{n-1}} \ell_i$ where $\ell_i = \avarprop_i$ if $b_i = 1$, otherwise  $\ell_i = \neg \avarprop_i$.
Hence, we get $\semnumber{\atreemodel}{\anode} = \tow(k,n)$ iff 
$\atreemodel, \anode \models \AX (\val{0} \leftrightarrow \aset_n = n)$.

For the induction step, we reason in a similar way. 
Assume that for all $1 \leq k' < k$, if $\atreemodel', \anode' \models \phitype{k'}$, we have
$\semnumber{\atreemodel'}{\anode'} = \tow(k',n)$ iff $\atreemodel, \anode \models \anothereqk{k'}$.  
Now, assume that $\atreemodel, \anode \models \phitype{k}$,
the node $\anode$ is of type $k$, and therefore $\anode$ has exactly $\tow(k,n)$ children. The number
$\semnumber{\atreemodel}{\anode}$ is determined by the truth values of $\val{}$ on its children, and the children
have themselves (bit) numbers spanning all over $\interval{0}{\tow(k,n)-1}$. 
Consequently, 
$\semnumber{\atreemodel}{\anode} = \tow(k,n)$ iff the unique child $\anode'$ of $\anode$ such that $\semnumber{\atreemodel}{\anode'} 
= \tow(k-1,n)$ satisfies $\val{}$ and all the other children does not satisfy $\val{}$. Indeed, if only 
the $\tow(k-1,n)$th bit is equal to 1, $\semnumber{\atreemodel}{\anode}$ is equal to $2^{\tow(k-1,n)}$, which is precisely
$\tow(k,n)$. Now, checking whether a node of type $k-1$, has value $\tow(k-1,n)$ can be expressed by $\anothereqk{k-1}$
invoking the induction hypothesis. 
Putting all together, we get $\semnumber{\atreemodel}{\anode} = \tow(k,n)$ iff 
$\atreemodel, \anode \models \AX (\val{k-1} \leftrightarrow (\anothereqk{k-1}))$.
\end{proof}

\subsection{More details on the proof of Lemma~\ref{lemma:correctness-tiling-to-ex}}\label{appendix:lemma:correctness-tiling-to-ex}

\begin{proof} 
We have seen that the satisfaction of $\phitype{k+1}$ is guaranteed by the way propositional variables
hold on the nodes. Observe that for the nodes in 
$\interval{0}{\tow(k+1,n)-1} \times \cdots \times \interval{0}{\tow(1,n)-1} \times 0^+$, 
the truth values of the propositional variables is irrelevant. 
Moreover, for all $j \in \interval{1}{k+1}$, 
 for all $m_{k+1}, \ldots, m_{j} \in
      \interval{0}{\tow(k+1,n)-1} \times \cdots \times \interval{0}{\tow(j,n)-1}$,
we have $\semnumber{\atreemodel}{m_{k+1}, \ldots, m_{j}} = m_{j}$. 
In order to check the satisfaction of the existentially quantified subformula, we consider the labelling $l'$
variant of $l$ only for the propositional variables $\mathsf{l}$, $\mathsf{s}$, and $\mathsf{r}$  such that
$\mathsf{l}$ holds on $\interval{\tow(k,n)+1}{\tow(k+1,n)-1}$,
$\mathsf{s}$ holds on $\tow(k,n)$, and 
$\mathsf{r}$ holds on $\interval{0}{\tow(k,n) -1}$. 
By Lemma~\ref{lemma:lsr-partition}, 
$\triple{V}{E}{l'}, \aroot \models \lsr{k+1}{\varepsilon}$ and
by Lemma~\ref{lemma:nb-eq-towkn}, we have $\triple{V}{E}{l'}, \aroot \models \EX(\mathsf{s}\wedge \anothereqk{k})$. 
The satisfaction of the formulae $\aformula_{\rm cov}$, $\aformula_{\cH}$ and $\aformula_{\cV}$ is inherited 
from the fact that for all $\pair{i}{j} \in \interval{0}{\tow(k,n)-1} \times \interval{0}{\tow(k,n)-1}$, 
     there is exactly one tile type satisfied by $\pair{i}{j}$ and
the mapping $\tau$ satisfies horizontal and vertical matching conditions. 
Here, we use the properties of the formulae of the form $\bindd{\anominal}{1}$, 
$\att{\anominal}{1} \aformulabis$, $\eqk{k+1}{\anominal}{\anominal'}$  and 
$\succk{k+1}{\anominal,\anominalbis}{\anominal,\anominalbis'}$ (see e.g. Lemma~\ref{lemma:kd-comparison}),
apart from the fact that $\mathsf{r}$ holds exactly on the nodes in $\interval{0}{\tow(k,n) -1}$.
Concerning the satisfaction of $\aformula_{\rm init}$, first observe that $\EX_{=j} \aformulabis$ holds true
exactly when there are $j$ children of the node satisfying the formula $\aformulabis$. Hence, the formula
below
\[
\forall \ \anominal \  ( \bindd{\anominal}{1} \wedge \att{\anominal}{1} (\phifirst{k}))
\rightarrow \att{\anominal}{1} (
\bigwedge_{j \in \interval{0}{n-1}} \exists \ \mathsf{l}, \mathsf{s}, \mathsf{r} \ 
\lsr{k}{\varepsilon} \wedge \EX_{=j} \ \mathsf{r} \wedge \EX(\mathsf{s} \wedge t_j)
)
\]
states for all $\pair{i}{j} \in \set{0} \times \interval{0}{n-1}$, $t_{j}$ holds on it.
In order to access the $j$th child of $\set{0}$, an $\mathsf{l}\mathsf{s}\mathsf{r}$-partition
on the children of $\set{0}$ is performed and $\mathsf{s}$ holds true exactly on  $\pair{i}{j}$ by counting
how many children satisfies $\mathsf{r}$.
The proof for the other direction uses similar principles and is omitted herein. The main idea
is to build $\tau$ so that assuming that $\atreemodel, \aroot \models \aformula_{\cP}$,
for all $\anode, \anode'$ such that $\aroot \ E \ \anode \ E \ \anode'$, and 
$\pair{\semnumber{}{\anode}}{\semnumber{}{\anode'}} \in \interval{0}{\tow(k,n)-1} \times \interval{0}{\tow(k,n)-1}$, 
$\tau(\semnumber{}{\anode},\semnumber{}{\anode'})$ takes the value of the unique $t$ in $\cT$ satisfied on  
the node $\anode'$. 
\end{proof}

\section{Proofs  from Section~\ref{section-collection}}

\subsection{Proof of Lemma~\ref{lemma:k-layered}}\label{appendix:proof-of:lemma:k-layered}

\begin{proof} 
First, let us suppose that $\atreemodel$ is $k$-layered and we show that $\atreemodel, \aroot \models {\rm shape}(k)$. 
\begin{itemize}
\item By the condition (a), we have that: 
\[
\atreemodel, \aroot \models  \AG \left( 
    (\layer{-1} \vee \layer{0} \vee \cdots \vee \layer{k} ) \wedge 
    \bigwedge_{-1 \leq i\neq j \leq k}  \neg (\layer{i} \wedge \layer{j})
    \right)
\] 
as the formula inside $\AG$ states that there is exactly one proposition
from $\asetbis_k = \set{\layer{-1},\layer{0},\ldots,\layer{k}}$ holds.
\item By the condition (b), namely its second point, for all $i \in \interval{0}{k}$, we have
 \[
 \atreemodel, \aroot \models \bigwedge_{-1 \leq i \leq k} \AG(\layer{i} \rightarrow \AG (\layer{-1} \vee \layer{0} \vee \cdots \vee \layer{i})).
 \] 

\item By the first point of the condition (b), for all $i \in \interval{0}{k}$, we have 
$\atreemodel, \aroot \models \AG(\layer{i} \rightarrow \EF \ \layer{i-1})$.  Actually, when $\atreemodel, \anode \models \layer{i}$
the witness descendant satisfying $\layer{i-1}$ is a child of $\anode$ by (b).

\item By the condition (d), we have $\atreemodel, \aroot \models \layer{k}$.

\item By the condition (c), for every node $\anode$ satisfying $\layer{j}$ for some $j \in \interval{0}{k}$, there is
no proper descendant of $\anode$ satisfying $\layer{j}$. Observe that the formula 
$\neg \exists{\avarprop} \ (\avarprop  \wedge \EF(\layer{j} \wedge \neg \avarprop))$ holds exactly on the nodes
such that there is no  proper descendant satisfying $\layer{j}$. 
Hence \[ \atreemodel, \aroot \models \bigwedge_{0 \leq i \leq k}
\AG \left( \layer{i} \rightarrow \neg \exists{\avarprop} \ \left(\avarprop  \wedge \EF(\layer{i} \wedge \neg \avarprop) \right) \right).\] 
\end{itemize}

Conversely, suppose that $\atreemodel, \aroot \models {\rm shape}(k)$ holds.  
The satisfaction of (a), (c) and (d) holds thanks to the corresponding formulae in ${\rm shape}(k)$ (see above). 
Let us check that (c) holds true. As 
\[
\atreemodel, \aroot \models \bigwedge_{0 \leq i \leq k}
\AG(\layer{i} \rightarrow \neg \exists{\avarprop} \ (\avarprop  \wedge \EF(\layer{i} \wedge \neg \avarprop))) \]
holds and for all $i \in \interval{0}{k}$, we have
\[
\atreemodel, \aroot \models \AG(\layer{i} \rightarrow \EF \ \layer{i-1}), 
\]
on the same branch two distinct nodes cannot satisfy $\layer{i}$ for some $i \in \interval{0}{k}$. 
Moreover, the satisfaction of $\layer{i}$ implies that some (proper) descendant satisfies $\layer{i-1}$. 
Due to the monotonicity of the layer numbers and no stuttering,  
$\layer{i}$ implies that a child satisfies $\layer{i-1}$,
which corresponds to the first point of~(b). 
The second point of (b) is a consequence of the monotonicity of the layer numbers.
\end{proof}

\subsection{Proof of Lemma~\ref{lemma:correctness-tQCTLEF}}\label{appendix:proof:of:lemma:correctness-tQCTLEF}

\begin{proof}
First, let us assume that $\aformula$ ($\md{\aformula} = k \geq 0$) is satisfiable for \tQCTLEX{}, \ie $\atreemodel, \aroot \models \aformula$ holds with the tree model $\atreemodel = \triple{V}{E}{l}$.
Let $\atreemodel' = \triple{V}{E}{l'}$ be the tree model obtained from $\atreemodel$ by providing truth values for the propositional variables in $\asetbis_k$.
More precisely, for all $\anode \in V$ with $\aroot E^{j} \anode$,  we have that~$l'(\anode) \egdef (l(\anode) \setminus \asetbis_k) \cup \set{\layer{\max(-1,k-j)}}$. 
As $\atreemodel$ is a tree model, $\aroot E^j \anode$ implies that $j$ is the unique number of steps to reach $\anode$ from $\aroot$. 
Obviously,  $\atreemodel'$ is $k$-layered and hence, by Lemma~\ref{lemma:k-layered}, $\atreemodel', \aroot \models {\rm shape}(k)$ holds.
Moreover, by structural induction, one can show that 
for all $j \in \interval{0}{k}$, for all $\anode \in V$ with $\aroot E^{(k-j)} \anode$, 
and for all subformulae $\aformulabis$ of $\aformula$ of modal depth less than $j$,
$\atreemodel, \anode \models \aformulabis$ 
if and only if~$\atreemodel', \anode \models  \textit{trans}(j, \aformulabis)$. 
This leads to the satisfaction of   $\atreemodel', \aroot \models \textit{trans}(k,\aformula)$.
\begin{itemize}
\item For the base case $j = 0$, for all formulae $\aformulabis$ of modal degree 0, we have $\atreemodel, \anode \models \aformulabis$ iff  $\atreemodel', \anode \models \textit{trans}(0, \aformulabis)$ due to the fact that $\textit{trans}(0, \aformulabis) = \aformulabis$ and, that $\atreemodel$ and $\atreemodel'$ agree on the propositional variables occurring in $\aformula$.
\item For the induction step, the proof for the cases with Boolean connectives is immediate.
\item We now consider the case with propositional quantification. 
Suppose that $\atreemodel, \anode \models \exists \ \avarprop \ \aformulabis$.
Hence, there is $\atreemodel^{\star} = \triple{V}{E}{l^{\star}}$ such that $\atreemodel^{\star} \approx_{\PVAR \setminus \set{\avarprop}} \atreemodel$ and $\atreemodel^{\star}, \anode \models \aformulabis$. 
Let $\atreemodel^{\star \star} = \triple{V}{E}{l^{\star \star}}$ be the variant 
obtained from $\atreemodel^{\star}$ such that for all $\anode' \in V$ with $\aroot E^{j} \anode'$,  we have 
$l^{\star \star}(\anode') \egdef (l^{\star}(\anode') \setminus \asetbis_k) \cup \set{\layer{\max(-1,k-j)}}$. 
By the induction hypothesis, $\atreemodel^{\star \star}, \anode \models \textit{trans}(j,\aformulabis)$.
It is easy to check that $\atreemodel^{\star \star}  \approx_{\PVAR \setminus \set{\avarprop}} \atreemodel'$
and therefore  $\atreemodel', \anode \models \exists \ \avarprop \ \textit{trans}(j,\aformulabis)$.
The proof for the other direction is analogous.
\item Finally, we consider the case with $\EX$. 
First, let us suppose that $\atreemodel, \anode \models\EX \aformulabis$ with $\aroot E^{(k-j)} \anode$ and the modal depth of $\aformulabis$ is less than $j$.
Hence, there is $\anode'$ such that $\anode E \anode'$ and $\atreemodel, \anode' \models \aformulabis$.
Thus,~$\aroot E^{(k-(j-1))} \anode'$ and therefore $\anode'$ satisfies $\layer{j-1}$ in $\atreemodel'$. 
By the induction hypothesis ($\aformulabis$ is also of modal depth less than $j-1$), we conclude $\atreemodel', \anode' \models \layer{j-1} \wedge \textit{trans}(j-1,\aformulabis)$.
As $\anode'$ is also a child of $\anode$ in $\atreemodel'$ (and therefore a descendant), we obtain $\atreemodel', \anode \models \EF(\layer{j-1} \wedge \textit{trans}(j-1,\aformulabis))$.
Conversely, suppose that $\atreemodel', \anode \models \EF(\layer{j-1} \wedge \textit{trans}(j-1,\aformulabis))$.
Thus, there is a descendant $\anode'$ such that
$\anode E^{*} \anode'$ and $\atreemodel', \anode' \models \layer{j-1} \wedge \textit{trans}(j-1,\aformulabis)$.
By definition of $l'$, we have $\aroot E^{k-j+1} \anode'$  and therefore $\anode E \anode'$.
By the induction hypothesis, we obtain $\atreemodel, \anode' \models \aformulabis$ (again, 
$\aformulabis$ is also of modal depth less than $j-1$), which implies $\atreemodel, \anode \models \EX \aformulabis$.
\end{itemize}

For the other implication, we assume that $ \textit{trans}(k,\aformula) \wedge {\rm shape}(k)$ is satisfiable for \tQCTLEF{}, 
that is~$\atreemodel, \aroot \models  \textit{trans}(k,\aformula) \wedge {\rm shape}(k)$ holds
with the tree model $\atreemodel = \triple{V}{E}{l}$ and $\md{\aformula} = k \geq 0$. 
By Lemma~\ref{lemma:k-layered}, the tree model $\atreemodel$ is $k$-layered and therefore
satisfying $\layer{i}$ and jumping to a node with the help of $\EF (\layer{i-1} \wedge \ldots)$
leads to a child node (assuming that $i \in \interval{0}{k}$). 
Let $\atreemodel' = \triple{V'}{E'}{l'}$ be the tree model defined as follows:
\begin{itemize}
\item $V'$ is the least subset of $V$ satisfying the conditions below:
      \begin{itemize}
      \item $\aroot \in V'$,   
      \item if $\anode \in V'$ and $\layer{j} \in l(\anode)$ for some $j \in \interval{0}{k}$, 
      then for all $\anode' \in V$
      such that $\layer{j-1} \in l(\anode')$ and $\anode E \anode'$, then $\anode' \in V'$.
      The children of $\anode$ that do not satisfy $\layer{j-1}$ are ignored in $\atreemodel'$. 
      \end{itemize}
\item $l'$ is the restriction of $l$ to $V'$.
\item For all $\anode, \anode' \in V'$,  $\anode E' \anode'$ $\equivdef$ one the conditions below holds:
      \begin{itemize}
      \item $\layer{-1} \in l(\anode) \cap l(\anode')$ and $\anode E \anode'$.
      \item For some $j \in \interval{0}{k}$, $\anode E \anode'$, $\layer{j} \in l(\anode)$ and  $\layer{j-1} \in l(\anode')$.
      \end{itemize}
\end{itemize}
It is not difficult to check that $\atreemodel'$ is a tree model (finite-branching tree and all the maximal
branches are infinite), as $\atreemodel$ satisfies the formula below (due to the satisfaction of ${\rm shape}(k)$):
$$
\bigwedge_{i \in \interval{0}{k}} \ \AG(\layer{i} \rightarrow \EF \ \layer{i-1}).
$$

Similarly to what we did above, by structural induction, one can show that for
all $j \in \interval{0}{k}$, for all $\anode \in V'$ such that $\layer{j} \in l(\anode)$, 
and for all subformulae $\aformulabis$ of $\aformula$ of modal depth less than $j$,
we have $\atreemodel, \anode \models  \textit{trans}(j,\aformulabis)$ iff  
$\atreemodel', \anode \models \aformulabis$. 
This leads to the satisfaction of $\atreemodel', \aroot \models \aformula$
as $\atreemodel, \aroot \models \layer{k}$ holds by the satisfaction of ${\rm shape}(k)$. 
\begin{itemize}
\item For the base case $j = 0$, for all formulae $\aformulabis$ of modal degree 0, 
we have $\atreemodel, \anode \models \textit{trans}(0, \aformulabis)$  iff  $\atreemodel', \anode \models 
\aformulabis$ due to the fact that $\textit{trans}(0, \aformulabis) = \aformulabis$ and, 
$\atreemodel$ and $\atreemodel'$ agree on the propositional variables occurring in $\aformula$.
\item For the induction step, the proof for the cases with Boolean connectives and propositional quantification 
is by easy verification (\cf~the proof in the other direction). 
\item Let us treat in depth the case with $\EX \aformulabis$, with $\layer{j} \in l(\anode)$
and $\EX \aformulabis$ is of modal depth less than~$j$. Suppose that $\atreemodel, \anode \models
 \textit{trans}(j, \EX \aformulabis)$. So, this means that $\atreemodel, \anode \models 
\EF(\layer{j-1} \wedge \textit{trans}(j-1,\aformulabis))$. There is $\anode' \in V$ such that
$\anode E^{*} \anode'$ and  $\atreemodel, \anode' \models \layer{j-1} \wedge \textit{trans}(j-1,\aformulabis)$.
As $\atreemodel$ is a $k$-layered tree model, necessarily $\anode E \anode'$ (otherwise there are 
two distinct nodes on the branch from $\anode$ such that either both satisfy $\layer{j}$ or both satisfy $\layer{j-1}$,
which leads to a contradiction). By definition of $E'$, we get $\anode E' \anode'$ and by the induction
hypothesis ($\aformulabis$ is also of modal depth less than $j-1$), we get $\atreemodel', \anode' \models \aformulabis$.
Hence, we obtain $\atreemodel', \anode \models \EX \aformulabis$. 

Conversely, assume that $\atreemodel', \anode \models \EX \aformulabis$. Thus, there exists $\anode' \in V'$
such that $\anode E' \anode'$ and $\atreemodel', \anode' \models \aformulabis$.
By definition of $E'$, $\atreemodel, \anode' \models \layer{j-1}$ and hence, by the induction hypothesis,
we get $\atreemodel, \anode' \models  \layer{j-1} \wedge \textit{trans}(j-1,\aformulabis)$.
As one can check that $E' \subseteq E$ and hence, we conclude that~$\atreemodel, \anode \models  
\EF(\layer{j-1} \wedge \textit{trans}(j-1,\aformulabis))$.\qedhere
\end{itemize}
\end{proof}

\subsection{Proof of Lemma~\ref{theorem:K}}\label{appendix:proof-of:theorem:K}

\begin{proof}
As far as \tower-hardness is concerned, in order to enforce finite tree models, it is sufficient 
to consider the reduction defined for 
\tQCTLEX{} in Section~\ref{section-tower-hardness} but to modify the 
definition of the formula $\phitype{0}$ so that $\phitype{0}$ is now equal to $ \neg \EX{\top}$. 
In that way, the finite grids of the form $\interval{0}{\tow(k,n)-1} \times \interval{0}{\tow(k,n)-1}$ 
can still be encoded but with finite tree models.

In order to get the \tower upper bound, let us define a reduction to the satisfiability problem for \tQCTL{} by simply
identifying finite trees within tree models for \tQCTL{} (known to be in \tower by~\cite{LaroussinieM14}). 
Let $\aformula$ be a formula in \QKt. 
Without loss of generality, we assume that $\aformula$ may contain occurrences of $\EX$ and no occurrences of $\AX$. 
We introduce the formula $\textit{trans}(\aformula) \wedge \aformula_{\rm fin}$ in \tQCTL{}, 
where  $\aformula_{\rm fin}$ enforces that the fresh propositional variable 
$\textit{in}$ holds true only finitely on each branch
and  $\textit{trans}(\aformula)$ admits a recursive definition, 
by relativising the occurrences of $\EX$ with respect to $\textit{in}$.
Let $\aformula_{\rm fin}$  be the formula
$
\textit{in} \wedge \AF \ \neg \textit{in} \wedge \AG(\neg \textit{in} \rightarrow \AG \ \neg \textit{in})
$. The satisfiability of  $\aformula_{\rm fin}$ at the root node $\aroot$ implies that $\textit{in}$ holds
exactly on a subtree from $\aroot$ where all the branches are finite. 
It remains to define $\textit{trans}(\aformula)$:
\begin{itemize}
\itemsep 0 cm 
\item $\textit{trans}(\avarprop) \egdef \avarprop$ for all propositional variables $\avarprop$, and 
      $\textit{trans}$ is homomorphic for Boolean connectives and propositional quantification,
\item $\textit{trans}(\EX \aformulabis) \egdef \EX (\textit{in} \wedge \textit{trans}(\aformulabis))$.
\end{itemize}
We recall that the satisfaction of  $\aformula_{\rm fin}$ at the root node $\aroot$ implies $\textit{in}$ holds
exactly on a subtree from $\aroot$ where all the branches are finite.

One can show that $\aformula$ is satisfiable for \QKt iff $\textit{trans}(\aformula)
\wedge \aformula_{\rm fin}$ is satisfiable in \tQCTL{} (note that $\textit{in}$ does not need to be part
of the conjunction as the root is always part of the model).
Moreover, as the models for \tQCTL{} are finite-branching trees (see \eg~\cite[Remark 5.7]{LaroussinieM14}),
$\aformula$ is satisfiable in a finite tree model iff $\textit{trans}(\aformula)
\wedge \aformula_{\rm fin}$ is satisfiable in a finite-branching tree model, which leads to the desired 
upper bound \tower.
\end{proof}

\subsection{Proof of Lemma~\ref{theorem:GL}}\label{appendix:proof-of:theorem:GL}
\begin{proof}
First, let us show that \satproblem{\QGLt} and \satproblem{\ftQCTLEFplus{}}
are identical problems modulo the rewriting of $\EX$ into $\EX \EF$. 
Herein, $\ftQCTLEFplus{}$ is defined as \ftQCTL{}
restricted to the combined temporal operator $\EX \EF$. According to Section~\ref{section-preliminaries},
the models for \ftQCTL{} are finite trees.
In a second part of the proof, we show that \satproblem{\ftQCTLEFplus{}} is
\tower-complete. 

(I) Let $\atranslation$ be the map from \QGLt formulae into 
\ftQCTLEFplus{} formulae such that $\atranslation(\aformula)$ is defined from
$\aformula$ by replacing every occurrence of $\EX$ by $\EX\EF$. 
Similarly, the inverse map $\atranslation^{-1}$ is defined so that 
$\atranslation^{-1}(\aformulabis)$ with a \ftQCTLEFplus{} formula $\aformulabis$,
is defined from $\aformulabis$ by replacing every occurrence of $\EX\EF$  by $\EX$. 

Let us show that (a) $\aformula$ is \QGLt satisfiable iff $\atranslation(\aformula)$
is \ftQCTLEFplus{} satisfiable and (b) $\aformulabis$ is \ftQCTLEFplus{} satisfiable 
iff $\atranslation^{-1}(\aformulabis)$
is \QGLt  satisfiable. Since $\atranslation$ is bijective, it is sufficient to show (a).
Let $\kripkeK = \triple{W}{R}{l}$ be a \QGLt model and $w \in W$ such that
$\kripkeK, w \models \aformula$. As $\kripkeK$ is a \QGLt model,
$\pair{W}{R} = \pair{W}{E^+}$ for some finite tree $\pair{W}{E}$ with root $w$. 
Consequently for all $w_1, w_2 \in W$, we have 
$\pair{w_1}{w_2} \in R$ iff $\pair{w_1}{w_2}$ belongs to the transitive closure of
$E$. By structural induction, one can easily show that for all $w' \in W$ and for all subformulae
$\aformula'$ of $\aformula$, we have 
$\kripkeK, w' \models \aformula'$ iff $\triple{W}{E}{l}, w' \models \atranslation(\aformula')$. 
Consequently, $\aformula$ is \QGLt satisfiable implies 
$\atranslation(\aformula)$
is \ftQCTLEFplus{} satisfiable as $\triple{W}{E}{l}$ is a model for \ftQCTLEFplus{}
since it is a finite tree.

Conversely, let $\atreemodel = \triple{V}{E}{l}$ be a \ftQCTLEFplus{} model
(finite tree) with root $\varepsilon$ such that $\atreemodel, \varepsilon \models \atranslation(\aformula)$. 
Let $\kripkeK$ be the Kripke structure $\triple{V}{E^+}{l}$ defined from $\atreemodel$ and
by definition,  $\kripkeK$ is an  \QGLt model.
Again, by structural induction, one can easily show that for all $w' \in V$ and for all subformulae
$\aformula'$ of $\aformula$, we have 
$\triple{V}{E}{l}, w' \models \atranslation(\aformula')$
iff 
$\kripkeK, w' \models \aformula'$. 
Consequently, $\aformula$ is \QGLt satisfiable as we have assumed that 
$\atreemodel, \varepsilon \models \atranslation(\aformula)$ and therefore $\kripkeK, \varepsilon \models 
\aformula$. 

(II) Let us show that the satisfiability problem for \ftQCTLEFplus{} is \tower-complete. 
We need to take care of  both the lower bound and of the upper bound.
As far as \tower-hardness is concerned, we define a reduction from the satisfiability for \ftQCTLEX{}
(see Theorem~\ref{theorem:K} as \ftQCTLEX{} and \QKt are identical),
and the proof is very similar to the one for \tQCTLEF{} (actually, it is a bit simpler). The main steps are summarised below. 
Let $\aformula$ be a formula in  \ftQCTLEX{} of the modal depth $\md{\aformula} = k$. 
Similarly to what was done before, let us consider the set of fresh propositional variables $\asetbis_{k} = \set{\layer{0}, \ldots, \layer{k}}$
with the intended meaning that a node satisfying $\layer{i}$ is of ``layer $i$'', the root node being of layer $k$.
Let us define the formula $\textit{trans}(k,\aformula) \wedge {\rm shape}(k)$ in \tQCTLEFplus{}, 
where ${\rm shape}(k)$ is the conjunction of the following formulae:
\begin{itemize}
\itemsep 0 cm 
\item Every node satisfies exactly one propositional variable from $\asetbis_{k}$ (layer unicity) and the 
root satisfies $\layer{k}$.
     $$
     (\layer{k} \wedge \bigwedge_{i \neq k} \neg \layer{i}) \wedge
     \AX \AG ((\layer{0} \vee \cdots \vee \layer{k} ) \wedge \bigwedge_{0 \leq i\neq j \leq k} \neg (\layer{i} \wedge \layer{j})).
      $$
\item When a node satisfies $\layer{i}$ with $i \leq k$, none of its descendants satisfies some
 $\layer{j}$ with
$j > i$ (monotonicity of layer numbers).
$$
\AX \AG (\layer{0} \vee \cdots \vee \layer{k-1}) \wedge 
$$
$$
\bigwedge_{i \leq k-1} \AX \AG(\layer{i} \rightarrow \AX \AG (\layer{0} \vee \cdots \vee \layer{i-1})).
$$
\item When a node satisfies $\layer{i}$ with $1 \leq i \leq k-1$,  there is a descendant
satisfying $\layer{i-1}$ (weak progress).
$$
\EX \EF \ \layer{k-1} \wedge 
\AX \AG(\layer{i} \rightarrow \EX \EF \ \layer{i-1}).
$$
\item The nodes satisfying $\layer{0}$ have no successor: $\AX \AG(\layer{0} \rightarrow \neg \EX \EF \top)$. 
\end{itemize}
The formula  $\textit{trans}(k,\aformula)$ is defined as in the reduction from
\tQCTLEX{} to \tQCTLEFplus{} (see Section~\ref{section-qctlef}).  
One can show that $\aformula$ is satisfiable in \ftQCTLEX{} iff  $\textit{trans}(k,\aformula) \wedge {\rm shape}(k)$
is satisfiable in~\ftQCTLEFplus{}. 

To get the
\tower upper bound, let us define a reduction to the satisfiability problem for~\tQCTL{} (known
to be in \tower by~\cite{LaroussinieM14}). 
Let $\aformula$ be a formula in  \ftQCTLEFplus{}. 
We introduce the formula $\textit{trans}(\aformula) \wedge \aformula_{\rm fin}$ in \tQCTL{}, 
where $\textit{trans}(\aformula)$ is recursively defined as follows:
\begin{itemize}
\itemsep 0 cm 
\item $\textit{trans}(\avarprop) \egdef \avarprop$ for all propositional variables $\avarprop$, and 
      $\textit{trans}$ is homomorphic for Boolean connectives and propositional quantification,
\item $\textit{trans}(\EX \EF \aformulabis) \egdef \EX \EF (\textit{in} \wedge \textit{trans}(\aformulabis))$.
\end{itemize}
As in the proof of Theorem~\ref{theorem:K}, 
$\aformula_{\rm fin}$ is equal to 
$
\textit{in} \wedge \AF \ \AG \ \neg \textit{in} \wedge \AG(\neg \textit{in} \rightarrow \AG \ \neg \textit{in})
$.
One can show that $\aformula$ is satisfiable in a tree model without infinite branches iff $\textit{trans}(\aformula)
\wedge \aformula_{\rm fin}$ is satisfiable in \tQCTL{}.
Again, as the models for \tQCTL{} are finite-branching,
$\aformula$ is satisfiable in a finite tree model iff $\textit{trans}(\aformula)
\wedge \aformula_{\rm fin}$ is satisfiable in a finite-branching tree models, which leads to \tower-easiness. 
\end{proof}

\subsection{Proof of Lemma~\ref{theorem:Kfour}}\label{appendix:proof-of:theorem:Kfour}

\begin{proof} 

First, let us show that \satproblem{\QKfourt} and \satproblem{\gtQCTLEFplus{}}
are identical problems 
modulo the rewriting of $\EX$ into $\EX \EF$. 
Herein, $\gtQCTLEFplus{}$ is defined as \gtQCTL{}
restricted to the temporal operator $\EX \EF$. According to Section~\ref{section-preliminaries},
the models for \gtQCTL{} are finite-branching trees.
In a second part of the proof, we show that \satproblem{\gtQCTLEFplus{}} is
\tower-complete. 

(I) Let $\atranslation$ be the map from \QKfourt formulae into 
\gtQCTLEFplus{} formulae such that $\atranslation(\aformula)$ is defined from
$\aformula$ by replacing every occurrence of $\EX$ by $\EX\EF$. 
Similarly, the inverse map $\atranslation^{-1}$ is defined so that 
$\atranslation^{-1}(\aformulabis)$ with a \ftQCTLEFplus{} formula $\aformulabis$,
is defined from $\aformulabis$ by replacing every occurrence of $\EX\EF$  by $\EX$. 
This is similar to what is done in the proof of Theorem~\ref{theorem:GL}.

Let us show that (a) $\aformula$ is \QKfourt satisfiable iff $\atranslation(\aformula)$
is \gtQCTLEFplus{} satisfiable and (b) $\aformulabis$ is \gtQCTLEFplus{} satisfiable 
iff $\atranslation^{-1}(\aformulabis)$
is \QKfourt  satisfiable. It is sufficient to show (a).
Let $\kripkeK = \triple{W}{R}{l}$ be an \QKfourt model and $w \in W$ such that
$\kripkeK, w \models \aformula$. As $\kripkeK$ is a \QKfourt model
$\pair{W}{R} = \pair{W}{E^+}$ for some finite-branching tree $\pair{W}{E}$ with root $w$. 
Consequently for all $w_1, w_2 \in W$, we have 
$\pair{w_1}{w_2} \in R$ iff $\pair{w_1}{w_2}$ belongs to the transitive closure of
$E$. By structural induction, one can easily show that for all $w' \in W$ and for all subformulae
$\aformula'$ of $\aformula$, we have 
$\kripkeK, w' \models \aformula'$ iff $\triple{W}{E}{l}, w' \models \atranslation(\aformula')$. 
Consequently, $\aformula$ is \QKfourt satisfiable implies 
$\atranslation(\aformula)$
is \gtQCTLEFplus{} satisfiable as $\triple{W}{E}{l}$ is a model for \gtQCTLEFplus{}.

Conversely, let $\atreemodel = \triple{V}{E}{l}$ be a \gtQCTLEFplus{} model
(finite-branching tree) with root $\varepsilon$ such that $\atreemodel, \varepsilon \models \atranslation(\aformula)$. 
Let $\kripkeK$ be the Kripke structure $\triple{V}{E^+}{l}$ defined from $\atreemodel$ and
by definition,  $\kripkeK$ is an  \QKfourt model.
Again, by structural induction, one can easily show that 
  $\kripkeK, \varepsilon \models 
\aformula$.

(II) In order to establish the upper bound \tower, let us provide a reduction from the 
satisfiability for \gtQCTLEFplus{} to the satisfiability problem for \tQCTLEFplus{}. 
Let $\aformula$ be a formula in  \gtQCTLEFplus{}. 
We introduce the formula $\textit{trans}(\aformula) \wedge \aformula_{\rm fin}'$ in \tQCTLEFplus{}, 
where  $\aformula_{\rm fin}'$ enforces that the  propositional variable 
$\textit{in}$ holds false on all descendants, as soon as it does not hold on a node
and that $\textit{trans}(\aformula)$ admits a recursive definition, by relativising the occurrences of $\EX \EF $ with respect to $\textit{in}$.
Let $\aformula_{\rm fin}'$  be the formula
$
\textit{in} \wedge  \AX \AG(\neg \textit{in} \rightarrow \AX \AG \ \neg \textit{in})
$.

It remains to define $\textit{trans}(\aformula)$:
\begin{itemize}
\itemsep 0 cm 
\item $\textit{trans}(\avarprop) \egdef \avarprop$ for all propositional variables $\avarprop$, and 
      $\textit{trans}$ is homomorphic for Boolean connectives and propositional quantification,
\item $\textit{trans}(\EX \EF \aformulabis) \egdef \EX \EF (\textit{in} \wedge \textit{trans}(\aformulabis))$.
\end{itemize}
It is easy to see that $\aformula$ is satisfiable for \gtQCTLEFplus{} iff $\textit{trans}(\aformula)
\wedge \aformula_{\rm fin}'$ is satisfiable for \tQCTLEFplus{}.

As far as \tower-hardness is concerned, for any formula $\aformula$ in 
\tQCTLEFplus{}, one can show that $\aformula$ is satisfiable for \tQCTLEFplus{}
iff $\aformula \wedge \EX \ \EF \ \top \wedge \AX \AG \ \EX  \EF \ \top$ is satisfiable
in \gtQCTLEFplus{}. 
The two last conjuncts simply state that from any node, there is a child, which enforces that
all the maximal branches are infinite. 
As the satisfiability problem for \tQCTLEFplus{} is \tower-hard (Theorem~\ref{theorem:tQCTLEFplus}),
this concludes the proof. 
\end{proof}

\end{document}